\documentclass[a4paper,11pt,reqno,noindent]{amsart}
\usepackage[centertags]{amsmath}
\usepackage{amsfonts,amssymb,amsthm,amscd,dsfont,cases,esint,enumerate,color,mathdots}
\usepackage[T1]{fontenc}
\usepackage[english]{babel}
\usepackage{tikzsymbols}
\usepackage[applemac]{inputenc}
\usepackage[body={15cm,21.5cm},centering]{geometry}
\usepackage{fancyhdr}
\numberwithin{equation}{section}

\mathchardef\emptyset="001F

\theoremstyle{plain}
\newtheorem{lem0}{Lemma}[section]
\newenvironment{lem}
  {\pushQED{\qed}\begin{lem0}}
  {\popQED\end{lem0}}
\newtheorem{theor0}[lem0]{Theorem}
\newenvironment{theor}
  {\pushQED{\qed}\begin{theor0}}
  {\popQED\end{theor0}}
\newtheorem{cor0}[lem0]{Corollary}
\newenvironment{cor}
  {\pushQED{\qed}\begin{cor0}}
  {\popQED\end{cor0}}
  
\theoremstyle{definition}
\newtheorem{defin0}[lem0]{Definition}
\newenvironment{defin}
  {\pushQED{\qed}\begin{defin0}}
  {\popQED\end{defin0}}

\newcommand{\ac}{{\operatorname{ac}}}
\newcommand{\pp}{{\operatorname{pp}}}
\newcommand{\sg}{{\operatorname{sc}}}
\newcommand{\fu}{{\operatorname{f}}}
\newcommand{\pu}{{\operatorname{p}}}
\newcommand{\finu}{{\operatorname{fin}}}
\newcommand{\ad}{{\operatorname{ad}}}
\newcommand{\e}{\varepsilon}
\newcommand{\R}{\mathbb R}
\newcommand{\N}{\mathbb N}
\newcommand{\C}{\mathbb C}
\newcommand{\Cc}{\mathcal C}
\newcommand{\Hf}{\mathfrak h}
\newcommand{\Gg}{\mathfrak g}
\newcommand{\Dc}{\mathcal D}
\newcommand{\Hc}{\mathcal H}
\newcommand{\Sc}{\mathcal S}
\newcommand{\Ran}{{\operatorname{Ran}}}
\newcommand{\db}{{\operatorname{d}}}
\newcommand{\pv}{\operatorname{p.v.}}
\newcommand{\Id}{\operatorname{Id}}
\newcommand{\supp}{\operatorname{supp}}
\newcommand{\sgn}{\operatorname{sgn}}
\newcommand{\adh}{{\operatorname{adh\,}}}
\newcommand{\Ld}{\operatorname{L}}
\newcommand{\step}[1]{\noindent \textit{Step} #1.}

\def\dbar{\,{{\mathchar'26\mkern-11.5mu\db}}}
\def\deltabar{\,{{\mathchar'26\mkern-9.5mu\delta}}}

\usepackage[colorlinks,citecolor=black,urlcolor=black]{hyperref}

\title[Massive Cherenkov radiation and quantum friction]{Cherenkov radiation with massive bosons\\and quantum friction}
\author[M. Duerinckx]{Mitia Duerinckx}
\address[Mitia Duerinckx]{Universit{\'e} Libre de Bruxelles, Département de Math{\'e}matique, 1050~Brussels, Belgium}
\email{mitia.duerinckx@ulb.be}
\author[C. Shirley]{Christopher Shirley}
\address[Christopher Shirley]{Université Paris-Saclay, CNRS, Laboratoire de Math{\'e}matiques d'Orsay, 91400~Orsay, France}
\email{christopher.shirley@universite-paris-saclay.fr}

\begin{document}

\maketitle

\begin{abstract}
This work is devoted to several translation-invariant models in non-relativistic quantum field theory (QFT), describing a non-relativistic quantum particle interacting with a quantized relativistic field of bosons. In this setting, we aim at the rigorous study of Cherenkov radiation or friction effects at small disorder, which amounts to the metastability of the embedded mass shell of the free non-relativistic particle when the coupling to the quantized field is turned on.
Although this problem is naturally approached by means of Mourre's celebrated commutator method, important regularity issues are known to be inherent to QFT models and restrict the application of this method.
In this perspective, we introduce a novel non-standard construction procedure for Mourre conjugate operators, which differs from second quantization and allows to circumvent regularity issues. 
To show its versatility, we apply this construction to the Nelson model with massive bosons, to Fröhlich's polaron model, and to a quantum friction model with massless bosons introduced by Bruneau and De Bièvre: for each of these examples, we improve on previous results.

\bigskip\noindent
{\sc MSC-class:}
81T10;
81V73;
47B25;
81Q10;
47A55;
81Q15;
47A40;
81U24.

\end{abstract}

\setcounter{tocdepth}{2}
\tableofcontents

\section{Introduction and main results}

\subsection{General overview}
This work is devoted to several models in non-relativistic quantum field theory~(QFT),
describing a non-relativistic quantum particle interacting with a quantized relativistic field of bosons,
and we focus on translation-invariant models where the total momentum is conserved.
In this context, we aim at the rigorous study of Cherenkov radiation and friction effects: if the initial momentum~$|P|$ of the non-relativistic particle exceeds some threshold $|P_\star|$ (more precisely, if the initial energy of the particle exceeds the minimal energy for one-boson states),
the particle is expected to dissipate energy and slow down by emitting so-called Cherenkov radiation.
In terms of spectral theory, this dissipative phenomenon translates into the continuity of the energy-momentum spectrum and the absence of embedded mass shell in some region with~\mbox{$|P|>|P_\star|$}.

We shall focus for simplicity on the perturbative regime of weak particle-field coupling: the mass shell~$E=\frac12P^2$ of the free non-relativistic particle is then expected to be metastable for~$|P|>|P_\star|$ and to disappear as the coupling to the quantized field is turned on. We naturally expect to further complement this with a scattering resonance description.
As this problem basically concerns the perturbation of an embedded eigenvalue in continuous spectrum and as the relevant interaction Hamiltonian in QFT models is not relatively compact, even this perturbative analysis at weak coupling is a nontrivial problem for which still only partial results are available~\cite{AMZ05,DFP10,MR13,DM15,DFS-17}.
Note that a different line of research concerns the reduced dynamics of the non-relativistic particle in the kinetic limit: in~\cite{Spohn-77,E02,DRFP-10}, it is shown to take form of a Boltzmann equation describing the slowdown of the particle.
Although supporting the same thesis, such results are limited to diagonal time regimes and do not provide any detailed spectral information.

In recent decades, much attention has been devoted to spectral and scattering theory for QFT models, aiming to adapt the various techniques originally developed for the study of $N$-particle Schrödinger operators~\cite{Hunziker-Sigal-00}.
In particular, Mourre's commutator method~\cite{Mourre-80,ABMG-96,Hunziker-Sigal-00} has emerged as a fundamental tool to explore the nature of the essential spectrum of such Hamiltonians. It happens to be more general than related dilation-analyticity techniques and further provides direct insight into time-dependent scattering theory~\cite{Orth-90,Sofer-Weinstein-98,CGH-06}. Yet, various difficulties arise when applying this method to QFT models:
\begin{enumerate}[---]
\item In the case of massive bosons, the natural choice of Mourre's conjugate operator displays a lack of regularity in the sense that its commutator with the Hamiltonian is not relatively bounded with respect to the latter; see e.g.~Section~\ref{sec:conjug/Nelson}.
\smallskip\item In the case of massless bosons, the lack of regularity is even worse in the sense that the commutator with the Hamiltonian is not even comparable to the latter; see e.g.~Section~\ref{sec:constr-fric-commut} below. In addition, the natural conjugate displays a singularity at small wavenumbers, which destroys its self-adjointness. These issues are partly overcome in~\cite{Skibsted-98,Moller-Skibsted-04,GGM04a,GGM04,Faupin-Moller-Skibsted-11}, leading to the development of so-called ``singular'' Mourre theory.
\end{enumerate}
In either case, the full power of ``regular'' Mourre theory cannot be used, and in particular it gives no access to time-dependent scattering theory.
In the present contribution, we propose a new construction that aims to cure regularity issues and bring us back to the realm of regular Mourre theory.
As inspired by our previous work~\cite{DS21}, the crucial point is that we modify the construction of natural conjugates in a form that is no longer that of second quantization. To show the wide applicability of this modification procedure, we illustrate our results on two paradigmatic systems:
\begin{enumerate}[---]
\item For the translation-invariant Nelson model with massive bosons, previous results were restricted to the energy-momentum spectrum below the two-boson threshold and were limited by the lack of regularity~\cite{AMZ05,M05,Gerard-Moller-Rasmussen-11,MR13,DM15}. First, suitably expressing relative boson momenta in the frame that minimizes the kinetic energy, we define new natural conjugate operators that allow to study a semi-axis of spectrum containing the uncoupled mass shell beyond the two-boson threshold. Next, appealing to our modification procedure, we manage to avoid any regularity issue and to apply regular Mourre theory.
We emphasize that this procedure is quite general and may be of independent interest for massive QFT models. The same analysis can be repeated for instance for Fröhlich's polaron model~\cite{M06}.
\smallskip\item For the quantum friction model with massless bosons introduced in~\cite{Bruneau-07,DFS-17}, previous results were also limited by the lack of regularity~\cite{DFS-17}.
Appealing to our modification procedure, we cure again all regularity issues and reduce to the application of regular Mourre theory.
We do not know for now whether this approach could also be adapted to the massless Nelson model.
\end{enumerate}
In both cases, the application of regular Mourre theory brings a better understanding of the Cherenkov radiation phenomenon for these QFT models,
indeed allowing us to derive scattering resonance descriptions for the first time.
In the next two subsections, we introduce these models in full detail and formulate our main results, the proofs of which are postponed to Sections~\ref{sec:friction-pr} and~\ref{sec:Nelson}.

\subsection{Translation-invariant massive Nelson model}
The Nelson model was introduced in~\cite{N64} as a toy model in QFT for a free non-relativistic quantum particle interacting linearly with a quantized radiation field of relativistic scalar bosons; see e.g.~\cite{M06} and references therein.
While very complete results are available in the confined setting~\cite{DG99,BFSS99,Gerard-00,DJ01,GGM04}, both for massive and massless bosons, the understanding remains quite limited in the translation-invariant setting that we consider here.
For massive bosons, a detailed description of the bottom of the energy-momentum spectrum is obtained in~\cite{Spohn-88,M05,M06}, but the structure of the essential spectrum is only understood below the two-boson threshold, both at weak~\cite{Minlos-92,AMZ05} and large coupling~\cite{MR13,DM15}.
For massless bosons, the bottom of the spectrum is studied in~\cite{F73,F74,Pizzo-03,FGS-04,Pizzo-05} and the upper spectrum in~\cite{CFFS-12} in the case $|P|<|P_\star|$, while no spectral result seems available for $|P|>|P_\star|$, and the only known result related to Cherenkov radiation is a weak form of instability for the embedded mass shell~\cite{DFP10}.
In the sequel, we shall focus on the case of massive bosons in the weak-coupling regime and use Mourre's theory to investigate the essential spectrum around the embedded mass shell beyond the two-boson threshold,
aiming at a detailed understanding of Cherenkov radiation in that case.

\subsubsection{Description of the model}
The state space for the Nelson model is given by the product Hilbert space
\begin{equation}\label{eq:Nelson-Hilbert}
\Hc\,:=\,\Hc^\pu\otimes\Hc^\fu,
\end{equation}
where:
\begin{enumerate}[---]
\item $\Hc^\pu:=\Ld^2(\R^d)$ is the state space for the non-relativistic quantum particle, and we denote respectively by~$x$ and $p=\frac1i\nabla_x$ the particle position and momentum coordinates;
\smallskip\item $\Hc^\fu$ is the state space for the quantized radiation field and takes form of the bosonic Fock space
\[\Hc^\fu\,:=\,\Gamma_s(\Hf)\,:=\,\textstyle\bigoplus_{n=0}^\infty\Gamma_s^{(n)}(\Hf),\]
constructed on the single-boson space $\Hf:=\Ld^2(\R^d)$.
In other words, we set $\Gamma_s^{(0)}(\Hf):=\C\Omega$ with $\Omega$ the vacuum state, and for $n\ge1$ the $n$-boson state space is the $n$-fold symmetric tensor product
\[\Gamma_s^{(n)}(\Hf)\,:=\,\Hf^{\otimes_sn}.\]
We work in the momentum representation, with $k\in\R^d$ standing for the momentum coordinate of the field bosons.
\end{enumerate}
On this bosonic Fock space~$\Hc^\fu$, we use standard notation for creation and annihilation operators $\{a^*(k)\}_{k\in\R^d}$ and $\{a(k)\}_{k\in\R^d}$, which obey the canonical commutation relations
\[[a^*(k),a^*({k'})]=[a(k),a({k'})]=0,\qquad[a(k),a^*({k'})]=\deltabar(k-k'),\qquad a(k)\Omega=0.\]
We also write $\db\Gamma(A)$ for the second quantization of an operator $A$ on $\Hf$, and in particular $N:=\db\Gamma(\mathds1)$ is the number operator on $\Hc^\fu$.
In this setting, we consider the following translation-invariant Hamiltonian,
\begin{equation}\label{eq:Nelson-Hamilt}
H_g\,:=\,H^\pu\otimes\mathds1_{\Hc^\fu}+\mathds1_{\Hc^\pu}\otimes H^\fu+g\Phi(\rho_x)\qquad \text{on $\Hc$},
\end{equation}
where:
\begin{enumerate}[---]
\item the Hamiltonian of the free quantum particle is given by the standard non-relativistic dispersion relation
\[H^\pu\,:=\,\tfrac12p^2\qquad\text{on $\Hc^\pu$};\]
\item the free field Hamiltonian is given by second quantization,
\[H^\fu\,:=\,\db\Gamma(\omega)\,=\,\int_{\R^d}\omega(k)\,a^*(k) a(k)\,\dbar k\qquad\text{on $\Hc^\fu$},\]
where for bosons of mass $m\ge0$ the single-boson dispersion relation reads
\begin{equation}\label{eq:Nelson}
\omega(k)\,:=\,\sqrt{m^2+|k|^2};
\end{equation}
\item the real number $g$ is the coupling constant for the particle with the bosonic field;
\smallskip\item the interaction Hamiltonian is given by a translation-invariant field operator
\begin{equation}\label{eq:Nelson-interact}
\Phi(\rho_x)\,:=\,\int_{\R^d}\rho(k)\Big(a^*(k)e^{-ik\cdot x}+a(k) e^{ik\cdot x}\Big)\,\dbar k,
\end{equation}
for some real-valued interaction kernel $\rho\in\Ld^2(\R^d)$ with $\rho\not\equiv0$.
\end{enumerate}
Our results on this model will be restricted to the case of massive bosons $m>0$ in the weak-coupling regime $|g|\ll1$.
We could also treat the case of a single-boson dispersion relation of the form $\omega(k)=m+|k|$ provided $m>0$.
Regarding the interaction kernel $\rho$, we shall need to assume strong enough regularity both in infrared and ultraviolet domains, typically requiring $\rho$ to have both some $H^s$ regularity and some polynomial decay.

Before studying this model, we recall its standard well-posedness properties.
For that purpose, we first define the vector subspace
\[\Cc^\fu\,:=\,\db\Gamma_\finu(C_c^\infty(\R^d))\,\subset\,\Hc^\fu,\]
where for a vector subspace $\Gg\subset\Hf$ we denote by $\Gamma_\finu(\Gg)$ the algebraic direct sum of the algebraic tensor products $\Gg^{\otimes_sn}$. In these terms, the uncoupled Nelson Hamiltonian
\[H_0\,=\,\Hc^\pu\otimes\mathds1_{H^\fu}+\mathds1_{H^\pu}\otimes \Hc^\fu\]
is clearly essentially self-adjoint on $C_c^\infty(\R^d)\otimes\Cc^\fu$ (henceforth, tensor products between spaces that are not complete are implicitly understood in the algebraic sense).
Besides, as $\rho\in\Ld^2(\R^d)$, standard estimates ensure that the field operator $\Phi(\rho_x)$ is $(\mathds1_{\Hc^\pu}\otimes N^{1/2})$-bounded.
In case of massive bosons $m>0$, as $\db\Gamma(\omega)\ge mN$, this entails that $\Phi(\rho_x)$ is an infinitesimal perturbation of $H_0$. The Kato--Rellich theorem then ensures that for all $g$ the coupled Nelson Hamiltonian $H_g$ is self-adjoint on the same domain $\Dc(H_0)$ and essentially self-adjoint on the same core~$C_c^\infty(\R^d)\otimes\Cc^\fu$.

\subsubsection{Translation invariance}
By definition, cf.~\eqref{eq:Nelson-Hamilt}, the Hamiltonian $H_g$ is translation-invariant in the sense that it commutes with the total momentum operator
\begin{equation*}
P_{\operatorname{tot}}\,:=\,p\otimes\mathds1_{\Hc^\fu}+\mathds1_{\Hc^\pu}\otimes\db\Gamma(k)\qquad\text{on $\Hc$}.
\end{equation*}
This allows to decompose $H_g$ as a direct integral with respect to the spectrum of the latter.
More precisely, in terms of the following unitary transformation, which goes back to Lee, Low, and Pines~\cite{LLP-53},
\[U:\Hc\to\int_{\R^d}^\oplus\Hc^\fu\,\dbar P,\qquad U\,:=\,(F\otimes\Id_{\Hc^\fu})\circ\Gamma(e^{i k\cdot x}),\]
where $F$ stands for the Fourier transform on $\Hc^\pu$ and where $\Gamma$ is the second quantization functor, we obtain the decomposition
\begin{equation}\label{eq:direct-int}
UH_gU^*\,=\,\int_{\R^d}^\oplus H_{g}(P)\,\dbar P\qquad\text{on $\int_{\R^d}^\oplus\Hc^\fu\,\dbar P$},
\end{equation}
where for all $P\in\R^d$ the fiber Hamiltonian $H_g(P)$ takes the form
\begin{equation}\label{eq:HgP}
H_g(P)\,:=\,\tfrac12(P-\db\Gamma(k))^2+H^\fu+g\Phi(\rho)\qquad\text{on $\Hc^\fu$},
\end{equation}
in terms of the fiber interaction Hamiltonian
\begin{equation}\label{eq:Phi-rho-fiber}
\Phi(\rho)\,:=\,\int_{\R^d}\rho(k)\,(a^*(k)+a(k))\,\dbar k.
\end{equation}

We recall well-posedness properties of these fiber Hamiltonians.
First, for~\mbox{$P=0$}, the uncoupled fiber Hamiltonian $H_0(0)=\frac12\db\Gamma(k)^2+H^\fu$ is essentially self-adjoint on~$\Cc^\fu$. Next, for any $P\in\R^d$, noting that
\mbox{$\tfrac12(P-\db\Gamma(k))^2-\tfrac12\db\Gamma(k)^2=\tfrac12|P|^2-P\cdot\db\Gamma(k)$}
is an infinitesimal perturbation of $H_0(0)$, the Kato--Rellich theorem ensures that the uncoupled fiber Hamiltonian
\begin{equation*}
H_0(P)\,=\,\tfrac12(P-\db\Gamma(k))^2+H^\fu
\end{equation*}
is also essentially self-adjoint on~$\Cc^\fu$ and that its domain is independent of $P$,
\begin{equation}\label{eq:dom-HP-Nelson}
\Dc\,:=\,\Dc(H_0(P))\,=\,\Dc(H_0(0))\,=\,\Dc(\db\Gamma(k)^2)\cap\Dc(\db\Gamma(\omega)).
\end{equation}
Besides, as $\rho\in\Ld^2(\R^d)$, standard estimates ensure that the field operator $\Phi(\rho)$ is $N^{1/2}$-bounded.
In case of massive bosons $m>0$, as $\db\Gamma(\omega)\ge mN$, this entails that $\Phi(\rho)$ is an infinitesimal perturbation of $H_0(P)$. The Kato--Rellich theorem then ensures that for all $g$ the coupled fiber Hamiltonian $H_g(P)$ is self-adjoint on the same domain $\Dc$ and essentially self-adjoint on the same core $\Cc^\fu$.

\subsubsection{Energy-momentum spectrum}
In this translation-invariant setting, the natural object of study is the energy-momentum spectrum $\{(P,E):E\in\sigma(H_g(P))\}$, where $\sigma(H_g(P))$ is the spectrum of the fiber Hamiltonian $H_g(P)$ at fixed total momentum $P$.
We start by recalling the explicit structure of this spectrum for the uncoupled Hamiltonian.

\begin{lem}[Spectrum of uncoupled Nelson model]\label{lem:spectrum/Nelson}
Consider the translation-invariant Nelson model with massive bosons $m>0$, cf.~\eqref{eq:Nelson-Hilbert}--\eqref{eq:Phi-rho-fiber}.
Given a total momentum $P\in\R^d$, the spectrum of the uncoupled fiber Hamiltonian~$H_0(P)$ is given by
\begin{equation}\label{eq:sigH0P}
\sigma_{\operatorname{pp}}(H_{0}(P))=\{\tfrac12P^2\},
\quad\sigma_{\operatorname{ac}}(H_{0}(P))=[E_0(P),\infty),
\quad\sigma_{\operatorname{sc}}(H_{0}(P))=\varnothing,
\end{equation}
where the eigenvalue $\frac12P^2$ is simple and is associated with the vacuum state $\Omega$, and where the bottom of the absolutely continuous spectrum is given by
\begin{equation}\label{eq:Sigma0}
E_0(P)\,:=\,\tfrac12c(P)^2+\sqrt{m^2+(|P|-c(P))^2},
\end{equation}
in terms of the unique solution $c(P)\in[0,1)$ of the implicit equation
\[c(P)\,=\,\frac{|P|-c(P)}{\sqrt{m^2+(|P|-c(P))^2}}.\]
Moreover, there is a unique critical value $|P_\star|>1$ such that
\begin{equation}\label{eq:Pstar}
E_0(P_\star)=\tfrac12P_\star^2,
\end{equation}
and the following alternative then holds:
\begin{enumerate}[---]
\item for $|P|<|P_\star|$, the fiber Hamiltonian $H_0(P)$ has an isolated ground state at $\frac12P^2$;
\item for $|P|>|P_\star|$, the fiber Hamiltonian $H_0(P)$ has no ground state and its eigenvalue is embedded in the absolutely continuous spectrum.
\qedhere
\end{enumerate}
\end{lem}

Before turning to our main results on coupled Hamiltonians, we further elaborate on this statement and emphasize the layered structure of the spectrum. By definition~\eqref{eq:HgP}, the uncoupled fiber Hamiltonian commutes with the number operator $N$ and thus splits as a direct sum on many-boson state spaces,
\begin{equation}\label{eq:Hilbert-ac-decomp}
H_0(P)\,=\,\bigoplus_{n=0}^\infty H_0^{(n)}(P)\qquad\text{on~~~$\Hc^\fu\,=\,\bigoplus_{n=0}^\infty\Gamma_s^{(n)}(\Hf)$},
\end{equation}
in terms of the restrictions
\[H_0^{(n)}(P)\,:=\,H_0(P)|_{\Gamma_s^{(n)}(\Hf)}.\]
While the eigenvalue~$\frac12P^2$ is associated with the vacuum state~$\Omega$ and corresponds to a free non-relativistic particle with momentum~$P$, the absolutely continuous spectrum corresponds to states supporting at least one boson,
\begin{equation}\label{eq:spectra-layer-decomp}
\sigma_{\ac}(H_0(P))\,=\,\adh\bigcup_{n=1}^\infty\sigma_\ac\big(H_0^{(n)}(P)\big).
\end{equation}
For all $n\ge1$, the restriction $H_0^{(n)}(P)$ is a multiplication operator in momentum coordinates, with symbol
\begin{equation}\label{eq:symbolE0n-0}
H_0^{(n)}(P;k_1,\ldots,k_n)\,:=\,
\tfrac12\Big(P-\sum_{j=1}^nk_j\Big)^2+\sum_{j=1}^n\omega(k_j).
\end{equation}
Its spectrum is absolutely continuous and coincides with the essential image of the symbol,
\[\sigma_\ac\big(H_0^{(n)}(P)\big)\,=\,\big[E_0^{(n)}(P),\infty\big),\qquad\sigma_\pp\big(H_0^{(n)}(P)\big)\,=\,\sigma_\sg\big(H_0^{(n)}(P)\big)\,=\,\varnothing,\]
in terms of the so-called $n$-boson energy threshold
\begin{equation}\label{eq:energy-thresh-def}
E_0^{(n)}(P)\,:=\,\min_{k_1,\ldots,k_n\in\R^d}H_0^{(n)}(P;k_1,\ldots,k_n).
\end{equation}
In the case of massive bosons $m>0$, it is easily checked that
\[E_0^{(n)}(P)<E_0^{(n+1)}(P)\quad\text{for all $n$,}\qquad\text{and}\qquad E_0^{(n)}(P)\uparrow\infty\quad\text{as $n\uparrow\infty$},\]
cf.~Lemma~\ref{lem:mono-increments} below. In view of~\eqref{eq:spectra-layer-decomp}, this ensures in particular that the bottom of the absolutely continuous spectrum is
\[E_0(P)\,:=\,E_0^{(1)}(P)\,=\,\min_{n\ge1}E_0^{(n)}(P),\]
and we then recover the expression~\eqref{eq:Sigma0} by computing the minimum of~\eqref{eq:symbolE0n-0} for~$n=1$.
Due to the layered structure of the spectrum, cf.~\eqref{eq:spectra-layer-decomp}, our results in the sequel are naturally restricted away from energy thresholds.

\subsubsection{Main results}
We may now formulate our main results on the massive Nelson model.
Our starting point is the following perturbative Mourre commutator result: for \mbox{$|P|>|P_\star|$}, as the eigenvalue $\frac{1}{2}P^2$ is embedded in the essential spectrum, a Mourre estimate is proven around and above the eigenvalue away from energy thresholds, as well as below the two-boson threshold, and this is complemented with a regularity statement for the interaction Hamiltonian.
We refer to Appendix~\ref{app:Mourre} for standard definitions and notation related to Mourre theory.

\begin{theor}[Mourre estimate for Nelson model]\label{th:main/Nelson}
Consider the translation-invariant Nelson model with massive bosons $m>0$, cf.~\eqref{eq:Nelson-Hilbert}--\eqref{eq:Phi-rho-fiber}. Given a total momentum~\mbox{$|P|>|P_\star|$}, define $n_P\ge1$ such that
\begin{equation}\label{eq:def-n0}
\tfrac12|P|^2\in\big[E_0^{(n_P)}(P),E_0^{(n_P+1)}(P)\big),
\end{equation}
where we recall definitions~\eqref{eq:Pstar} and~\eqref{eq:energy-thresh-def}.
Then, for $n=1$ as well as for any $n\geq n_P$, we can construct an operator~$A_{P,n}$ on $\Hc^\fu$, essentially self-adjoint on~$\Cc^\fu$, with the following properties.
\begin{enumerate}[(i)]
\item The uncoupled fiber Hamiltonian $H_0(P)$ is of class $C^\infty(A_{P,n})$. Moreover, the unitary group generated by $A_{P,n}$ leaves the domain of $H_0(P)$ invariant, and the iterated commutators $\ad_{iA_{P,n}}^s(H_0(P))$ extend as $H_0(P)$-bounded operators for all $s\ge0$.
\item For all $\e>0$ and all energy intervals $I\subset \big[E_0^{(n)}(P)+\e,E_0^{(n+1)}(P)\big)$, the following Mourre estimate holds with respect to $A_{P,n}$ on $I$,
\begin{equation*}
\mathds1_{I}(H_{0}(P))\,[H_0(P),iA_{P,n}]\,\mathds1_{I}(H_0(P))~\ge~\e\bar\Pi_{\Omega} \,\mathds1_{I}(H_0(P))\,\bar\Pi_{\Omega},
\end{equation*}
where $\bar\Pi_{\Omega}$ is the orthogonal projection on $\C\Omega^\perp$. In particular, the Mourre estimate is strict if the interval $I$ does not contain the eigenvalue $\frac12P^2$.
\smallskip\item The fiber interaction Hamiltonian $\Phi(\rho)$ satisfies the following regularity condition: if the interaction kernel $\rho$ belongs to $H^\nu(\R^d)$ with $\langle k\rangle^\nu\nabla^\nu\rho\in\Ld^2(\R^d)$ for some $\nu\ge1$, then the iterated commutators $\ad_{iA_{P,n}}^s(\Phi(\rho))$ extend as $H_0(P)^{1/2}$-bounded operators for all $1\le s\le\nu$.
\qedhere
\end{enumerate} 
\end{theor}

Compared to previous work on the topic~\cite{AMZ05,MR13,DM15}, this provides the first Mourre estimate above the two-boson threshold. In addition, the $C^\infty$-regularity stated in item~(i) allows us to exploit for the first time the full power of Mourre's theory, cf.~Appendix~\ref{app:Mourre}, while previous constructions were restricted to $C^2$-regularity.
As a corollary, we deduce the following description of the essential spectrum of fiber Hamiltonians at weak coupling, which proves in particular the instability of the mass shell~$E=\frac12P^2$ of the free non-relativistic particle when coupled to the bosonic field.
This is further complemented with a dynamical resonance description, which exploits the $C^\infty$-regularity and is thus new even below the two-boson threshold.
It constitutes a precise formulation of Cherenkov radiation for the massive Nelson model.

\begin{cor}[Cherenkov radiation for Nelson model]\label{cor:main2}
Consider the translation-invariant Nelson model with massive bosons $m>0$, cf.~\eqref{eq:Nelson-Hilbert}--\eqref{eq:Phi-rho-fiber},
and assume that the interaction kernel~$\rho$ belongs to $H^2(\R^d)$ with $\langle k\rangle^2\nabla^2\rho\in\Ld^2(\R^d)$.
Given a total momentum~$|P|>|P_\star|$, define $n_P\ge1$ as in~\eqref{eq:def-n0}, and assume that Fermi's condition holds,
\begin{equation}\label{eq:def-gamma0-00}
\gamma_P\,:=\,\tfrac12(2\pi)^{1-d}\int_{\{k\,:\,\frac12(P-k)^2+\omega(k)=\frac12P^2\}}\tfrac{|\rho(k)|^2}{|k-P+\nabla\omega(k)|}\,\db\mathcal H_{d-1}(k)\,>\,0,
\end{equation}
where $\mathcal H_{d-1}$ stands for the $(d-1)$th-dimensional Hausdorff measure.
(This condition holds in particular if $\rho$ does not vanish.) Then, the following properties hold.
\begin{enumerate}[(i)]
\item \emph{Absence of embedded mass shell:}\\
For all $g\ne0$, the essential spectrum of the fiber Hamiltonian $H_g(P)$ is purely absolutely continuous below~$E_0^{(2)}(P)$ and above~$E_0^{(n_P)}(P)$ away from thresholds: more precisely,
there is a sequence $(C_{P,n})_n$ such that $H_g(P)$ has purely absolutely continuous spectrum in
\begin{multline*}
\hspace{1cm}I_g(P)\,:=\,\Big(E_0^{(1)}(P)+\sqrt gC_{P,1}\,,\,E_0^{(2)}(P)-gC_{P,1}\Big)\\
\bigcup\bigcup_{n\ge n_P}\Big(E_0^{(n)}(P)+\sqrt gC_{P,n}\,,\,E_0^{(n+1)}(P)-gC_{P,n}\Big).
\end{multline*}
\item \emph{Quasi-exponential decay law:}\\
Further assume that for some $\nu\ge0$ the interaction kernel $\rho$ belongs to $H^{5+\nu}(\R^d)$ with $\langle k\rangle^{5+\nu}\nabla^{5+\nu}\rho\in\Ld^2(\R^d)$.
Then, there is $g_0>0$ such that, for all smooth cut-off functions~$h$ supported in $I_{g_0}(P)$ and equal to $1$ in a neighborhood of the uncoupled eigenvalue $\frac12P^2$, there holds for all $t\ge0$ and $|g|\le g_0$,
\[\qquad\Big|\Big\langle\Omega\,,\,e^{-itH_g(P)}h(H_g(P))\Omega\Big\rangle-e^{-itz_g(P)}\Big|\,\lesssim_{g_0,\rho,h}\,\left\{\begin{array}{ll}
g^2|\!\log g|\langle t\rangle^{-\nu},&\text{if $\nu\ge0$};\\
g^2\langle t\rangle^{-(\nu-1)},&\text{if $\nu\ge1$};
\end{array}\right.\]
where the dynamical resonance $z_g(P)$ is given by Fermi's golden rule,
\[\qquad z_g(P)\,=\,\tfrac12P^2-g^2(\theta_P+i\gamma_P),\]
where $\gamma_P>0$ is defined in~\eqref{eq:def-gamma0-00} and where the real part $\theta_P\in\R$ takes the form
\begin{multline*}
\qquad\theta_P\,:=\,(2\pi)^{-d}\,\pv\int_{\R}(t-\tfrac12P^2)^{-1}\bigg(\int_{\{k\,:\,\frac12(P-k)^2+\omega(k)=t\}}\tfrac{|\rho(k)|^2}{|k-P+\nabla\omega(k)|}\,\db\Hc_{d-1}(k)\bigg)\,\db t.
\end{multline*}
\end{enumerate}
In particular, for all $u_\circ\in\Ld^2(\R^d)$ with Fourier transform compactly supported in the set $\{P:|P|>|P_\star|\}$, provided that $\rho$ does not vanish and satisfies the requirements of~(ii) for some $\nu\ge0$, there holds uniformly for all~$t\ge0$,
\[\langle(\delta_x\otimes\Omega),e^{-itH_g}(u_\circ\otimes\Omega)\rangle
=\int_{\R^d}\hat u_\circ(P)e^{ix\cdot P-itz_g(P)}\,\dbar P+o_{g}(1),\]
where $o_{g}(1)$ tends to $0$ in $\Ld^\infty_x(\R^d)$ as $g\downarrow0$ (depending on $\rho,u_\circ$).\qedhere
\end{cor}

We briefly comment on possible extensions and open problems.
First note that the above result gives curiously no access to the spectrum in the energy interval $\big(E_0^{(2)}(P),E_0^{(n_P)}(P)\big)$ between the two-boson threshold and the last threshold below the embedded eigenvalue: although we also expect the same behavior in this interval (away from thresholds), a different type of construction seems to be needed for conjugate operators.
Another question concerns the use of the above Mourre estimate to further investigate asymptotic completeness, thus aiming to extend~\cite{Gerard-Moller-Rasmussen-11,DM15} beyond the two-boson threshold in the weak-coupling regime; this is postponed to a future work.
Finally, a last important open question concerns the validity of the Mourre estimate at large coupling: this was solved in~\cite{MR13} below the two-boson threshold and we may expect our present contribution to give valuable inspiration to get beyond that.

The Nelson model belongs to the class of so-called translation-invariant Pauli--Fierz models. These also include Fröhlich's polaron model in solid-state physics~\cite{F54,Feynman-55}, as well as non-relativistic QFT models with vector bosons.
Our findings are easily extended to those settings:
\begin{enumerate}[---]
\item The polaron model introduced by Fröhlich~\cite{F54,Feynman-55} describes an electron interacting with lattice vibrations of a polar crystal. These are naturally represented in terms of a Bose field over a crystalline lattice and we consider the continuum limit of the latter, thus treating the crystal as a polarizable continuum. The model then takes the same form as the Nelson model~\eqref{eq:Nelson-Hilbert}--\eqref{eq:Nelson-Hamilt}, where the single-boson dispersion relation is now taken to be constant, $\omega(k)=1$, so that the free field Hamiltonian is the number operator~$H^\fu=N$. We refer to~\cite{M06} and references therein for a detailed discussion of this model.
Our analysis of the massive Nelson model can be repeated mutatis mutandis in this setting and
we note that several calculations actually reduce dramatically: in particular, the critical value of the total momentum and the energy thresholds are simply
\[|P_\star|=\sqrt2\qquad\text{and}\qquad E_0^{(n)}(P)=n.\]
Corollary~\ref{cor:main2} yields the first rigorous justification of Cherenkov radiation for the polaron model (see formal discussion in~\cite[p.227--230]{Feynman-72}).
\smallskip\item
Consider the translation-invariant non-relativistic QFT model for a non-relativistic quantum particle minimally coupled to a quantized radiation field of relativistic vector bosons~\cite{Spohn-04,Minlos-08}.
For $d=3$, the single-boson space is then $\Hf:=\Ld^2(\R^3\times\{+,-\})$, where~$+/-$ stands for boson polarization, and we consider the Hamiltonian
\[H_\alpha\,:=\,\tfrac12\big(p-\alpha^\frac12A_x\big)^2+\mathds1\otimes H^\fu,\]
where $\alpha\ge0$ is the coupling constant, where the vector potential $A_x$ is linear in creation and annihilation operators, and where the free field Hamiltonian is as before \mbox{$H^\fu=\db\Gamma(\omega)$} with single-boson dispersion relation $\omega(k)=\sqrt{m^2+k^2}$.
In case of massive vector bosons $m>0$, our analysis of the massive Nelson model is easily adapted to this other model under suitable regularity assumptions on the interaction kernel defining the vector potential; we skip the detail.
\end{enumerate}
While we focus here on massive bosons, the case of massless bosons is different and will be commented at the end of the next section.

\subsection{Quantum friction model}\label{sec:intro-friction}
We turn to the quantum version~\cite{Bruneau-07,DFS-17} of a translation-invariant Hamiltonian model for friction introduced by Bruneau and De Bièvre~\cite{Bruneau-DeBievre-02}.
It describes a non-relativistic quantum particle moving through a translation-invariant medium consisting of uncoupled quantized vibration fields at each point in space.
This model happens to be substantially simpler to study than the Nelson model, precisely due to the fact that vibration fields are uncoupled in space, and we shall thus be able to further treat the case of massless bosons for this model. In the sequel, we aim at a detailed understanding of friction effects in the weak coupling regime, improving on previous results in~\cite{DFS-17}.

\subsubsection{Description of the model}
The state space for the model is given by the product Hilbert space
\begin{equation}\label{eq:Hilbert/friction}
\Hc\,:=\,\Hc^\pu\otimes\Hc^\fu,
\end{equation}
where:
\begin{enumerate}[---]
\item $\Hc^\pu:=\Ld^2(\R^d)$ is the state space for the non-relativistic quantum particle, and we denote respectively by $x$ and $p=\frac1i\nabla_x$ the particle position and momentum coordinates;
\item $\Hc^\fu$ is the state space for the quantized vibration fields and takes form of the bosonic Fock space $\Hc^\fu:=\Gamma_s(\Hf)$ constructed on the single-boson space $\Hf:=\Ld^2(\R^q\times\R^d)$. We work in momentum representation,
with $k\in\R^q$ standing for the momentum coordinate for vibrational degrees of freedom, and with $\xi\in\R^d$ standing for the momentum coordinate dual to the particle position $x$.
\end{enumerate}
On this bosonic Fock space $\Hc^\fu$, we use standard notation for creation and annihilation operators $\{a^*(k,\xi)\}_{(k,\xi)\in\R^q\times\R^d}$ and $\{a(k,\xi)\}_{(k,\xi)\in\R^q\times\R^d}$, we use the notation $\db\Gamma(A)$ for the second quantization of an operator $A$ on $\Hf$, and in particular $N:=\db\Gamma(\mathds1)$ is the number operator.
In this setting, we consider the following translation-invariant Hamiltonian,
\begin{equation}\label{eq:Hamilt/friction}
H_g\,:=\,H^\pu\otimes\mathds1_{\Hc^\fu}+\mathds1_{\Hc^\pu}\otimes H^\fu+g\Phi(\rho_x)\qquad\text{on $\Hc$},
\end{equation}
where:
\begin{enumerate}[---]
\item the Hamiltonian of the free quantum particle is given by the standard non-relativistic dispersion relation
\[H^\pu\,:=\,\tfrac12p^2\qquad\text{on $\Hc^\pu$};\]
\item the free field Hamiltonian is given by second quantization
\[H^\fu\,:=\,\db\Gamma(\omega)\,=\,\iint_{\R^q\times\R^d}\omega(k,\xi)\,a^*(k,\xi) a(k,\xi)\,\dbar k\dbar\xi\qquad\text{on $\Hc^\fu$},\]
where the single-boson dispersion relation reads
\begin{equation}\label{eq:Nelson}
\omega(k,\xi)\,:=\,|k|,
\end{equation}
which corresponds to massless bosons and is naturally taken independent of $\xi$ as vibration fields at different values of $x$ are not coupled;
\smallskip\item the real number $g$ is the coupling constant for the particle with the vibration fields;
\smallskip\item the interaction Hamiltonian is given by a translation-invariant field operator
\begin{equation}\label{eq:Nelson-interact}
\Phi(\rho_x)\,:=\,\iint_{\R^q\times\R^d}\rho(k,\xi)\Big(a^*(k,\xi)e^{-i\xi\cdot x}+a(k,\xi) e^{i\xi\cdot x}\Big)\,\dbar k\dbar\xi,
\end{equation}
for some real-valued interaction kernel
$\rho\in\Ld^2(\R^q\times\R^d)$ with $\rho\not\equiv0$.
\end{enumerate}
Our results on this model will be restricted to the weak-coupling regime $|g|\ll1$.
Regarding the interaction kernel $\rho$ in~\eqref{eq:Nelson-interact}, we shall need strong ultraviolet regularity, but a quite general infrared behavior will be allowed. More precisely, we shall consider the following assumption, for some $\nu\ge0$,
\begin{enumerate}[{\rm(Reg$_\nu$)}]
\item There holds $(1+|k|^{-\frac12})(k\cdot\nabla_k)^\alpha(\xi\cdot\nabla_\xi)^\beta\nabla_\xi^\gamma\rho\in\Ld^2(\R^q\times\R^d)$ for all $\alpha,\beta,\gamma\ge0$ with~$\alpha+\beta+\gamma\le\nu$.
\end{enumerate}
Note that this holds for all $\nu$ for instance if $\rho$ takes the particular form $\rho(k,\xi)=|k|^\mu\sigma(k,\xi)$ for some $\sigma\in\Sc(\R^q\times\R^d)$ and $\mu>-\frac12(q-1)$. This restriction on the infrared parameter $\mu$ is essentially optimal in the sense that it precisely ensures that the interaction Hamiltonian~$\Phi(\rho_x)$ be relatively bounded with respect to $H_0$.

\subsubsection{Translation invariance}
As in~\eqref{eq:direct-int} (see also~\cite[Section~3.2]{DFS-17}), we can decompose the above Hamiltonian~$H_g$ as a direct integral
\[H_g\,\cong\,\int_{\R^d}^\oplus H_g(P)\,\dbar P\qquad\text{on $\int_{\R^d}^\oplus\Hc^\fu\,\dbar P$},\]
where for all $P\in\R^d$ the fiber Hamiltonian $H_g(P)$ takes the form
\begin{equation}\label{eq:def-HgP-frict}
H_g(P)\,:=\,\tfrac12(P-\db\Gamma(\xi))^2+H^\fu+g\Phi(\rho).
\end{equation}
We recall well-posedness properties of these fiber Hamiltonians. First note that the uncoupled fiber Hamiltonian $H_0(P)$ is essentially self-adjoint on
\[\Cc^\fu\,:=\,\db\Gamma_\finu(C^\infty_c(\R^q\times\R^d)),\]
and that its domain is independent of $P$,
\begin{equation}\label{eq:dom-HP-friction}
\Dc\,:=\,\Dc(H_0(P))\,=\,\Dc(H_0(0))\,=\,\Dc(\db\Gamma(\xi)^2)\cap\Dc(\db\Gamma(|k|)).
\end{equation}
Besides, standard estimates ensure that $\Phi(\rho)$ is $(H^\fu)^{1/2}$-bounded provided that the interaction kernel satisfies Assumption~(Reg$_0$).
The Kato--Rellich theorem then ensures that for all $g$ the coupled fiber Hamiltonian $H_g(P)$ is self-adjoint on the same domain $\Dc$ and essentially self-adjoint on the same core $\Cc^\fu$.

\subsubsection{Energy-momentum spectrum}
We start by recalling the structure of the energy-momentum spectrum of the uncoupled Hamiltonian.

\begin{lem}[Spectrum of uncoupled quantum friction model]\label{lem:spectrum-bis}
Consider the quantum friction model~\eqref{eq:Hilbert/friction}--\eqref{eq:def-HgP-frict}.
For any total momentum $P\in\R^d$, the spectrum of the uncoupled fiber Hamiltonian $H_0(P)$ is given by
\begin{equation}
\sigma_{\operatorname{pp}}(H_{0}(P))=\{\tfrac12P^2\},
\quad\sigma_{\operatorname{ac}}(H_{0}(P))=[0,\infty),
\quad\sigma_{\operatorname{sc}}(H_{0}(P))=\varnothing,
\end{equation}
where the eigenvalue $\frac12P^2$ is simple and is associated with the vacuum state $\Omega$.
\end{lem}

We emphasize the key difference with the Nelson model: as vibration fields at different positions in space are not coupled, the propagation speed of bosons in space vanishes and the spectrum of the uncoupled fiber Hamiltonian is therefore $[0,\infty)$ for any total momentum~$P$, as stated above.
In particular, the eigenvalue $\frac12P^2$ is strictly embedded in absolutely continuous spectrum whenever $P\ne0$. Hence, Cherenkov radiation is expected to occur whenever the non-relativistic particle has a non-vanishing momentum, which then results in a friction effect that tends to stop the particle.

\subsubsection{Main results}
We may now formulate our main results on the quantum friction model.
Our starting point is the following perturbative Mourre commutator result.
Note that our construction only yields a Mourre estimate above $\frac1{18}P^2$, but this is enough for our purposes as it covers a neighborhood of the embedded eigenvalue $\frac12P^2$.

\begin{theor}[Mourre estimate for quantum friction model]\label{th:main/friction}
Consider the quantum friction model~\eqref{eq:Hilbert/friction}--\eqref{eq:def-HgP-frict}. Given a total momentum $P\neq 0$ and given $\delta>0$, we can construct a self-adjoint operator $A_{P;\delta}$ in $\Hc^\fu$, essentially self-adjoint on $\Cc^\fu$, with the following properties.
\begin{enumerate}[(i)]
\item The uncoupled fiber Hamiltonian $H_0(P)$ is of class $C^\infty(A_{P;\delta})$. Moreover, the unitary group generated by $A_{P;\delta}$ leaves the domain of $H_0(P)$ invariant, and the iterated commutators $\ad_{iA_{P;\delta}}^s(H_0(P))$ extend as $H_0(P)$-bounded operators for all $s\ge0$.
\smallskip\item For all $\e>0$ and all energy intervals $I\subset\big[\frac1{18}P^2+\delta+\e\,,\,\infty\big)$, the following Mourre estimate holds with respect to $A_{P;\delta}$ on $I$,
\begin{equation*}
\mathds1_I(H_0(P))\,[H_0(P),iA_{P;\delta}]\,\mathds1_I(H_0(P))\,\ge\,\e\mathds1_I(H_0(P))-P^2\Pi_\Omega,
\end{equation*}
where $\Pi_\Omega$ is the orthogonal projection on $\C\Omega$. In particular, the Mourre estimate is strict if the interval $I$ does not contain the eigenvalue $\frac12P^2$.
\smallskip\item The fiber interaction Hamiltonian $\Phi(\rho)$ satisfies the following regularity condition: if the interaction kernel $\rho$ satisfies~{\rm(Reg$_\nu$)} for some $\nu\ge1$, then the iterated commutators $\ad_{iA_{P;\delta}}^s(\Phi(\rho))$ extend as $H_0(P)^{1/2}$-bounded operators for all $1\le s\le\nu$.\qedhere
\end{enumerate} 
\end{theor}

This provides the first Mourre estimate with $C^\infty$-regularity for this model. In contrast, in~\cite{DFS-17}, the authors used the generator of radial translations $R:=\frac i2\db\Gamma(\frac{k}{|k|}\cdot\nabla_k+\nabla_k\cdot\frac{k}{|k|})$ as a natural conjugate and were confronted both with the singularity of this operator at small $k$ and with the dramatic lack of associated regularity: the commutator with $H_0(P)$ is formally
\[[H_0(P),iR]\,=\,N,\]
which is positive on $\C\Omega^\bot$ but is not controlled by $H_0(P)$ in the case of massless bosons.
To accommodate this difficulty, the authors of~\cite{DFS-17} had to appeal to the ``singular'' Mourre theory developed in~\cite{Skibsted-98,Moller-Skibsted-04,GGM04a,GGM04,Faupin-Moller-Skibsted-11}, which is precisely meant for this situation but has weaker consequences than the regular theory.
In particular, they deduced in~\cite{DFS-17}  the instability of the embedded mass shell $E=\frac12P^2$ at weak coupling, but no resonance description was obtained. In addition, their restriction on the infrared behavior of the interaction kernel~$\rho$ was much stronger than what we impose here in Assumption~(Reg$_\nu$).
Rather combining Theorem~\ref{th:main/friction} above with the full power of regular Mourre theory,
we are led to the following improved description.

\begin{cor}[Quantum friction]\label{cor:friction}
Consider the quantum friction model~\eqref{eq:Hilbert/friction}--\eqref{eq:def-HgP-frict}, and assume that the interaction kernel $\rho$ satisfies Assumption~{\rm(Reg$_2$)}. Given a total momentum $P\neq 0$, assume that Fermi's condition holds,
\begin{equation}\label{eq:def-Fermi-cond-friction}
\gamma_P\,:=\,\tfrac12(2\pi)^{1-d}\int_{|k|\le \frac12 P^2}\int_{\{\xi\,:\,\frac12(P-\xi)^2=\frac12 P^2-|k|\}}\tfrac{|\rho(k,\xi)|^2}{\sqrt{(P-\xi)^2+1}}\db\Hc_{d-1}(\xi)\dbar k~>~0,
\end{equation}
where $\Hc_{d-1}$ stands for the $(d-1)$th-dimensional Hausdorff measure. (This condition holds in particular if $\rho$ does not vanish.) Then, the following properties hold.
\begin{enumerate}[(i)]
\item \emph{Absence of embedded mass shell:}\\
For all $g\ne0$, the coupled fiber Hamiltonian $H_g(P)$ has purely absolutely continuous spectrum in
\[\qquad I_g(P)\,:=\,\big(\tfrac1{18}P^2+\sqrt gC_P\,,\,\infty\big).\]
\item \emph{Quasi-exponential decay law:}\\
Further assume that for some $\nu\ge0$ the interaction kernel $\rho$ satisfies Assumption~{\rm(Reg$_{5+\nu}$)}. Then, there is $g_0>0$ such that, for all smooth cut-off functions~$h$ supported in $I_{g_0}(P)$ and equal to $1$ in a neighborhood of the uncoupled eigenvalue $\frac12P^2$, there holds for all $t\ge0$ and $|g|\le g_0$,
\[\qquad\Big|\Big\langle\Omega\,,\,e^{-itH_g(P)}h(H_g(P))\Omega\Big\rangle-e^{-itz_g(P)}\Big|\,\lesssim_{g_0,\rho,h}\,\left\{\begin{array}{ll}
g^2|\!\log g|\langle t\rangle^{-\nu},&\text{if $\nu\ge0$};\\
g^2\langle t\rangle^{-(\nu-1)},&\text{if $\nu\ge1$};
\end{array}\right.\]
where the dynamical resonance $z_g(P)$ is given by Fermi's golden rule,
\[\qquad z_g(P)\,=\,\tfrac12P^2-g^2(\theta_P+i\gamma_P),\]
where $\gamma_P>0$ is defined in~\eqref{eq:def-Fermi-cond-friction} and where the real part $\theta_P\in\R$ takes the form
\begin{multline*}
\qquad\theta_P\,:=\,(2\pi)^{-d}\,\pv\int_{\R}(t-\tfrac12P^2)^{-1}\\
\times\bigg(\int_{|k|\le t}\int_{\{\xi\,:\,\frac12(P-\xi)^2=t-|k|}\tfrac{|\rho(k,\xi)|^2}{\sqrt{(P-\xi)^2+1}}\db\Hc_{d-1}(\xi)\dbar k\bigg)\db t.
\end{multline*}
\end{enumerate}
In particular, for all $u_\circ\in\Ld^2(\R^d)$ with compactly supported Fourier transform, provided that $\rho$ does not vanish and satisfies the requirements of~(ii) for some $\nu\ge0$, there holds uniformly for all $t\ge0$,
\[\langle(\delta_x\otimes\Omega),e^{-itH_g}(u_\circ\otimes\Omega)\rangle
=\int_{\R^d}\hat u_\circ(P)e^{ix\cdot P-itz_g(P)}\,\dbar P+o_{g}(1),\]
where $o_{g}(1)$ tends to $0$ in $\Ld^\infty_x(\R^d)$ as $g\downarrow0$ (depending on $\rho,u_\circ$).\qedhere
\end{cor}

As this result concerns the case of massless bosons, we may naturally wonder whether our constructions could be adapted to some extent to the massless Nelson model. At the moment, however, we leave this as an open question.
Compared to the above quantum friction model, the difficulty is that for the Nelson model~\eqref{eq:Nelson-Hilbert}--\eqref{eq:Nelson-Hamilt} the momentum coordinates~$k$ and~$\xi$ coincide: in view of our construction~\eqref{eq:def-Ckst} of the conjugate operator below, we are then essentially reduced to controlling the commutator $\big[\sum_{j=1}^n|\nabla_{z_j}|\,,\,\max_jz_j\big]$ uniformly with respect to~$n$, which we do not know how to do. Controlling this commutator may actually require to further adapt our construction of the conjugate operator.
For massive bosons, the problem simplifies drastically as a bound $O(n)$ on this commutator is sufficient, cf.~Lemma~\ref{lem:Cinfty-H0P-A-Nelson}.


\subsection{Link to random Schrödinger operators}\label{sec:random-Schr}
As is well known, e.g.~\cite[Section~1.3]{DRFP-10} and~\cite[Lemma~5.6]{DS21}, random Schrödinger operators can be viewed as particular instances of Pauli--Fierz models. This comparison was actually our original motivation for the present contribution, in link with our previous work~\cite{DS21}. More precisely, given a stochastically translation-invariant Gaussian field $V$ on $\R^d$, constructed on a probability space~$\Omega$, a Gaussian chaos decomposition of~$\Ld^2(\Omega)$ ensures that the random Schrödinger operator $-\frac12\triangle+gV$ in $\Ld^2(\R^d\times\Omega)$ is unitarily equivalent to the translation-invariant Hamiltonian
\[H_g\,:=\,\tfrac12p^2\otimes\mathds1_{\Hc^\fu}+g\Phi(\rho_x)\qquad\text{on $\Hc=\Hc^\pu\otimes\Hc^\fu$},\]
where $\Hc^\pu:=\Ld^2(\R^d)$ and $\Hc^\fu:=\Gamma_s(\Hf)$ with $\Hf:=\Ld^2(\R^d)$, and where the interaction kernel~$\rho$ is such that $\rho\ast\rho$ is the covariance function of $V$.
In other words, the action of the Gaussian field is viewed as the interaction with a bosonic field.
Via this isomorphism, the fiber decomposition~\eqref{eq:direct-int} for $H_g$ is precisely equivalent to the Floquet--Bloch type fibration that we introduced in~\cite{DS21},
\[H_g\,\cong\,\int_{\R^d}^{\oplus}H_g(P)\,\dbar P,\qquad \text{where}\qquad H_g(P)\,:=\,\tfrac12(P-k)^2+g\Phi(\rho)\qquad\text{on $\Hc^\fu$}.\]
Compared to QFT models studied in this contribution,
the key difficulty is that degrees of freedom associated with the random potential $V$ do not evolve over time, hence the dispersion relation for the corresponding bosons is trivial, $\omega=0$, and the free field Hamiltonian vanishes.
This makes the study of random Schrödinger operators particularly intricate from this perspecive: the perturbation~$\Phi(\rho)$ is not relatively bounded with respect to $H_0(P)$, and its commutators with any Mourre conjugate for $H_0(P)$ are not relatively bounded either.
In other words, $\Phi(\rho)$ is not a regular perturbation in the sense of Mourre's theory.
In~\cite{DS21}, we proceeded by truncating~$\Phi(\rho)$ to state spaces with a bounded number of bosons, depending on the size of the coupling constant $g$, and this allows us in the end to deduce a dynamical resonance description at least up to the kinetic timescale $t\lesssim g^{-2}$. As explained in~\cite{DS21}, any improvement would be of tremendous interest in link with the quantum diffusion conjecture.

\section{Quantum friction model}\label{sec:friction-pr}

This section is devoted to the proof of our main results on the quantum friction model at weak coupling, cf.~\eqref{eq:Hilbert/friction}--\eqref{eq:def-HgP-frict}.
We start with the construction of a suitable conjugate operator and with the proof of the Mourre estimate, thus establishing Theorem~\ref{th:main/friction}, before turning to consequences on the metastability of the embedded mass shell.

\subsection{Construction of conjugate operator}\label{sec:constr-fric-commut}
As the uncoupled fiber Hamiltonian splits as a sum
\[H_0(P)=\tfrac12(P-\db\Gamma(\xi))^2+\db\Gamma(|k|)\]
of two operators acting on different variables,
we shall similarly construct a conjugate operator as a sum
\begin{equation}\label{eq:def-AP-BPR}
A_P=B_P+D_1,
\end{equation}
where $B_P$ acts only on the variable $\xi$ and $D_1$ on the variable $k$. The commutator then splits formally as
\begin{equation*}
[H_0(P),iA_P]\,=\,\tfrac12[(P-\db\Gamma(\xi))^2,i B_P]+[\db\Gamma(|k|),i D_1],
\end{equation*}
so we are essentially reduced to proving a Mourre estimate for both contributions separately.
For the second contribution $\db\Gamma(|k|)$, a natural choice for the conjugate $D_1$ is the second quantization of the generator of dilations in the $k$-direction,
\begin{equation}\label{eq:pre-Mourre-D}
D_1\,:=\,\db\Gamma(d_1),\qquad d_1\,:=\,\tfrac{i}{2}\left(k\cdot \nabla_k+\nabla_k\cdot k\right),
\end{equation}
which satisfies the following commutator identity,
\begin{equation}\label{eq:pre-Mourre-D-est}
[\db\Gamma(|k|),iD_1]=\db\Gamma(|k|).
\end{equation}
It remains to construct a conjugate $B_P$ for $(P-\db\Gamma(\xi))^2$, and we start by briefly underlining the difficulty and motivating our construction.

\subsubsection{Motivation for construction of $B_P$}
A natural choice is to consider the generator of dilations in the $\xi$-direction, say around some point $\xi_P\in\R^d$ to be suitably determined,
\[B'_P\,:=\,\db\Gamma(b'_P),\qquad b'_P\,:=\,\tfrac i2\Big((\xi-\xi_P)\cdot\nabla_\xi+\nabla_\xi\cdot(\xi-\xi_P)\Big).\]
Note that this can be split as
\begin{equation}\label{eq:BP'-try}
B_P'\,=\,D_2-\db\Gamma(i\xi_P\cdot\nabla_\xi),
\end{equation}
in terms of the generator of dilations in the $\xi$-direction,
\begin{equation}
D_2\,:=\,\db\Gamma(d_2),\qquad d_2\,:=\,\tfrac i2(\xi\cdot\nabla_\xi+\nabla_\xi\cdot\xi).
\end{equation}
For this choice, a direct computation yields
\begin{eqnarray*}
\tfrac12[(P-\db\Gamma(\xi))^2,iB'_P]&=&(N\xi_P-\db\Gamma(\xi))\cdot(P-\db\Gamma(\xi))\\
&=&\Big(\tfrac12(P+N\xi_P)-\db\Gamma(\xi)\Big)^2-\tfrac14(P-N\xi_P)^2,
\end{eqnarray*}
where we recall that $N=\db\Gamma(\mathds1)$ is the number operator.
Yet, whatever the choice of~$\xi_P$, this commutator has no definite sign for energies close to the eigenvalue $\frac12P^2$.
This computation is instructive as it indicates that we should actually adapt the center $\xi_P$ of the dilation to the number operator and choose $\xi_P=n^{-1}P$ on the $n$-boson state space. Instead of~\eqref{eq:BP'-try}, this leads us to rather defining the following, which is no longer a second-quantization operator due to its dependence on the number operator,
\begin{gather}\label{eq:def-try-CP'}
B_P''\,:=\,D_2-N^{-\frac12}\db\Gamma(iP\cdot\nabla_\xi)N^{-\frac12},
\end{gather}
where we implicitly define the pseudo-inverse $N^{-\frac12}:=\bar\Pi_{\Omega}N^{-\frac12}\bar\Pi_\Omega$ with some abuse of notation, recalling that $\bar\Pi_\Omega=1-\Pi_\Omega$ is the orthogonal projection on $\C\Omega^\bot$.
For this choice, the commutator becomes
\begin{eqnarray}\label{eq:commut-CP'}
\tfrac12[(P-\db\Gamma(\xi))^2,iB_P'']&=&\bar\Pi_\Omega(P-\db\Gamma(\xi))^2\bar\Pi_\Omega\nonumber\\
&=&(P-\db\Gamma(\xi))^2-P^2\Pi_\Omega,
\end{eqnarray}
which has now exactly the desired behavior: indeed, combined with~\eqref{eq:pre-Mourre-D-est}, it yields a Mourre estimate for $H_0(P)$ on the whole spectrum.

Although this might look like the end of the story, this choice $B_P''$ does actually not suit our purposes as it happens to behave badly with respect to the fiber interaction Hamiltonian $\Phi(\rho)$: a direct computation shows that the commutator $[\Phi(\rho),iB_P'']$ is not even relatively bounded when restricted to any fixed $n$-boson state space.
The problem is naturally related to the fact that $B_P''$ is not a second-quantization operator.
While we have no choice but sticking to operators that are not obtained by second quantization, we note that~$B_P''$ is actually not the only possible choice.
Indeed, on the $n$-boson state space, this operator amounts to the arithmetic average $\frac1n\sum_{j=1}^niP\cdot\nabla_{\xi_j}$ of coordinates \mbox{$\{iP\cdot\nabla_{\xi_j}\}_{1\le j\le n}$}, which could be replaced for instance by the signed maximum of coordinates.
This highly non-standard choice is directly inspired by our previous work~\cite{DS21}.
We shall show that it does essentially not change the commutator relation~\eqref{eq:commut-CP'} while behaving much better with the field operator~$\Phi(\rho)$.

\subsubsection{Signed maximum and regularization}
We turn to the suitable construction of the signed maximum of coordinates $\{iP\cdot\nabla_{\xi_j}\}_{1\le j\le n}$ on the $n$-boson state space.
First, instead of momentum representation on~$\Hf$, we shall use position representation: we denote by $y:=i\nabla_\xi$ the position coordinate dual to $\xi$, and we set $z:=P\cdot y$ for the coordinate in the $P$-direction.
For all $n\ge1$, we define the function $m_n:\R^n\to\R$ as the signed maximum of coordinates: for all $z_1,\ldots,z_n\in\R$, we set
\[m_{n}(z_1,\ldots,z_n)\,:=\,z_{j_0}\]
where the index $j_0$ is chosen such that $|z_{j_0}|=\max_j|z_j|$.
This is obviously well-defined on~$\R^n$ up to a null set. Equivalently, we can write
\begin{eqnarray}\label{eq:def-mn}
m_n(z_1,\ldots,z_n)&:=&\textstyle\big(\!\max_j|z_j|\big)\sgn r_n(z_1,\ldots,z_n),\\
r_n(z_1,\ldots,z_n)&:=&\textstyle\max_j z_j+\min_j z_j,\nonumber
\end{eqnarray}
where we take e.g.\@ the convention $\sgn(0)=0$. This function is clearly symmetric with respect to the variables $z_1,\dots,z_n$ and has the following property.

\begin{lem}\label{lem:discont-mn}
For all $n\ge1$, the function $m_n$ is continuous on $\R^n\setminus \Sc_n$, where $\Sc_n$ stands for the hypersurface
\begin{equation}\label{eq:charact-Sn}
\Sc_n\,:=\,r_n^{-1}\{0\}\,=\,\Big\{z\in\R^n\,:\,\exists j_0\neq j_1\text{ such that $z_{j_0}=-z_{j_1}$ and $|z_{j_0}|=\textstyle\max_j|z_j|$}\Big\}.
\end{equation}
In addition, there holds in the distributional sense
\begin{equation}\label{eq:comm-mp0}
\sum_{j=1}^n \partial_jm_n\,=\,1+|\!\textstyle\max_jz_j-\min_jz_j|\,\mathcal H_{\Sc_n}\,\ge\,1,
\end{equation}
where $\mathcal H_{\Sc_n}$ stands for the $(n-1)$th-dimensional Hausdorff measure on $\Sc_n$.
\end{lem}

\begin{proof}
The continuity of $m_n$ is clear outside the zero locus $\Sc_n$ of $r_n$, and we turn to the second part of the statement.
On $\R^n\setminus \Sc_n$, we have $m_n(z_1,\ldots,z_n)=z_{j_0}$ with \mbox{$|z_{j_0}|=\max_j|z_j|$}, and thus $\sum_{j}\partial_jm_n=1$.
It remains to examine the jump of $m_n$ on~$\Sc_n$.
Given a point $z:=(z_1,\dots,z_n)\in\R^n$, we may assume $z_1=\min_j z_j$ and $z_2=\max_j z_j$ up to permuting coordinates, and we
consider the line $\{z(t):=z+t(1,\ldots,1):t\in\R\}$. In view of~\eqref{eq:charact-Sn}, we note that this line intersects $\Sc_n$ at a single point: $z(t)\in \Sc_n$ if and only if $t=-\frac12(z_1+z_2)$. The jump of~$m_n$ at this point along this line is easily checked to be $|z_1-z_2|$, and the conclusion follows.
\end{proof}

Next, we regularize $m_n$ to smoothen the singular part of the derivative~\eqref{eq:comm-mp0}.
We start with the following reformulation of~$m_n$,
\[m_n(z_1,\ldots,z_n)\,=\,\textstyle\frac{1}{2}(\max_j z_j+\min_j z_j)+\frac12(\max_j z_j-\min_j z_j)\sgn(\max z_j+\min_j z_j),\]
where only the sign function needs to be regularized.
Given $\delta>0$, we choose a smooth odd function $\chi_\delta:\R\to[-1,1]$ such that
\begin{gather}
\chi_\delta|_{(-\infty,-1]}=-1,\qquad\chi_\delta|_{[1,\infty)}=1,\qquad0\le\chi_\delta'\le1+\delta~~\text{pointwise},\nonumber\\
\chi_\delta(s)\le s~~\text{for $-1\le s\le 0$},\qquad\text{and}\qquad\chi_\delta(s)\ge s~~\text{for $0\le s\le 1$},\label{eq:choice-prop-chi-delta-0}
\end{gather}
and we then define the following regularization of $m_n$,
\begin{multline}\label{eq:def-tilde-mn-delta}
\widetilde m_{n;\delta}(z_1,\ldots,z_p)\,:=\,\textstyle\frac12(\max_j z_j+\min_j z_j)\\
+\textstyle\frac12(\max_j z_j-\min_j z_j)\,\chi_\delta\Big(\frac{\max_j z_j+\min_j z_j}{1+\max_jz_j-\min_jz_j}\Big),
\end{multline}
which is obviously globally well-defined and continuous.\footnote{A similar construction was first used in our previous work~\cite{DS21}. Note however the following slight mistake in~\cite[Section~5.6]{DS21}: we forgot to add $1$ in the denominator of the argument of $\chi_\delta$, which then actually poses regularity issues at the origin in $\R^n$.}
The denominator in the argument of $\chi_\delta$ is easily understood: in order to make the derivative~\eqref{eq:comm-mp0} uniformly bounded, it is not enough to regularize the sign function in a fixed neighborhood of~$\Sc_n$ as the derivative would still produce an unbounded term due to the multiplication by $\max_jz_j-\min_jz_j$.
In view of properties of $\chi_\delta$, a direct computation yields
\begin{equation}\label{eq:prop-tilde-mn-1}
1\,\le\,\sum_{j=1}^n\partial_j\widetilde m_{n;\delta}\,\le\,2+\delta,
\end{equation}
and in addition, for all $r\ge1$,
\begin{equation}\label{eq:prop-tilde-mn-2}
\bigg|\Big(\sum_{j=1}^n\partial_j\Big)^r\widetilde m_{n;\delta}\bigg|\,\lesssim_{\chi_\delta,r}\,1,\qquad
\bigg|\Big(\sum_{j=1}^nz_j\partial_j\Big)^r\widetilde m_{n;\delta}-\widetilde m_{n;\delta}\bigg|\,\lesssim_{\chi_\delta,r}\,1.
\end{equation}
In particular, note that $\widetilde m_{n;\delta}$ is smooth in the direction $(1,\ldots,1)$.
We also establish the following property, which will be key to estimate commutators with field operators.
\begin{lem}\label{lem:m-function}
For all $n\ge0$, there holds for all $z,z_1,\dots,z_n\in\R$,
\begin{equation*}
\big|\widetilde m_{n+1;\delta}(z,z_1,\dots,z_n)-\widetilde m_{n;\delta}(z_1,\dots,z_n)\big|\,\le\,2|z|+1.\qedhere
\end{equation*}
\end{lem}

\begin{proof}
As $|\chi_\delta|\le1$, the definition of $\widetilde m_{n;\delta}$ ensures
\[|\widetilde m_{n;\delta}(z_1,\ldots,z_n)|\,\le\,\textstyle\max_j|z_j|,\]
which trivially yields the conclusion in case $\max_j|z_j|\le |z|\vee\frac12$. It thus remains to consider the case $\max_j|z_j|> |z|\vee\frac12$.
Up to permuting coordinates, we can assume $z_1=\min_jz_j$ and $z_2=\max_jz_j$.
In the case $z_1\le z\le z_2$, we simply find
\[\widetilde m_{n+1;\delta}(z,z_1,\dots,z_n)=\widetilde m_{n;\delta}(z_1,\dots,z_n),\]
and the conclusion follows.
It remains to treat the case $z_1\le z_2\le z$, while the symmetric case $z\le z_1\le z_2$ is similar.
Given $z_1\le z_2\le z$, the assumption~$\max_j|z_j|>|z|\vee\frac12$ implies~$z_1<-|z|\vee\frac12$, and thus in particular $\frac{z_2+z_1}{1+z_2-z_1}\le\frac{z+z_1}{1+z-z_1}\le0$.
In addition, there holds $\chi_\delta\big(\tfrac{y+z_1}{1+y-z_1}\big)=-1$ whenever $z_1 \le y\le-\frac12$.
Using this together with properties of $\chi_\delta$, we may then estimate
\begin{eqnarray*}
\lefteqn{\big|\widetilde m_{n+1;\delta}(z,z_1,\ldots,z_n)-\widetilde m_{n;\delta}(z_1,\ldots,z_n)\big|}\\
&=&\textstyle\frac12(z-z_2)\Big(1+\chi_\delta\big(\tfrac{z_2+z_1}{1+z_2-z_1}\big)\Big)+\frac12(z-z_1)\Big(\chi_\delta\big(\tfrac{z+z_1}{1+z-z_1}\big)-\chi_\delta\big(\tfrac{z_2+z_1}{1+z_2-z_1}\big)\Big)\\
&\le&\textstyle\frac12(z-z_2)\Big(1+\chi_\delta\big(\tfrac{z_2+z_1}{1+z_2-z_1}\big)\Big)+\frac12(z-z_1)\Big(1+\chi_\delta\big(\tfrac{z+z_1}{1+z-z_1}\big)\Big)\\
&\le&\textstyle\frac12(z-z_2)\mathds1_{z_2\ge-\frac12}+\frac12(z-z_1)\big(1+\tfrac{z+z_1}{1+z-z_1}\big)\mathds1_{z\ge-\frac12}\\
&\le&\textstyle\frac32(|z|+\tfrac12),
\end{eqnarray*}
as claimed.
\end{proof}

\subsubsection{Back to conjugate operator}
With the above construction at hand, we turn to the suitable replacement for the second term in the choice~\eqref{eq:def-try-CP'} of the conjugate operator.
For all $n\ge0$, we define the operator $M_{P,n;\delta}$ on the $n$-boson state space as the multiplication with the function
\[(k_1,y_1,\ldots,k_n,y_n)~\mapsto~\widetilde m_{n;\delta}(P\cdot y_1,\ldots,P\cdot y_n),\]
using position representation $y=i\nabla_\xi$ on $\Hf$, and we define
\[M_{P;\delta}\,=\,\bigoplus_{n=1}^\infty M_{P,n;\delta}\qquad\text{on~~$\Hc^\fu\,=\,\bigoplus_{n=0}^\infty\Gamma_s^{(n)}(\Hf)$}.\]
Coming back to~\eqref{eq:def-AP-BPR} and~\eqref{eq:def-try-CP'}, we then define the following conjugate operator,
\begin{equation}\label{eq:def-Ckst}
A_{P;\delta}\,:=\,2D_1+D_2-\tfrac{2}{3+\delta}M_{P;\delta}\qquad\text{on $\Hc^\fu$},
\end{equation}
where we recall that $D^1$ and $D^2$ stand for the generators of dilations in the $k$-direction and the $\xi$-direction, respectively,
\begin{eqnarray*}
D_1&=&\db\Gamma(d_1),\qquad d_1~=~\tfrac{i}2(k\cdot\nabla_k+\nabla_k\cdot k),\\
D_2&=&\db\Gamma(d_2),\qquad d_2~=~\tfrac{i}2(\xi\cdot\nabla_\xi+\nabla_\xi\cdot\xi).
\end{eqnarray*}
The reader may wonder why this definition of $A_{P;\delta}$ is chosen instead of $D_1+D_2-M_{P;\delta}$, which would seem more natural in view of~\eqref{eq:def-AP-BPR} and~\eqref{eq:def-try-CP'}. The basic reason is as follows: the computation of the relevant commutators involves the derivative $\sum_{j=1}^n\partial_j\widetilde m_{n;\delta}$, which in view of the regularization $\widetilde m_{n;\delta}$ of $m_n$ is not uniformly equal to $1$ but takes values in the whole interval~$[1,2+\delta]$ close to the hypersurface $\Sc_n$.
Symbols are thus deformed in the computation of commutators, and the choice $D^1+D^2-M_P$ would actually fail to provide a Mourre estimate close to the eigenvalue~$\frac12P^2$. Definition~\eqref{eq:def-Ckst} is precisely meant to overcome this issue.

By definition, the operator $A_{P;\delta}$ commutes with the number operator. Given its action on $n$-boson state space, it is clearly essentially self-adjoint on $\Cc^\fu$, and we show that it generates an explicit unitary group that preserves the domain of fiber Hamiltonians.

\begin{lem}\label{lem:C_k^st}
The operator $A_{P;\delta}$ is essentially self-adjoint on $\Cc^\fu$ and its closure generates a unitary group $\{e^{itA_{P;\delta}}\}_{t\in\R}$ on $\Hc^\fu$, which commutes with the number operator and has the following explicit action: for all $n\ge1$ and $u_n\in\Gamma_s^{(n)}(\Hf)$,
\begin{multline*}
\big(e^{itA_{P;\delta}}u_n\big)(k_1,y_1,\ldots,k_n,y_n)
\,=\,\exp\bigg(\!-\tfrac{2i}{3+\delta}\!\int_0^t\!\widetilde m_{n;\delta}(P\cdot e^sy_1,\ldots,P\cdot e^sy_n)\,ds\bigg)\\
\times e^{tn(\frac d2-q)}u_n(e^{-2t}k_1,e^ty_1,\ldots,e^{-2t}k_n,e^ty_n),
\end{multline*}
where on $\Hf=\Ld^2(\R^q\times\R^d)$ we use momentum representation in the first variable and position representation in the second.
In particular, the domain $\Dc$ of fiber Hamiltonians~\eqref{eq:dom-HP-friction} is invariant under this group action.
\end{lem}

\begin{proof}
For all $n\ge1$, consider the family $\{U^t_{P,n;\delta}\}_{t\in\R}$ of operators on $\Gamma_s^{(n)}(\Hf)$, defined by the above formula,
\begin{multline*}
(U^t_{P,n;\delta}u_n)(k_1,y_1,\ldots,k_n,y_n)
\,:=\,\exp\bigg(\!-\tfrac{2i}{3+\delta}\!\int_0^t\!\widetilde m_{n;\delta}(P\cdot e^sy_1,\ldots,P\cdot e^sy_n)\,ds\bigg)\\
\times e^{tn(\frac d2-q)}\,u_n(e^{-2t}k_1,e^ty_1,\ldots,e^{-2t}k_n,e^ty_n).
\end{multline*}
This defines a unitary group on $\Gamma_s^{(n)}(\Hf)$.
In addition, for all $u_n\in C^\infty_c(\R^q\times\R^d)^{\otimes_s n}$, we note that the following convergence holds in $\Gamma_s^{(n)}(\Hf)$,
\begin{multline*}
\lim_{t\downarrow0}\tfrac1{t}(U^t_{P,n;\delta}u_n-u_n)\\
\,=\,-\sum_{j=1}^n(k_j\cdot\nabla_{k_j}+\nabla_{k_j}\cdot k_j)u_n
+\tfrac12\sum_{j=1}^n(y_j\cdot\nabla_{y_j}+\nabla_{y_j}\cdot y_j)u_n
-\tfrac{2i}{3+\delta}M_{P,n;\delta}u_n,
\end{multline*}
where the right-hand side coincides with $iA_{P;\delta}u_n$ since we have in position representation
\[d_2\,=\,\tfrac i2(\xi\cdot\nabla_\xi+\nabla_\xi\cdot\xi)\,=\,\tfrac{1}{2i}(y\cdot\nabla_y+\nabla_y\cdot y).\]
This proves that $\{U^t_{P,n;\delta}\}_{t\in\R}$ is a unitary $C_0$-group on $\Gamma_s^{(n)}(\Hf)$ and that its self-adjoint generator coincides with $A_{P;\delta}$ on its core $C^\infty_c(\R^q\times\R^d)^{\otimes_s n}$.
The conclusion follows.
\end{proof}

Next, we show the relative boundedness of commutators of the uncoupled fiber Hamiltonian $H_0(P)$ with the above-constructed conjugate operator $A_{P;\delta}$.
Combined with Lemma~\ref{lem:C_k^st}, this actually proves Theorem~\ref{th:main/friction}(i), further noting that the $C^\infty(A_{P;\delta})$-regularity property follows by applying the sufficient criterion in Lemma~\ref{lem:suff-crit-smooth}.

\begin{lem}
For all $s\ge1$, the $s$-th iterated commutator $\ad_{iA_{P;\delta}}^s(H_0(P))$ extends as an $H_0(P)$-bounded self-adjoint operator with domain $\Dc=\Dc(H_0(P))$.
\end{lem}

\begin{proof}
For all $n\ge1$, we define the operators $M'_{P,n;\delta}$ and $M''_{P,n;\delta}$ on $\Gamma_s^{(n)}(\Hf)$ as the multiplications with the functions
\begin{eqnarray*}
(k_1,y_1,\ldots,k_n,y_n)&\mapsto& \big(\textstyle\sum_{j=1}^n\partial_j\widetilde m_{n;\delta}\big)(P\cdot y_1,\ldots,P\cdot y_n),\\
(k_1,y_1,\ldots,k_n,y_n)&\mapsto& \big(\textstyle\sum_{j,l=1}^n\partial_{j,l}\widetilde m_{n;\delta}\big)(P\cdot y_1,\ldots,P\cdot y_n),
\end{eqnarray*}
respectively, and we set
\[M_{P;\delta}^\prime:=\bigoplus_{n=1}^\infty M_{P,n;\delta}^\prime,\qquad M_{P;\delta}'':=\bigoplus_{n=1}^\infty M_{P,n;\delta}'',\qquad\text{on $\Hc^\fu=\bigoplus_{n=0}^\infty\Gamma_s^{(n)}(\Hf)$.}\]
A direct computation yields in these terms
\[[\db\Gamma(\nabla_y),M_{P;\delta}]=P M_{P;\delta}',\qquad[\db\Gamma(\nabla_y),M_{P;\delta}']=PM_{P;\delta}''.\]
By definition~\eqref{eq:def-Ckst} of $A_{P;\delta}$, recalling that $\xi=-i\nabla_y$, we compute in the sense of forms on~$\Cc^\fu$,
\begin{multline*}
[H_0(P),iA_{P;\delta}]\,=\,2\db\Gamma(|k|)-\db\Gamma(\xi)\cdot(P-\db\Gamma(\xi))\\
+\tfrac1{3+\delta}P\cdot\Big((P-\db\Gamma(\xi))M_{P;\delta}'+M_{P;\delta}'(P-\db\Gamma(\xi))\Big),
\end{multline*}
which can be reorganized as
\begin{multline}\label{eq:commut-HC}
[H_0(P),iA_{P;\delta}]\,=\,2\db\Gamma(|k|)+(P-\db\Gamma(\xi))^2\\
+\tfrac1{3+\delta}P\cdot\Big((P-\db\Gamma(\xi))(M_{P;\delta}'-\tfrac{3+\delta}2)+(M_{P;\delta}'-\tfrac{3+\delta}2)(P-\db\Gamma(\xi))\Big).
\end{multline}
Alternatively, further using $[\db\Gamma(\nabla_y),M_{P;\delta}']=PM_{P;\delta}''$, and recognizing $H_0(P)$ in the right-hand side, we get
\begin{equation}\label{eq:commut-HC-re}
[H_0(P),iA_{P;\delta}]\,=\,2H_0(P)
+\tfrac2{3+\delta}(M_{P;\delta}'-\tfrac{3+\delta}2)P\cdot(P-\db\Gamma(\xi))
+\tfrac{i}{3+\delta}P^2M_{P;\delta}''.
\end{equation}
As~\eqref{eq:prop-tilde-mn-1} yields $1\le M_{P;\delta}'\le2+\delta$
and $|M_{P;\delta}''|\lesssim_{\chi_\delta}1$, we find for all~$u\in \Cc^\fu$,
\begin{eqnarray*}
\lefteqn{\big\|\tfrac2{3+\delta}(M_{P;\delta}'-\tfrac{3+\delta}2)P\cdot(P-\db\Gamma(\xi))u
\big\|+\big\|\tfrac{i}{3+\delta}P^2M_{P;\delta}''u\big\|}\\
&\hspace{2cm}\lesssim_{\chi_\delta}&|P|\|(P-\db\Gamma(\xi))u\|+P^2\|u\|\\
&\hspace{2cm}\le&|P|\|H_0(P)^\frac12u\|+P^2\|u\|.
\end{eqnarray*}
Combined with~\eqref{eq:commut-HC-re}, this shows that $[H_0(P),iA_{P;\delta}]$ is equal to $2H_0(P)$ up to an infinitesimal perturbation. By the Kato--Rellich theorem, we deduce that the commutator $\ad_{iA_{P;\delta}}(H_0(P))=[H_0(P),iA_{P;\delta}]$ extends as an $H_0(P)$-bounded self-adjoint operator with domain $\Dc=\Dc(H_0(P))$.
Similarly computing iterated commutators and appealing to~\eqref{eq:prop-tilde-mn-2}, the conclusion easily follows; we skip the detail.
\end{proof}

\subsection{Mourre estimate}
We turn to the proof of the Mourre estimate for $H_0(P)$. This amounts to showing that the commutator identity~\eqref{eq:commut-CP'} is essentially preserved. Due to the deformation of commutators, however, we only manage to cover energy intervals above~$\frac1{18}P^2$.
Up to renaming $\e,\delta$, this proves Theorem~\ref{th:main/friction}(ii).

\begin{lem}\label{lem:Mourre-frict-0}
For all $\e>0$, the commutator $[H_0(P),iA_{P;\delta}]$ satisfies the following Mourre estimate on $J_\e:=\big[(\frac{1+\delta}{3+\delta}+\e)^2\frac12P^2,\infty\big)$,
\[\mathds1_{J_\e}(H_0(P))[H_0(P),iA_{P;\delta}]\mathds1_{J_\e}(H_0(P))\,\ge\,\e(\tfrac13+\e)P^2\mathds1_{J_\e}(H_0(P))-P^2\Pi_{\Omega}.\qedhere\]
\end{lem}

\begin{proof}
We split the proof into two steps.

\medskip
\step1 Proof that for all $\alpha>0$,
\begin{equation}\label{eq:Hk0-Mourre-pre-0}
[H_0(P),iA_{P;\delta}]\,\ge\,\bar\Pi_\Omega\Big(2(1-\tfrac{1}{2\alpha})H_0(P)
-\tfrac\alpha2\big(\tfrac{1+\delta}{3+\delta}\big)^2P^2\Big)\bar\Pi_\Omega.
\end{equation}
As the vacuum state~$\Omega$ is an eigenvector of $H_0(P)$, it belongs to the kernel of the commutator $[H_0(P),iA_{P;\delta}]$, and it suffices to establish this lower bound~\eqref{eq:Hk0-Mourre-pre-0} on $\Cc^\fu\cap\C\Omega^\bot$. Given $u\in\Cc^\fu\cap\C\Omega^\bot$, starting from identity~\eqref{eq:commut-HC}, we can bound
\begin{multline*}
\big\langle u,[H_0(P),iA_{P;\delta}]u\big\rangle\,\ge\,2\langle u,\db\Gamma(|k|)u\rangle+\langle u,(P-\db\Gamma(\xi))^2u\rangle\\
-\tfrac2{3+\delta}|P|\|(P-\db\Gamma(\xi))u\|\|(M_{P;\delta}'-\tfrac{3+\delta}2)u\|,
\end{multline*}
and thus, recalling that~\eqref{eq:prop-tilde-mn-1} implies $1\le M_{P;\delta}'\le2+\delta$, we get
\begin{equation*}
\big\langle u,[H_0(P),iA_{P;\delta}]u\big\rangle\,\ge\,2\langle u,\db\Gamma(|k|)u\rangle+\langle u,(P-\db\Gamma(\xi))^2u\rangle
-\tfrac{1+\delta}{3+\delta}|P|\|(P-\db\Gamma(\xi))u\|\|u\|.
\end{equation*}
For all $\alpha>0$, Young's inequality then yields
\begin{eqnarray*}
\big\langle u,[H_0(P),iA_{P;\delta}]u\big\rangle&\ge&2\langle u,\db\Gamma(|k|)u\rangle+(1-\tfrac{1}{2\alpha})\langle u,(P-\db\Gamma(\xi))^2u\rangle
-\tfrac\alpha2\big(\tfrac{1+\delta}{3+\delta}\big)^2P^2\|u\|^2\\
&\ge&2(1-\tfrac{1}{2\alpha})\langle u,H_0(P)u\rangle
-\tfrac\alpha2\big(\tfrac{1+\delta}{3+\delta}\big)^2P^2\|u\|^2,
\end{eqnarray*}
that is,~\eqref{eq:Hk0-Mourre-pre-0}.

\medskip
\step2 Conclusion.\\
Given $E\ge0$, applying~\eqref{eq:Hk0-Mourre-pre-0} to $\mathds1_{[E,\infty)}(H_0(P))u$, we get
\begin{multline*}
\Big\langle\mathds1_{[E,\infty)}(H_0(P)) u,[H_0(P),iA_{P;\delta}]\mathds1_{[E,\infty)}(H_0(P))u\Big\rangle\\
\,\ge\,\Big(2E(1-\tfrac{1}{2\alpha})
-\tfrac\alpha2\big(\tfrac{1+\delta}{3+\delta}\big)^2P^2\Big)\|\mathds1_{[E,\infty)}(H_0(P))u\|^2
-P^2\|\Pi_\Omega u\|^2.
\end{multline*}
Hence, optimizing with respect to $\alpha>0$,
\begin{multline*}
\Big\langle\mathds1_{[E,\infty)}(H_0(P))u,[H_0(P),iA_{P;\delta}]\mathds1_{[E,\infty)}(H_0(P))u\Big\rangle\\
\,\ge\,\sqrt{2E}\Big(\sqrt{2E}-\tfrac{1+\delta}{3+\delta}|P|\Big)\|\mathds1_{[E,\infty)}(H_0(P))u\|^2
-P^2\|\Pi_\Omega u\|^2,
\end{multline*}
and the stated Mourre estimate follows.
\end{proof}

\subsection{Regularity of the interaction}
We turn to the regularity of the fiber interaction Hamiltonian $\Phi(\rho)$ with respect to $A_{P;\delta}$, thus establishing Theorem~\ref{th:main/friction}(iii).
While this result would fail for the naïve choice~\eqref{eq:def-try-CP'} of the conjugate, it crucially requires our special definition of regularized signed maximum, and the proof builds mainly on Lemma~\ref{lem:m-function}.

\begin{lem}\label{lem:com-dgamma-bounded}
Let the interaction kernel~$\rho$ satisfy Assumption~{\rm(Reg$_\nu$)} for some~\mbox{$\nu\ge1$}.
Then, for all $0\le s\le\nu$, the $s$-th iterated commutator $\ad^s_{iA_{P;\delta}}(\Phi(\rho))$ extends as a $\db\Gamma(|k|)^{1/2}$-bounded self-adjoint operator.\qedhere
\end{lem}

\begin{proof}
We split the proof into two steps.

\medskip
\step1 Analysis of the first commutator.\\
By definition~\eqref{eq:def-Ckst} of $A_{P;\delta}$, the first commutator can be split as
\begin{equation}\label{eq:commut-Phi(rho)-APdelta-decomp}
[\Phi(\rho),iA_{P;\delta}]=2[\Phi(\rho),i\db\Gamma(d_1)]+[\Phi(\rho),i\db\Gamma(d_2)]-\tfrac{2}{3+\delta}[\Phi(\rho),iM_{P;\delta}],
\end{equation}
and thus, as the first two terms involve second-quantization operators,
\begin{equation*}
[\Phi(\rho),iA_{P;\delta}]=-2\Phi(id_1\rho)-\Phi(id_2\rho)-\tfrac{2}{3+\delta}[\Phi(\rho),iM_{P;\delta}].
\end{equation*}
Standard estimates ensure that the first two terms $\Phi(id_1\rho)$ and $\Phi(id_2\rho)$ are $\db\Gamma(|k|)^{1/2}$-bounded provided that $(1+|k|^{-1/2})d_1\rho$ and $(1+|k|^{-1/2})d_2\rho$ belong to~\mbox{$\Ld^2(\R^q\times\R^d)$}, hence in particular provided that $\rho$ satisfies Assumption~(Reg$_1$).

\medskip\noindent
It remains to estimate the commutator $[\Phi(\rho),iM_{P;\delta}]$.
As $\Phi(\rho)=a^*(\rho)+a(\rho)$, it actually suffices by symmetry to estimate $[a^*(\rho),iM_{P;\delta}]$.
We recall the standard definition of the creation operator: for all $n\ge0$ and $u_n\in\Gamma_s^{(n)}(\Hf)$,
\[\big(a^*(\rho)u_n\big)\big((k_l,\xi_l)_{1\le l\le n+1}\big)\,=\,\tfrac1{\sqrt{n+1}}\sum_{j=1}^{n+1}\rho(k_j,\xi_j)\,u_n\big((k_l,\xi_l)_{l\in\{1,\ldots, n+1\}\setminus\{j\}}\big),\]
or alternatively, using position representation $y=i\nabla_\xi$ in the second variable,
\[\big(a^*(\rho)u_n\big)\big((k_l,y_l)_{1\le l\le n+1}\big)\,=\,\tfrac1{\sqrt{n+1}}\sum_{j=1}^{n+1}\tilde\rho(k_j,y_j)\,u_n\big((k_l,y_l)_{l\in\{1,\ldots, n+1\}\setminus\{j\}}\big),\]
where $\tilde\rho$ stands for the partial inverse Fourier transform $\tilde\rho(k,y):=\int_{\R^d}e^{iy\cdot\xi}\rho(k,\xi)\,\dbar\xi$.
In these terms, the commutator with $M_{P;\delta}$ takes the explicit form
\begin{multline}\label{eq:form-commut-ast-MP}
\Big([a^*(\rho),iM_{P;\delta}]u_n\Big)\big((k_l,y_l)_{1\le l\le n+1}\big)\\
\,=\,\tfrac{1}{\sqrt{n+1}}\sum_{j=1}^{n+1}\Big(i\widetilde m_{n;\delta}\big((P\cdot y_l)_{l\in\{1,\ldots,n+1\}\setminus\{j\}}\big)-i\widetilde m_{n+1;\delta}\big((P\cdot y_l)_{1\le l\le n+1}\big)\Big)\\
\times\,\tilde\rho(k_j,y_j)\,u_n\big((k_l,y_l)_{l\in\{1,\ldots,n+1\}\setminus\{j\}}\big).
\end{multline}
Appealing to Lemma~\ref{lem:m-function} to estimate the difference between $\widetilde m_{n;\delta}$ and $\widetilde m_{n+1;\delta}$, we deduce
\begin{eqnarray*}
\big|[a^*(\rho),iM_{P;\delta}]u_n\big|\,\le\,2|P|\tilde a^*(|y\tilde\rho|)|u_n|+\tilde a^*(|\tilde\rho|)|u_n|,
\end{eqnarray*}
where we use the short-hand notation $\tilde a^*(\tilde\sigma):=a^*(\sigma)$.
Standard estimates then entail that the commutator $[a^*(\rho),iM_{P;\delta}]$
is $\db\Gamma(|k|)^{1/2}$-bounded provided that $(1+|k|^{-1/2})y\tilde\rho$ belongs to $\Ld^2(\R^q\times\R^d)$. This is equivalent to requiring that $(1+|k|^{-1/2})\nabla_\xi\rho$ belongs to $\Ld^2(\R^q\times\R^d)$, which holds in particular provided that $\rho$ satisfies~(Reg$_1$).
The conclusion follows.

\medskip
\step2 Analysis of iterated commutators.\\
As in~\eqref{eq:commut-Phi(rho)-APdelta-decomp}, we start by decomposing the commutator with $iA_{P;\delta}$ in terms of commutators with $iD_1$, $iD_2$, and $iM_{P;\delta}$. Upon iteration, we are then led to estimating products of $\ad_{iD_1}$, $\ad_{iD_2}$, and $\ad_{iM_{P;\delta}}$, applied to $\Phi(\rho)=a^*(\rho)+a(\rho)$.
In line with~\eqref{eq:form-commut-ast-MP}, we argue that such expressions are explicit and thus easily estimated.
By symmetry, as in Step~1, it suffices to consider commutators applied to $a^*(\rho)$.
Iterating~\eqref{eq:form-commut-ast-MP}, we find for all $s\ge0$,
\begin{multline}
\Big(\ad_{iM_{P;\delta}}^s(a^*(\rho))u_n\Big)\big((k_l,y_l)_{1\le l\le n+1}\big)\\
\,=\,\tfrac{1}{\sqrt{n+1}}\sum_{j=1}^{n+1}\Big(i\widetilde m_{n;\delta}\big((P\cdot y_l)_{l\in\{1,\ldots,n+1\}\setminus\{j\}}\big)-i\widetilde m_{n+1;\delta}\big((P\cdot y_l)_{1\le l\le n+1}\big)\Big)^s\\
\times\,\tilde\rho(k_j,y_j)\,u_n\big((k_l,y_l)_{l\in\{1,\ldots,n+1\}\setminus\{j\}}\big).
\end{multline}
Further taking the commutator with $iD_1=i\db\Gamma(d_1)$ and $iD_2=i\db\Gamma(d_2)$, we easily find
\begin{eqnarray}
\ad_{iD_1}\Big(\ad_{iM_{P;\delta}}^s(a^*(\rho))\Big)&=&\ad_{iM_{P;\delta}}^s\Big(a^*(k\cdot\nabla_k\rho)\Big),\label{eq:identity-ad-ad-s}\\
\ad_{iD_2}\Big(\ad_{iM_{P;\delta}}^s(a^*(\rho))\Big)&=&\ad_{iM_{P;\delta}}^s\Big(a^*\big((\xi\cdot\nabla_{\xi}-s)\rho\big)\Big)
-s\,\ad_{iR_{P;\delta}}\Big(\ad_{iM_{P;\delta}}^{s-1}\big(a^*(\rho)\big)\Big),\nonumber
\end{eqnarray}
where the operator $R_{P;\delta}$ in the last term is defined as follows: for all $n\ge1$, we set
\[\widetilde r_{n;\delta}(z_1,\ldots,z_n)\,:=\,\textstyle\sum_{j=1}^nz_j\partial_j\widetilde m_{n;\delta}-\widetilde m_{n;\delta},\]
we define the operator $R_{P,n;\delta}$ on $\Gamma_s^{(n)}(\Hf)$ as the multiplication with the function
\[(k_1,y_1,\ldots,k_n,y_n)\,\mapsto\, \widetilde r_{n;\delta}(P\cdot y_1,\ldots,P\cdot y_n),\]
and we set $R_{P;\delta}:=\bigoplus_{n=1}^\infty R_{P,n;\delta}$ on $\Hc^\fu$. 
In view of~\eqref{eq:prop-tilde-mn-2}, the function $\widetilde r_{n;\delta}$ is bounded uniformly in $n$,
and the term involving $\ad_{iR_{P;\delta}}$ in~\eqref{eq:identity-ad-ad-s} can thus be viewed as a better-behaved lower-order remainder.
Up to such a remainder, identities~\eqref{eq:identity-ad-ad-s}
show that products of $\ad_{iD^1}$, $\ad_{iD^2}$, and $\ad_{iM_{P;\delta}}$, when applied to~$a^*(\rho)$, can be reduced to powers of $\ad_{iM_{P;\delta}}$ up to transforming $\rho$. The conclusion easily follows from this observation and we skip the detail.
\end{proof}

\subsection{Consequences of Mourre estimate}
Given a total momentum $P\ne0$, we turn to the proof of Corollary~\ref{cor:friction}.
By items~(i) and~(iii) in Theorem~\ref{th:main/friction}, the sufficient criterion in Lemma~\ref{lem:suff-crit-smooth} ensures that the coupled fiber Hamiltonian $H_g(P)$ is of class $C^\infty(A_{P;\delta})$ for all~$g$.
Next, by Theorem~\ref{th:main/friction}(iii), for all $\e>0$, Lemma~\ref{lem:pert-Mourre} allows to infer that $H_g(P)$ satisfies a Mourre estimate with respect to~$A_{P;\delta}$ on the energy interval
\[\Big(\tfrac1{18}P^2+\delta+\e+\tfrac{gC_P}\e\,,\,\infty\Big),\]
for some constant $C_{P}$. 
Taking $\delta$ arbitrarily small and optimizing in $\e$, we deduce that~$H_g(P)$ satisfies a Mourre estimate on any compact subinterval of
\[J_{P,g}\,:=\,\big(\tfrac1{18}P^2+\sqrt gC_P\,,\,\infty\big).\]
Moreover, the Mourre estimate is strict outside $K_{P,g}:=\big[\frac12P^2-gC_P,\frac12P^2+gC_P\big]$.
We may then appeal to Theorem~\ref{th:Mourre-spectrum}, which states that $H_g(P)$ has no singular spectrum and at most a finite number of eigenvalues in $J_{P,g}$, and has no eigenvalue in $J_{P,g}\setminus K_{P,g}$.
In order to exclude the existence of eigenvalues in $K_{P,g}$, we appeal to Theorem~\ref{th:instab}, which states the instability of the uncoupled eigenvalue~$\frac12P^2$ provided that Fermi's condition~\eqref{eq:Fermi-cond0} holds.
Altogether, this proves item~(i) of Corollary~\ref{cor:friction}, and item~(ii) follows by further applying Theorem~\ref{th:CGH}.
It remains to make Fermi's condition~\eqref{eq:Fermi-cond0} more explicit for the model at hand, which is the purpose of the following lemma (see also~\cite[Lemma~6.7]{DFS-17}).

\begin{lem}\label{Fermi-fric}
For all $P\ne0$, we have
\begin{multline*}
\lim_{\e\downarrow0}\Big\langle\Omega\,,\,\Phi( \rho )\bar\Pi_\Omega \big( H_0(P) - \tfrac12P^2 - i\e\big)^{-1}\bar\Pi_\Omega \Phi( \rho )\Omega\Big\rangle\\
\,=\,(2\pi)^{-d}\,\pv\int_0^\infty( t- \tfrac12P^2)^{-1}\bigg(\int_{|k|\le t}\int_{\{\xi\,:\,\frac12(P-\xi)^2=t-|k|\}}\tfrac{|\rho(k,\xi)|^2}{\sqrt{(P-\xi)^2+1}}\db\mathcal H_{d-1}(\xi)\dbar k\bigg)\db t\\
+\tfrac i2(2\pi)^{1-d}\int_{|k|\le\frac12P^2}\int_{\{\xi\,:\,\frac12(P-\xi)^2=\frac12P^2-|k|\}}\tfrac{|\rho(k,\xi)|^2}{\sqrt{(P-\xi)^2+1}}\db\mathcal H_{d-1}(\xi)\dbar k,
\end{multline*}
where $\Hc_{d-1}$ stands for the $(d-1)$th-dimensional Hausdorff measure.
In particular, the imaginary part is positive if $\rho$ does not vanish.
\end{lem}

\begin{proof}
For any $\e>0$, we compute
\begin{eqnarray*}
\lefteqn{\Big\langle\Omega\,,\,\Phi( \rho )\bar\Pi_\Omega\big( H_0(P) - \tfrac12P^2 - i\e\big)^{-1}\bar\Pi_\Omega \Phi( \rho )\Omega\Big\rangle}\\
&=&\Big\langle a^*( \rho )\Omega\,,\,\big( H_0(P) - \tfrac12P^2 - i\e\big)^{-1} a^*( \rho )\Omega\Big\rangle\\
&=&\iint_{\R^q\times\R^d}|\rho(k,\xi)|^2\Big( H_0^{(1)}(P;k,\xi) - \tfrac12P^2 - i\e\Big)^{-1}\dbar k\dbar\xi,
\end{eqnarray*}
where $H_0^{(1)}(P;k,\xi):=\frac12(P-\xi)^2+|k|$ is the symbol of $H_0(P)$ on the single-boson state space.
As this symbol is Lipschitz continuous, the coarea formula yields
\begin{multline*}
\Big\langle\Omega\,,\,\Phi( \rho )\bar\Pi_\Omega \big( H_0(P) - \tfrac12P^2 - i\e\big)^{-1}\bar\Pi_\Omega \Phi( \rho )\Omega\Big\rangle\\
\,=\,(2\pi)^{-q-d}\int_0^\infty\big( t- \tfrac12P^2 - i\e\big)^{-1}\bigg(\int_{\{(k,\xi):H_0^{(1)}(P;k,\xi)=t\}}\tfrac{|\rho(k,\xi)|^2}{|\nabla_{k,\xi} H_0^{(1)}(P;k,\xi)|}\db\mathcal H_{q+d-1}(k,\xi)\bigg)\db t,
\end{multline*}
where $\mathcal H_{q+d-1}$ stands for the $(q+d-1)$th-dimensional Hausdorff measure and where we note that the integrand is summable.
As $|\nabla_{k,\xi}H_0^{(1)}(P;k,\xi)|=\sqrt{(P-\xi)^2+1}$ is non-degenerate,
the conclusion easily follows from the Plemelj formula.
\end{proof}

\section{Translation-invariant massive Nelson model}\label{sec:Nelson}
This section is devoted to the proof of our main results on the translation-invariant Nelson model with massive bosons at small coupling, cf.~\eqref{eq:Nelson-Hilbert}--\eqref{eq:Phi-rho-fiber}. We start by describing the energy-momentum spectrum for uncoupled Hamiltonians, in particular proving Lemma~\ref{lem:spectrum/Nelson}, and we establish some important properties of energy thresholds. Next, we turn to the construction of a conjugate operator in the weak-coupling regime, thus establishing Theorem~\ref{th:main/Nelson}. Note that a suitable modification of our first choice of conjugate will be needed to ensure $C^\infty$-regularity.
The modification procedure is presented in Section~\ref{sec:modif-proc} below, further building on our constructions in Section~\ref{sec:constr-fric-commut} for the quantum friction model, and we believe that it is of independent interest for massive QFT models.

\subsection{Spectrum of uncoupled Hamiltonians}
We start with the proof of Lemma~\ref{lem:spectrum/Nelson}, that is, the characterization of the spectrum of uncoupled fiber Hamiltonians.
More precisely, we establish the following result.

\begin{lem}\label{lem:Sigma0-n}
Consider the translation-invariant Nelson model with massive bosons $m>0$, cf.~\eqref{eq:Nelson-Hilbert}--\eqref{eq:Phi-rho-fiber}.
Given a total momentum $P\in\R^d$, 
the uncoupled fiber Hamiltonian $H_0(P)$ commutes with the number operator $N$ and thus splits as a direct sum~\eqref{eq:Hilbert-ac-decomp}.
There holds
\[H_0(P)\Omega\,=\,\tfrac12P^2\Omega,\]
and for all $n\ge1$ the restriction $H_0^{(n)}(P)=H_0(P)|_{\Gamma_s^{(n)}(\Hf)}$ satisfies
\begin{gather}
\sigma_{\operatorname{ac}}\big(H_0^{(n)}(P)\big)=\big[E_0^{(n)}(P),\infty\big),\label{eq:sac-H0P-n}
\qquad\sigma_{\operatorname{pp}}\big(H_0^{(n)}(P)\big)=\sigma_{\operatorname{sc}}\big(H_0^{(n)}(P)\big)=\varnothing,
\end{gather}
where the $n$-boson energy threshold $E_0^{(n)}(P)$ is given by
\begin{equation}\label{eq:def-EnP}
E_0^{(n)}(P)\,:=\,\tfrac12c(n,P)^2+\sqrt{m^2n^2+(|P|-c(n,P))^2},
\end{equation}
in terms of the unique solution $c(n,P)\in[0,1)$ of the implicit equation
\begin{equation}\label{eq:def-cnP}
c(n,P)\,=\,\frac{|P|-c(n,P)}{\sqrt{m^2n^2+(|P|-c(n,P))^2}}.\qedhere
\end{equation}
\end{lem}

\begin{proof}
In view of~\eqref{eq:Hilbert-ac-decomp}, it suffices to analyze separately the spectrum of restrictions on each $n$-boson state space.
For $n\ge1$, the restriction $H_0^{(n)}(P)$ is a multiplication operator in momentum coordinates,
with symbol
\begin{equation}\label{eq:symbolE0n}
H_0^{(n)}(P;k_1,\ldots,k_n)\,:=\,
\tfrac12\Big(P-\sum_{j=1}^nk_j\Big)^2+\sum_{j=1}^n\omega(k_j).
\end{equation}
Its spectrum is thus absolutely continuous and coincides with the essential image of this symbol,
which in this case obviously takes the form~\eqref{eq:sac-H0P-n} with
\[E_0^{(n)}(P)\,:=\,\min_{k_1,\ldots,k_n\in\R^d}H_0^{(n)}(P;k_1,\ldots,k_n).\]
A straightforward computation shows that this minimum is attained at
\begin{equation}\label{eq:min-symbol}
k_1=\ldots=k_n=k_\star(n,P)\,:=\,
\tfrac1n\big(|P|-c(n,P)\big)\tfrac{P}{|P|},
\end{equation}
where $c(n,P)\in[0,1)$ is the unique solution of equation~\eqref{eq:def-cnP}. The minimum of the symbol is thus indeed given by~\eqref{eq:def-EnP}.
\end{proof}

Next, we establish some important properties of energy thresholds.
It follows from definitions~\eqref{eq:def-EnP}--\eqref{eq:def-cnP} that $P\mapsto c(n,P)$ and $P\mapsto E_0^{(n)}(P)$ are radially symmetric for all~$n\ge1$. In addition, we find
\begin{equation}\label{eq:prop-cnP}
c(n,0)=0,\qquad\text{and}\qquad 0\,<\, c(n,P)\,<\,|P|\wedge1\qquad\text{for $P\ne0$}.
\end{equation}
Other important properties are collected in the following statement.
Item~(iii) provides a simple criterion to compare the uncoupled eigenvalue $\frac12P^2$ to energy thresholds, which proves in particular the last part of Lemma~\ref{lem:spectrum/Nelson}.

\begin{lem}
\label{lem:mono-increments}
Given boson mass $m>0$, let energy thresholds be defined in~\eqref{eq:def-EnP}--\eqref{eq:def-cnP}.
\begin{enumerate}[(i)]
\item For all $n\ge1$, we have 
\begin{equation*}
c(n,P)\uparrow 1\qquad\text{and}\qquad E_0^{(n)}(P)=|P|-\tfrac12+o(1)\qquad\text{as $|P|\uparrow \infty$.}
\end{equation*}
\item For all $P$, we have
\begin{equation*}
c(n,P)\downarrow0\qquad\text{and}\qquad E_0^{(n)}(P)=mn+o(1)\qquad\text{as $n\uparrow\infty$}.
\end{equation*}
\item For all $n\ge1$, there exists a unique value $|P^{(n)}|$ such that the following equivalence holds:
\begin{equation}\label{eq:equiv-Pn}
\tfrac12P^2\ge E_0^{(n)}(P)\qquad\Longleftrightarrow\qquad |P|\ge|P^{(n)}|.
\end{equation}
In addition, this value $|P^{(n)}|$ is increasing in $n$ and we have $|P^{(1)}|=|P_\star|>1$, where~$|P_\star|$ is the critical value defined in Lemma~\ref{lem:spectrum/Nelson}.
\smallskip\item Energy increments satisfy the following monotonicity properties,
\[\begin{array}{rlll}
\text{for all $n$}:&\quad 0\,<\,E_0^{(n+1)}(P)-E_0^{(n)}(P)~\downarrow~ 0&\quad\text{as}\quad |P|\uparrow\infty,\\
\vspace{-0.3cm}&&\\
\text{for all $P$}:&\quad 0\,<\,E_0^{(n+1)}(P)-E_0^{(n)}(P)~\uparrow~ m&\quad\text{as}\quad n\uparrow\infty.
\end{array}\qedhere\]
\end{enumerate}
\end{lem}

\begin{proof}
Items~(i) and~(ii) are direct consequences of definitions~\eqref{eq:def-EnP}--\eqref{eq:def-cnP}, so it remains to establish~(iii) and~(iv).
We split the proof into two steps.

\medskip
\step1 Proof of~(iii).\\
For all $n\ge1$, starting from~\eqref{eq:def-EnP} and differentiating in $|P|$, a direct computation yields
\begin{multline*}
\tfrac{\partial}{\partial |P|}E_0^{(n)}(P)\\
\,=\,\frac{|P|-c(n,P)}{\sqrt{m^2n^2+(|P|-c(n,P))^2}}
+\bigg(c(n,P)-\frac{|P|-c(n,P)}{\sqrt{m^2n^2+(|P|-c(n,P))^2}}\bigg)\tfrac{\partial}{\partial |P|} c(n,P),
\end{multline*}
and thus, by definition of $c(n,P)$, cf.~\eqref{eq:def-cnP},
\begin{equation}\label{eq:diff-Sigma0nP}
\tfrac{\partial}{\partial |P|}E_0^{(n)}(P)
\,=\,\frac{|P|-c(n,P)}{\sqrt{m^2n^2+(|P|-c(n,P))^2}}\,=\,c(n,P).
\end{equation}
In view of~\eqref{eq:prop-cnP}, this implies that the map $|P|\mapsto \frac12|P|^2-E_0^{(n)}(P)$ is increasing. Moreover, the asymptotic behavior in~(i) ensures that this map is unbounded as $|P|\uparrow\infty$, and we see from~\eqref{eq:def-EnP} that it takes the value $-mn<0$ at $P=0$.
This entails that for all $n\ge1$ there exists a unique value $|P^{(n)}|$ such that
\begin{equation}\label{eq:def-Pn}
\tfrac12|P^{(n)}|^2-E_0^{(n)}(P^{(n)})\,=\,0,
\end{equation}
and the claimed equivalence~\eqref{eq:equiv-Pn} then follows by monotonicity.

\medskip\noindent
It remains to check that the value $|P^{(n)}|$ is increasing in $n$.
By~\eqref{eq:def-Pn}, this follows provided that we show that energy thresholds are increasing in $n$ for any fixed $P\ne0$,
\begin{equation}\label{eq:increase-E0n-inn}
E_0^{(n+1)}(P)-E_0^{(n)}(P)\,>\,0.
\end{equation}
For that purpose, we start by noting that definitions~\eqref{eq:def-EnP}--\eqref{eq:def-cnP} yield
\[E_0^{(n)}(P)\,=\,F_P(c(n,P)),\qquad F_P(c)\,:=\,\tfrac12c^2+\tfrac{|P|}c-1.\]
As~\eqref{eq:prop-cnP} ensures $c(n,P)<|P|\wedge1$ for all $n$,
as the function $F_P$ is decreasing on $(-\infty,|P|\wedge1)$,
and as~(ii) states that $c(n,P)$ is decreasing in $n$,
the claim~\eqref{eq:increase-E0n-inn} follows.
Alternatively, this result follows from identity~\eqref{eq:form-KNk} below.

\medskip
\step2 Proof of~(iv).\\
We first investigate the behavior in $|P|$ for fixed $n$. From~\eqref{eq:diff-Sigma0nP} we deduce
\[\tfrac{\partial}{\partial|P|}\big(E_0^{(n+1)}(P)-E_0^{(n)}(P)\big)\,=\,c(n+1,P)-c(n,P),\]
which is negative in view of~(ii). In addition, the asymptotic behavior in~(i) ensures that energy increments tend to $0$ as $|P|\uparrow\infty$.

\medskip\noindent
We turn to the behavior in $n$ for fixed $P$.
As definitions~\eqref{eq:def-EnP}--\eqref{eq:def-cnP} make sense for any $n\in(0,\infty)$, we may treat $n$ as a continuous variable.
Then starting from~\eqref{eq:def-EnP} and differentiating in~$n$, we find
\begin{multline*}
\tfrac{\partial}{\partial n}E_0^{(n)}(P)\\
\,=\,\frac{m^2n}{\sqrt{m^2n^2+(|P|-c(n,P))^2}}
+\bigg(c(n,P)-\frac{|P|-c(n,P)}{\sqrt{m^2n^2+(|P|-c(n,P))^2}}\bigg)\tfrac{\partial}{\partial n} c(n,P),
\end{multline*}
and thus, by definition of $c(n,P)$, cf.~\eqref{eq:def-cnP},
\begin{equation*}
\tfrac{\partial}{\partial n}E_0^{(n)}(P)\,=\,\frac{m^2n}{\sqrt{m^2n^2+(|P|-c(n,P))^2}}.
\end{equation*}
This allows to write energy increments as
\begin{equation}\label{eq:form-KNk}
E_0^{(n+1)}(P)-E_0^{(n)}(P)\,=\,m\int_n^{n+1}\!\frac{mr}{\sqrt{m^2r^2+(|P|-c(r,P))^2}}\,\db r,
\end{equation}
which entails in particular for all $n$,
\[0\,<\,E_0^{(n+1)}(P)-E_0^{(n)}(P)\,\le\,m,\qquad\lim_{n\uparrow\infty}\big(E_0^{(n+1)}(P)-E_0^{(n)}(P)\big)= m.\]
It remains to check that energy increments are increasing in $n$.
For that purpose, we further compute
\begin{multline}\label{eq:further-diff-n}
\tfrac{\partial}{\partial n}\frac{mn}{\sqrt{m^2n^2+(|P|-c(n,P))^2}}\\
\,=\,\frac{m(|P|-c(n,P))}{(m^2n^2+(|P|-c(n,P))^2)^\frac32}\Big(|P|-c(n,P)+n\tfrac{\partial}{\partial n} c(n,P)\Big).
\end{multline}
Differentiating the definition~\eqref{eq:def-cnP} of $c(n,P)$ with respect to $n$, we find after straightforward simplifications,
\[\tfrac\partial{\partial n}c(n,P)\,=\,-\frac{m^2n^2}{(m^2n^2+(|P|-c(n,P))^2)^\frac32}\tfrac\partial{\partial n}c(n,P)-\frac{m^2n(|P|-c(n,P))}{(m^2n^2+(|P|-c(n,P))^2)^\frac32},\]
and thus,
\begin{equation}\label{eq:diff-cnP}
\tfrac\partial{\partial n}c(n,P)\,=\,-\tfrac1n(|P|-c(n,P))\bigg(1+\frac{(m^2n^2+(|P|-c(n,P))^2)^\frac32}{m^2n^2}\bigg)^{-1}.
\end{equation}
This entails
\[\tfrac\partial{\partial n}c(n,P)\,>\,-\tfrac1n(|P|-c(n,P)),\]
so that~\eqref{eq:further-diff-n} becomes
\begin{equation*}
\tfrac{\partial}{\partial n}\frac{mn}{\sqrt{m^2n^2+(|P|-c_n(P))^2}}
\,>\,0.
\end{equation*}
Combined with~\eqref{eq:form-KNk}, this proves that the map $n\mapsto E_0^{(n+1)}(P)-E_0^{(n)}(P)$ is increasing, and the conclusion follows.
\end{proof}

\subsection{A first construction of conjugate operator}\label{sec:conjug/Nelson}
We turn to the construction of a conjugate operator for the uncoupled fiber Hamiltonian $H_0(P)$.
At first, a conjugate is constructed here as a second-quantization operator, which we shall subsequently modify in Section~\ref{sec:modif-proc} to improve on its regularity properties.

In case of massive bosons, as energy thresholds satisfy $E_0^{(n)}(P)\uparrow\infty$, cf.~Lemma~\ref{lem:mono-increments}(iv), it suffices to construct a conjugate and prove a Mourre estimate separately on each energy interval
\[I_n(P)\,:=\,\big[E_0^{(n)}(P),E_0^{(n+1)}(P)\big),\]
and on this interval we only need to compute commutators on state spaces with at most~$n$ bosons.
To motivate our choice of conjugate operator in this setting, we first
recall that on the $n$-boson state space the uncoupled fiber Hamiltonian has symbol
\begin{equation}\label{eq:kinetic-energy}
H_0^{(n)}(P;k_1,\ldots,k_n)\,=\,\tfrac12\Big(P-\sum_{j=1}^nk_j\Big)^2+\sum_{j=1}^n\omega(k_j),
\end{equation}
which attains a unique minimum at
\begin{equation}\label{eq:min-symbol-re}
k_1=\ldots=k_n=k_\star(n,P),
\end{equation}
cf.~\eqref{eq:min-symbol}. By convexity, a natural choice of conjugate on $\Gamma_s^{(n)}(\Hf)$ is thus given by the generator of dilations around this minimum,
\begin{equation}\label{eq:def-APn-Gn}
\sum_{j=1}^n\tfrac{i}{2}\Big(\big(k_j-k_\star(n,P)\big)\cdot\nabla_{k_j}+\nabla_{k_j}\cdot\big(k_j-k_\star(n,P)\big)\Big)\qquad\text{on $\Gamma_s^{(n)}(\Hf)$}.
\end{equation}
We could consider the operator that coincides with this choice on $\Gamma_s^{(n)}(\Hf)$ for all $n$, but it would not be a second-quantization operator
and would thus cause major issues like the lack of regularity of the fiber interaction Hamiltonian~$\Phi(\rho)$, just as for~\eqref{eq:def-try-CP'} in the quantum friction model.
Instead, when focussing on the energy interval~$I_n(P)$, we choose to define the following $n$-dependent second-quantization operator,
\begin{equation}\label{eq:def-APn}
A_{P,n}^\circ\,:=\,\db\Gamma(a_{P,n}^\circ),\qquad a_{P,n}^\circ\,:=\,\tfrac{i}2\Big(\big(k-k_\star(n,P)\big)\cdot \nabla_k+\nabla_k\cdot\big(k-k_\star(n,P)\big)\Big).
\end{equation}
This coincides with~\eqref{eq:def-APn-Gn} on the $n$-boson state space, but for $\ell<n$ bosons it corresponds to the generator of dilations around $k_1=\ldots=k_\ell=k_\star(n,P)$. Although this choice may seem inadequate, we shall show that it is precisely compensated by the fact that the energy interval $I_n(P)$ is further away from the minimum of the symbol for~$\ell<n$ bosons. Based on this observation, a Mourre estimate with respect to $A_{P,n}^\circ$ will indeed be established on~$I_n(P)$ as long as the latter is not below the eigenvalue $\frac12P^2$.

We emphasize that the above definition~\eqref{eq:def-APn} of the conjugate operator is quite different from previous choices in the literature~\cite{Rasmussen-10,MR13}. Indeed, we consider here boson momenta \mbox{$k_j-k_\star(n,P)$} as measured in the reference frame minimizing the total kinetic energy~\eqref{eq:kinetic-energy}, while in~\cite{Rasmussen-10,MR13} the starting point is instead to consider relative boson group velocities~$\nabla_{k_j}H_0^{(n)}(P;k_1,\ldots,k_n)$.
Our new choice appears particularly adapted to the problem and makes it possible to investigate for the first time the essential spectrum above the two-boson energy threshold in the weak-coupling regime.

Before turning to the proof of a Mourre estimate, we investigate properties of the above-defined conjugate operator $A_{P,n}^\circ$.
In particular, item~(ii) states the $C^2$-regularity of the uncoupled fiber Nelson Hamiltonian $H_0(P)$.
We emphasize that this limited regularity is optimal, cf.~\cite[Section~2.2]{MR13}: it comes from the fact that $A_{P,n}^\circ$ is a dilation around a point at a nontrivial distance from the origin, which entails that the commutator $[H_0(P),A_{P,n}^\circ]$ is only $NH_0(P)^{1/2}$-bounded, hence $H_0(P)^{3/2}$-bounded, but not $H_0(P)$-bounded.
In applications, this would prohibit to use the full power of Mourre's theory: results like Theorem~\ref{th:CGH}, for instance, are not available without stronger regularity. This issue will be resolved in Section~\ref{sec:modif-proc} below by a suitable modification of the conjugate operator $A_{P,n}^\circ$.

\begin{samepage}
\begin{lem}\label{lem:limit-reg-C2}$ $
\begin{enumerate}[(i)]
\item The conjugate operator $A_{P,n}^\circ$ is essentially self-adjoint on $\Cc^\fu$ and its closure generates a unitary group that commutes with the number operator and leaves the domain~$\Dc$ of fiber Hamiltonians~\eqref{eq:dom-HP-Nelson} invariant.
\smallskip\item The fiber Hamiltonian $H_0(P)$ is of class $C^2(A_{P,n}^\circ)$.
\smallskip\item Let the interaction kernel $\rho$ belong to $H^\nu(\R^d)$ with $\langle k\rangle^\nu\nabla^\nu\rho\in\Ld^2(\R^d)$ for some $\nu\ge1$.
Then, for all $0\le s\le\nu$, the $s$-th iterated commutator $\ad_{iA_{P,n}^\circ}^s(\Phi(\rho))$ extends as an~$N^{1/2}$-bounded self-adjoint operator.
\qedhere
\end{enumerate}
\end{lem}
\end{samepage}

\begin{proof}
We start with item~(i). Clearly, $a_{P,n}^\circ$ is essentially self-adjoint on $C_c^\infty(\R^d)$, and the essential self-adjointness of $A_{P,n}^\circ$ on $\Cc^\fu$ follows. In addition, as the unitary group generated by~$a_{P,n}^\circ$ takes the explicit form
\[\big(e^{ita_{P,n}^\circ}u\big)(k)\,=\,e^{-t\frac d2}\,u\Big(e^{-t}(k-k_\star(n,P))+k_\star(n,P)\Big),\qquad u\in\Ld^2(\R^d),\]
the domain $\Dc=\Dc(H_0(P))$ is obviously invariant under $e^{itA_{P,n}^\circ}=\Gamma(e^{ita_{P,n}^\circ})$.
Next, the proof of~(ii) follows from a direct computation and is a particular case of~\cite[Proposition~2.5]{MR13}. It remains to check~(iii). As $A_{P,n}^\circ=\db\Gamma(a_{P,n}^\circ)$ is a second-quantization operator, we find
\[[\Phi(\rho),iA_{P,n}^\circ]\,=\,-\Phi(ia_{P,n}^\circ\rho).\]
As $\Phi(ia_{P,n}^\circ\rho)$ is $N^{1/2}$-bounded provided $a_{P,n}^\circ\rho\in\Ld^2(\R^d)$, and repeating the same computation for iterated commutators, the conclusion follows.
\end{proof}

\subsection{Mourre estimate}\label{sec:Mourre/Nelson}
We turn to the proof of a Mourre estimate for the uncoupled fiber Hamiltonian $H_0(P)$ with respect to the above-constructed conjugate operator $A_{P,n}^\circ$.
It requires a quite delicate computation based on fine properties of the symbol~\eqref{eq:symbolE0n} around its minimizer~\eqref{eq:min-symbol}.
Note that, surprisingly, our construction does not allow to treat the case of energy intervals $I$'s with $1<n<n_P$ in the notation below.

\begin{lem}\label{lem:Mourre-pre-Nelson}
Given a total momentum $|P|>|P_\star|$, define $n_P\ge1$ such that
\begin{equation}\label{eq:def-nP-statement-1}
\tfrac12P^2\in\big[E_0^{(n_P)}(P),E_0^{(n_P+1)}(P)\big).
\end{equation}
For all $\e>0$ and all energy intervals $I\subset\big[E_0^{(n)}(P)+\e,E_0^{(n+1)}(P)\big)$ with $n=1$ or $n\ge n_P$, the following Mourre estimate holds with respect to $A_{P,n}^\circ$ on $I$,
\begin{equation}\label{eq:Mourre-Nelson-pre}
\mathds1_I(H_0(P))[H_0(P),iA_{P,n}^\circ]\mathds1_I(H_0(P))\,\ge\,\e\bar\Pi_\Omega\mathds1_I(H_0(P))\bar\Pi_\Omega.
\end{equation}
In particular, the Mourre estimate is strict if $I$ does not contain the eigenvalue~$\frac12P^2$.
\end{lem}

\begin{proof}
Let $|P|>|P_\star|$ be fixed, define $n_P$ via~\eqref{eq:def-nP-statement-1}, and consider an energy interval
\begin{equation}\label{eq:choice-I}
I\subset\big[E_0^{(n)}(P)+\e,E_0^{(n+1)}(P)\big)
\end{equation}
for some $n\ge1$ and $\e>0$.
We split the proof into four steps.

\medskip
\step1 Proof that the Mourre estimate~\eqref{eq:Mourre-Nelson-pre} follows if we show for all $1\le\ell\le n$,
\begin{equation}\label{eq:proj-Mourre-Nelson}
\Pi_{\ell}\mathds1_I(H_0(P))[H_0(P),iA_{P,n}^\circ]\mathds1_I(H_0(P))\Pi_{\ell}\,\ge\,\e\Pi_{\ell}\mathds1_I(H_0(P))\Pi_{\ell}.
\end{equation}
where $\Pi_{\ell}$ is the orthogonal projection on the $\ell$-boson state space $\Gamma_s^{(\ell)}(\Hf)$.

\medskip\noindent
First recall that $H_0(P)$ and $A_{P,n}^\circ$ commute with the number operator, hence with each projection $\Pi_\ell$. It is thus enough to prove~\eqref{eq:proj-Mourre-Nelson} for each $\ell\ge0$.
Now, for $\ell=0$, this lower bound is trivial as $\Omega$ is an eigenvector of $H_0(P)$.
Next, as the symbol of $H_0(P)$ is bounded below by $E_0^{(\ell)}(P)$ on the $\ell$-boson state space,
and as we have $E_0^{(\ell)}(P)\ge E_0^{(n+1)}(P)$ for~$\ell>n$ in view of Lemma~\ref{lem:mono-increments}(iv), the choice~\eqref{eq:choice-I} of the energy interval $I$ entails
\begin{equation}\label{eq:Mourre-restr-mass>n}
\mathds1_I(H_0(P))\Pi_\ell\,=\,0\qquad\text{for all $\ell>n$},
\end{equation}
and the claim follows.

\medskip
\step2 Proof that for all $1\le\ell\le n$,
\begin{equation}\label{eq:commut-AH0P}
\Pi_\ell[H_0(P),iA^\circ_{P,n}]\Pi_\ell
\,\ge\,\Pi_\ell H_0(P)\Pi_\ell-H_0^{(\ell)}\big(P;k_\star^{(n)},\dots,k_\star^{(n)}\big)\,\Pi_\ell,
\end{equation}
where henceforth we set for notational simplicity
\begin{equation}\label{eq:def-kstar-abbr}
k_\star^{(n)}\,:=\,k_\star(n,P)\,=\,\tfrac1n(|P|-c(n,P))\tfrac{P}{|P|},
\end{equation}
cf.~\eqref{eq:min-symbol}, thus omitting the dependence on $P$ in the notation.

\medskip\noindent
For that purpose, we start by noting that the symbol~\eqref{eq:symbolE0n} of $H_0(P)$ on the $\ell$-boson state space can be decomposed as
\begin{multline}\label{eq:g_P decomposition}
H_0^{(\ell)}(P;k_1,\dots,k_\ell)=\tfrac{1}{2}\Big(\sum_{j=1}^\ell k_j-n k_\star^{(n)}\Big)^2+\sum_{j=1}^\ell \Big(\omega(k_j)-\omega(k_\star^{(n)})-c(n,P)\tfrac{P}{|P|}\cdot(k_j-k_\star^{(n)})\Big)\\
+H_0^{(\ell)}\big(P;k_\star^{(n)},\ldots,k_\star^{(n)}\big)-\tfrac12(n-\ell)^2|k_\star^{(n)}|^{2}.
\end{multline} 
By definition~\eqref{eq:def-APn} of the conjugate operator, writing
\[iA_{P,n}^\circ\,=\,-\db\Gamma\Big((k-k_\star^{(n)})\cdot \nabla_k+\tfrac{d}{2}\Big),\]
a direct computation then yields on $\Gamma_s^{(\ell)}(\Hf)$,
\begin{multline*}
[H_0(P),iA_{P,n}^\circ]|_{\Gamma_s^{(\ell)}(\Hf)}\,=\,\bigg[\sum_{j=1}^\ell(k_j-k_\star^{(n)})\cdot\nabla_{k_j}~,~H_0^{(\ell)}(P;k_1,\ldots,k_\ell)\bigg]\\
\,=\,\Big(\sum_{j=1}^\ell(k_j-k_\star^{(n)})\Big)\cdot \Big(\sum_{j=1}^\ell k_j-nk_\star^{(n)}\Big)+\sum_{j=1}^\ell(k_j-k_\star^{(n)})\cdot\Big(\nabla\omega(k_j)-c(n,P)\tfrac{P}{|P|}\Big).
\end{multline*}
This identity can be further reorganized as follows,
\begin{multline*}
[H_0(P),iA_{P,n}^\circ]|_{\Gamma_s^{(\ell)}(\Hf)}
\,=\,\tfrac{n-\ell}{2n}\Big(\sum_{j=1}^\ell k_j\Big)^2+\tfrac{n+\ell}{2n}\Big(\sum_{j=1}^\ell k_j-nk_\star^{(n)}\Big)^2-\tfrac12n(n-\ell)|k_\star^{(n)}|^2\\
+\sum_{j=1}^\ell\Big(\omega(k_j)-\omega(k_\star^{(n)})-c(n,P)\tfrac{P}{|P|}\cdot(k_j-k_\star^{(n)})\Big)
+\sum_{j=1}^\ell\Big(\omega(k_\star^{(n)})-\tfrac{k_j\cdot k_\star^{(n)}+m^2}{\omega(k_j)}\Big).
\end{multline*}
Recognizing the symbol $H_0^{(\ell)}(P;k_1,\ldots,k_\ell)$ in the right-hand side in form of \eqref{eq:g_P decomposition}, and noting that 
\begin{displaymath}
\omega(k_\star^{(n)})\,\ge\,\dfrac{k\cdot k_\star^{(n)}+m^{2}}{\omega(k)}\qquad\text{for all $k$},
\end{displaymath} 
we deduce
\begin{multline*}
[H_0(P),iA_{P,n}^\circ]|_{\Gamma_s^{(\ell)}(\Hf)}
\,\ge\,H_0^{(\ell)}(P;k_1,\dots,k_\ell)-H_0^{(\ell)}(P;k_\star^{(n)},\dots,k_\star^{(n)})\\
+\tfrac{n-\ell}{2n}\Big(\sum_{j=1}^\ell k_j\Big)^2+\tfrac{\ell}{2n}\Big(\sum_{j=1}^\ell k_j-nk_\star^{(n)}\Big)^2
-\tfrac12\ell(n-\ell)|k_\star^{(n)}|^2.
\end{multline*}
Finally, noting that
\begin{eqnarray*}
\lefteqn{\tfrac{n-\ell}{2n}\Big(\sum_{j=1}^\ell k_j\Big)^2+\tfrac{\ell}{2n}\Big(\sum_{j=1}^\ell k_j-nk_\star^{(n)}\Big)^2}\\
&=&\tfrac12\Big(\sum_{j=1}^\ell k_j\Big)^2-\ell k_\star^{(n)}\cdot\Big(\sum_{j=1}^\ell k_j\Big)+\tfrac12\ell n|k_\star^{(n)}|^2\\
&=&\tfrac12\Big(\sum_{j=1}^\ell k_j-\ell k_\star^{(n)}\Big)^2+\tfrac12\ell (n-\ell)|k_\star^{(n)}|^2\\
&\ge&\tfrac12\ell (n-\ell)|k_\star^{(n)}|^2,
\end{eqnarray*}
we conclude
\[[H_0(P),iA_{P,n}^\circ]|_{\Gamma_s^{(\ell)}(\Hf)}
\,\ge\,H_0^{(\ell)}(P;k_1,\dots,k_\ell)-H_0^{(\ell)}(P;k_\star^{(n)},\dots,k_\star^{(n)}),\]
that is,~\eqref{eq:commut-AH0P}.

\medskip
\step3
Proof that, given an energy interval $I$ as in~\eqref{eq:choice-I}, we have for all $1\le\ell\le n$,
\begin{multline}\label{eq:commut-AH0P-re}
\Pi_\ell\mathds1_I(H_0(P))[H_0(P),iA_{P,n}^\circ]\mathds1_I(H_0(P))\Pi_\ell\\
\,\ge\,\e\Pi_\ell+\big(1-\tfrac\ell n\big)\Big(\tfrac1{2n}(|P|-c(n,P))^2-\alpha(n,P)\Big)\Pi_\ell,
\end{multline}
where $\alpha(n,P)$ stands for the positive distance between the eigenvalue $\frac12P^2$ and the energy threshold below $I$,
\begin{equation}\label{eq:def-alpha-nP}
\alpha(n,P)\,:=\,\tfrac12P^2-E_0^{(n)}(P)\,=\,\tfrac12P^2-H_0^{(n)}(P;k_\star^{(n)},\ldots,k_\star^{(n)}).
\end{equation}
Starting from~\eqref{eq:commut-AH0P},
and recalling that the choice~\eqref{eq:choice-I} of $I$ yields
\[\inf I\,\ge\,\e+E^{(n)}_0(P)\,=\,\e+H_0^{(n)}(P;k_\star^{(n)},\ldots,k_\star^{(n)}),\]
we are led to
\begin{multline*}
\Pi_\ell\mathds1_I(H_0(P))[H_0(P),iA^\circ_{P,n}]\mathds1_I(H_0(P))\Pi_\ell\\
\,\ge\,\Big(\e
+H_0^{(n)}(P;k_\star^{(n)},\ldots,k_\star^{(n)})-H_0^{(\ell)}(P;k_\star^{(n)},\ldots,k_\star^{(n)})\Big)\,\Pi_\ell.
\end{multline*}
To prove~\eqref{eq:commut-AH0P-re}, it thus remains to check for all $\ell\le n$,
\begin{multline}\label{eq:commut-AH0P-re0}
H_0^{(n)}(P;k_\star^{(n)},\ldots,k_\star^{(n)})-H_0^{(\ell)}(P;k_\star^{(n)},\ldots,k_\star^{(n)})\\
\,\ge\,\big(1-\tfrac\ell n\big)\Big(\tfrac1{2n}(|P|-c(n,P))^2-\alpha(n,P)\Big).
\end{multline}
For that purpose, we decompose
\begin{eqnarray*}
\lefteqn{H_0^{(\ell)}(P;k_\star^{(n)},\ldots,k_\star^{(n)})}\\
&=&\tfrac12(P-\ell k_\star^{(n)})^2+\ell\omega(k_\star^{(n)})\\
&=&\tfrac12P^2+\tfrac12\ell^2|k_\star^{(n)}|^2-\ell P\cdot k_\star^{(n)}+\ell\omega(k_\star^{(n)})\\
&=&\tfrac\ell n\Big(\tfrac12P^2+\tfrac12n^2|k_\star^{(n)}|^2-nP\cdot k_\star^{(n)}+n\omega(k_\star^{(n)})\Big)+\tfrac12\big(1-\tfrac\ell n\big)P^2-\tfrac12\ell(n-\ell)|k_\star^{(n)}|^2\\
&=&\tfrac\ell nH_0^{(n)}(P;k_\star^{(n)},\ldots,k_\star^{(n)})+\tfrac12\big(1-\tfrac\ell n\big)P^2-\tfrac12\ell n\big(1-\tfrac\ell n\big)|k_\star^{(n)}|^2.
\end{eqnarray*}
In terms of~\eqref{eq:def-alpha-nP}, this yields
\begin{eqnarray*}
\lefteqn{H_0^{(n)}(P;k_\star^{(n)},\ldots,k_\star^{(n)})-H_0^{(\ell)}(P;k_\star^{(n)},\ldots,k_\star^{(n)})}\\
&=&\big(1-\tfrac\ell n\big)H_0^{(n)}(P;k_\star^{(n)},\ldots,k_\star^{(n)})-\tfrac12\big(1-\tfrac\ell n\big)P^2+\tfrac12\ell n\big(1-\tfrac\ell n\big)|k_\star^{(n)}|^2\\
&=&\big(1-\tfrac\ell n\big)\Big(\tfrac12\ell n|k_\star^{(n)}|^2-\alpha(n,P)\Big).
\end{eqnarray*}
Now recalling $|k_\star^{(n)}|=\frac1n(|P|-c(n,P))$, cf.~\eqref{eq:def-kstar-abbr}, the claim~\eqref{eq:commut-AH0P-re0} follows.

\medskip
\step4 Conclusion.\\
As $n_P$ is defined via~\eqref{eq:def-nP-statement-1}, the definition~\eqref{eq:def-alpha-nP} of $\alpha$ yields $\alpha(n,P)\le0$ for~$n>n_P$.
The right-hand side in~\eqref{eq:commut-AH0P-re} is thus bounded below by $\e\Pi_\ell$ if $\ell=n$, or if $\ell<n$ and~$n>n_P$. It remains to prove the corresponding result in the case $\ell<n=n_P$. In other words, it remains to prove the following implication, for all $n,P$,
\begin{equation}\label{eq:bnd-alphanP}
\tfrac12|P|^2\in\big[E_0^{(n)}(P),E_0^{(n+1)}(P)\big)\quad\Longrightarrow\quad\tfrac1{2n}(|P|-c(n,P))^2-\alpha(n,P)\,\ge\,0.
\end{equation}
For that purpose, we start by noting that the definition~\eqref{eq:def-cnP} of $c(n,P)$ yields
\begin{equation}\label{eq:|P|-c_n(P)}
|P|-c(n,P)\,=\,mn\frac{c(n,P)}{\sqrt{1-c(n,P)^2}},
\end{equation}
which allows to rewrite~\eqref{eq:def-EnP} in particular as
\begin{equation*}
E_0^{(n)}(P)\,=\,\tfrac12c(n,P)^2+\frac{mn}{\sqrt{1-c(n,P)^2}},
\end{equation*}
and thus
\begin{equation}\label{eq:ident-alphanP}
\alpha(n,P)\,=\,\tfrac12 P^2-\tfrac12 c(n,P)^2-\frac{mn}{\sqrt{1-c(n,P)^2}}.
\end{equation}
Further inserting~\eqref{eq:|P|-c_n(P)} in this last identity to eliminate $|P|$, we get
\begin{eqnarray}
\alpha(n,P)
&=&\tfrac12c(n,P)^2\bigg(1+\frac{mn}{\sqrt{1-c(n,P)^2}}\bigg)^2-\tfrac12 c(n,P)^2-\frac{mn}{\sqrt{1-c(n,P)^2}}\nonumber\\
&=&\tfrac12m^2n^2\frac{c(n,P)^2}{1-c(n,P)^2}-mn\sqrt{1-c(n,P)^2}.\label{eq:ident-alphanP-re}
\end{eqnarray}
Combining this with~\eqref{eq:|P|-c_n(P)} again to reformulate the quantity of interest in~\eqref{eq:bnd-alphanP}, we find
\begin{equation*}
\tfrac1{2n}(|P|-c(n,P))^2-\alpha(n,P)\,=\,-\tfrac12m^2n(n-1)\frac{c(n,P)^2}{1-c(n,P)^2}+mn\sqrt{1-c(n,P)^2},
\end{equation*}
which entails that the implication~\eqref{eq:bnd-alphanP} is actually equivalent to
\begin{equation*}
\tfrac12P^2\in\big[E_0^{(n)}(P),E_0^{(n+1)}(P)\big)\quad\Longrightarrow\quad\frac{(1-c(n,P)^2)^\frac32}{c(n,P)^2}\,\ge\,\tfrac12m(n-1).
\end{equation*}
As the map $|P|\mapsto c(n,P)$ is increasing, cf.~Lemma~\ref{lem:mono-increments}(i), it suffices to prove this implication for $P$ such that $\tfrac12P^2=E_0^{(n+1)}(P)$, that is, for $|P|=|P^{(n+1)}|$, cf.~Lemma~\ref{lem:mono-increments}(iii). We are thus reduced to proving for all $n$,
\begin{equation}\label{eq:bnd-alphanP-red}
\frac{(1-c(n,P^{(n+1)})^2)^\frac32}{c(n,P^{(n+1)})^2}\,\ge\,\tfrac12m(n-1).
\end{equation}
As by definition $\alpha(n+1,P^{(n+1)})=\frac12|P^{(n+1)}|^2-E_0^{(n+1)}(P^{(n+1)})=0$, identity~\eqref{eq:ident-alphanP-re} entails
\begin{equation}\label{eq:ident-vanish-alpha}
\tfrac12m(n+1)=\frac{\big(1-c(n+1,P^{(n+1)})^2\big)^\frac32}{c(n+1,P^{(n+1)})^2}.
\end{equation}
The claim~\eqref{eq:bnd-alphanP-red} can thus be reformulated as
\begin{equation}\label{eq:bnd-alphanP-red2}
\frac{\big(1-c(n+1,P^{(n+1)})^2\big)^\frac32}{c(n+1,P^{(n+1)})^2}-\frac{\big(1-c(n,P^{(n+1)})^2\big)^\frac32}{c(n,P^{(n+1)})^2}\,\le\,m,
\end{equation}
and it directly follows in this form from Lemma~\ref{lemma:derivative} below.
\end{proof}

The conclusion of the above proof of the Mourre estimate relies on the following key computation, which we state as a separate lemma for convenience.

\begin{lem}\label{lemma:derivative}
For all $n\ge1$, the function defined by
\[f_{n+1}(r)\,:=\,\frac{\big(1-c(r,P^{(n+1)})^2\big)^\frac32}{c(r,P^{(n+1)})^2}\]
satisfies $f_{n+1}'(r)\le m$ for all $0<r\le n+1$.
\end{lem}

\begin{proof}
We split the proof into two steps.

\medskip
\step1 Proof that it suffices to show
\begin{equation}\label{eq:reduce-pr-f'-bnd}
f_{n+1}(r)\,\le\,\tfrac12 rm\qquad\text{for all $0<r\le n+1$}.
\end{equation}
For notational simplicity, we set $c(r):=c(r,P^{(n+1)})$. The derivative of $f_{n+1}$ takes the form
\begin{eqnarray}
f_{n+1}'(r)&=&-\tfrac1{c(r)^4}\Big(3c(r)^3c'(r)\sqrt{1-c(r)^2}+2c(r)c'(r)(1-c(r)^2)^\frac32\Big)\nonumber\\
&=&-c'(r)\frac{c(r)^2+2}{c(r)^3}\sqrt{1-c(r)^2}.\label{eq:first-rewr-f'}
\end{eqnarray}
Recall that the derivative of $c$ was computed in~\eqref{eq:diff-cnP},
\begin{equation*}
c'(r)\,=\,-\tfrac1r(|P^{(n+1)}|-c(r))\bigg(1+\frac{(m^2r^2+(|P^{(n+1)}|-c(r))^2)^\frac32}{m^2r^2}\bigg)^{-1},
\end{equation*}
and thus, using~\eqref{eq:|P|-c_n(P)} in form of
\begin{equation*}
|P^{(n+1)}|-c(r)\,=\,mr\frac{c(r)}{\sqrt{1-c(r)^2}},
\end{equation*}
we find
\begin{eqnarray*}
c'(r)&=&-\frac{mc(r)}{\sqrt{1-c(r)^2}}\bigg(1+\frac{mr}{(1-c(r)^2)^\frac32}\bigg)^{-1}\\
&=&-\frac{mc(r)(1-c(r)^2)}{mr+(1-c(r)^2)^\frac32}.
\end{eqnarray*}
Inserting this into~\eqref{eq:first-rewr-f'}, we get
\begin{eqnarray*}
f_{n+1}'(r)
&=&\frac{m(1-c(r)^2)^\frac32}{mr+(1-c(r)^2)^\frac32}\frac{c(r)^2+2}{c(r)^2},
\end{eqnarray*}
and we deduce the following equivalence: for all $r$,
\begin{eqnarray*}
f_{n+1}'(r)\le m&\Longleftrightarrow&\frac{(1-c(r)^2)^\frac32}{mr+(1-c(r)^2)^\frac32}\frac{c(r)^2+2}{c(r)^2}\le 1\\
&\Longleftrightarrow&2(1-c(r)^2)^\frac32\le mrc(r)^2\\
&\Longleftrightarrow&f_{n+1}(r)\le \tfrac12 mr,
\end{eqnarray*}
as claimed.

\medskip
\step2 Conclusion.\\
Let $0<r\le n+1$ be fixed. As the map $|P|\mapsto c(r,P)$ is increasing and as $|P^{(n+1)}|\ge|P^{(r)}|$ in view of Lemma~\ref{lem:mono-increments}(iii) (where we can extend the definition of $|P^{(r)}|$ to all real $r>0$), we find
\begin{equation}
f_{n+1}(r)\,=\,\frac{(1-c(r,P^{(n+1)})^2)^{\frac32}}{c(r,P^{(n+1)})^2}\,\le\,\frac{(1-c(r,P^{(r)})^2)^{\frac32}}{c(r,P^{(r)})^2}.
\end{equation}
Noting that the same argument as for~\eqref{eq:ident-vanish-alpha} yields
\begin{equation*}
\frac{(1-c(r,P^{(r)})^2)^\frac32}{c(r,P^{(r)})^2}\,=\,\tfrac12mr,
\end{equation*}
we deduce $f_{n+1}(r)\le\frac12mr$, and the conclusion follows from Step~1.
\end{proof}

\subsection{Modification procedure and improved regularity}\label{sec:modif-proc}
This section is devoted to the modification of the conjugate operator $A_{P,n}^\circ$ in order to improve on the associated regularity properties in~Lemma~\ref{lem:limit-reg-C2}(ii), in view of the proof of Theorem~\ref{th:main/Nelson}.
More precisely, we shall modify $A_{P,n}^\circ$ on $\ell$-boson state spaces for all $\ell>n$, while keeping it unchanged elsewhere. In view of~\eqref{eq:Mourre-restr-mass>n}, we recall that a Mourre estimate on the energy interval $I_n(P)$ only needs to be checked on $\ell$-boson state spaces for all $1\le\ell\le n$, hence our modification will not impact its validity.
We emphasize that our modification procedure is quite general and may be of independent interest for massive QFT models.

\subsubsection{Motivation for modification procedure}
We start by further examining the above-defined conjugate operator $A_{P,n}^\circ$, decomposing it as
\begin{equation}\label{eq:decomp-APncirc}
A_{P,n}^\circ\,=\,D_\circ-\db\Gamma\big(ik_\star(n,P)\cdot\nabla_k\big),
\end{equation}
where $D_\circ$ stands for the generator of dilations,
\[D_\circ\,:=\,\db\Gamma(d_\circ),\qquad d_\circ\,:=\,\tfrac{i}{2}(k\cdot\nabla_k+\nabla_k\cdot k).\]
The lack of regularity of $H_0(P)$ with respect to $A_{P,n}^\circ$ precisely originates from the second term $\db\Gamma\big(ik_\star(n,P)\cdot\nabla_k\big)$ in~\eqref{eq:decomp-APncirc}, as indeed its commutator with $(P-\db\Gamma(k))^2$ is not $H_0(P)$-bounded.
To cure this issue, we might naïvely want to rather consider the truncated operator
\begin{equation}\label{eq:decomp-APncirc-tronc}
A_{P,n}'\,:=\,D_\circ-\Pi_{\le n}\db\Gamma\big(ik_\star(n,P)\cdot \nabla_k\big)\Pi_{\le n},
\end{equation}
in terms of the orthogonal projection $\Pi_{\le n}$ onto $\bigoplus_{\ell=0}^n\Gamma_s^{(\ell)}(\Hf)$.
By definition, $H_0(P)$ is now of class $C^\infty(A_{P,n}')$.
In addition, as this operator coincides with~$A^\circ_{P,n}$ on the range of~$\Pi_{\le n}$, we deduce that $H_0(P)$ satisfies the same Mourre estimate with respect to~$A_{P,n}'$ as in Lemma~\ref{lem:Mourre-pre-Nelson}.

This is however not the end of the story: the brutal truncation in~\eqref{eq:decomp-APncirc-tronc} happens to behave badly with respect to the fiber Hamiltonian $\Phi(\rho)$, in link with the fact that $A_{P,n}'$ is no longer a second-quantization operator.
The truncation thus needs to be suitably complemented on $\ell$-boson state spaces for $\ell>n$,
although not by means of second quantization.
In the spirit of our constructions in Section~\ref{sec:constr-fric-commut} for the quantum friction model, instead of considering the second quantization $\db\Gamma\big(ik_\star(n,P)\cdot \nabla_k\big)$ in~\eqref{eq:decomp-APncirc}, which amounts to taking sums of coordinates $\{ik_\star(n,P)\cdot \nabla_{k_j}\}_j$, and instead of taking a brutal truncation as in~\eqref{eq:decomp-APncirc-tronc} on $\ell$-boson state spaces with $\ell>n$, we shall consider partial sums of the $n$ largest signed values of the coordinates.

\subsubsection{Partial sums of largest signed values and regularization}
Instead of momentum representation on $\Hf$, we shall use position representation: we denote by $y:=i\nabla_\xi$ the position coordinate, and we set $z:=\frac P{|P|}\cdot y$ for the coordinate in the $P$-direction. For all $1\le j\le \ell$, we define the function $m_{j,\ell}:\R^\ell\to\R$ as the $j$th largest signed value of the entries: for all $z_1,\ldots,z_\ell\in\R$, we set
\[m_{j,\ell}(z_1,\ldots,z_\ell)\,:=\,z_{i_0}\]
where the index $i_0$ is chosen such that $|z_{i_0}|$ is the $j$th largest value among $|z_1|,\ldots,|z_\ell|$.
This is obviously well-defined on $\R^\ell$ up to a null set. Note that for $j=1$ the function $m_{1,\ell}$ coincides with the signed maximum $m_\ell$ defined in~\eqref{eq:def-mn}.
For all~$1\le j<\ell$, we then define the function $s_{j,\ell}:\R^\ell\to\R$ as the sum of the $j$ largest signed entries,
\[s_{j,\ell}\,:=\,m_{1,\ell}+\ldots+m_{j,\ell},\]
and for $j\ge\ell$ we simply define
\[s_{j,\ell}(z_1,\ldots,z_\ell)\,:=\,z_1+\ldots+z_\ell.\]
As in Lemma~\ref{lem:discont-mn}, we note that $s_{j,\ell}$ is not continuous for $j<\ell$ and thus needs to be regularized,
which we shall carefully perform in the spirit of~\eqref{eq:def-tilde-mn-delta}.
For that purpose, we start by defining for all $1\le j\le \ell$ the functions $\max_{j,\ell}:\R^\ell\to\R$ and $\min_{j,\ell}:\R^\ell\to\R$ as the $j$th largest and the $j$th smallest entries, respectively: more precisely, these functions are defined to be symmetric upon permutation of their entries, and to satisfy $\max_{j,\ell}(z_1,\ldots,z_\ell)=z_{\ell-j+1}$ and $\min_{j,\ell}(z_1,\ldots,z_\ell)=z_{j}$ if $z_1\le\ldots\le z_\ell$.
These functions are both obviously well-defined and continuous, and we have the relations
\begin{gather*}
|m_{j,\ell}(z_1,\ldots,z_\ell)|\,=\,\textstyle\max_{j,\ell}(|z_1|,\ldots,|z_\ell|),\\
\textstyle\max_{1,\ell}(z_1,\ldots,z_\ell)\,=\,\displaystyle\max_{1\le l\le \ell}z_l\,=\,\textstyle\min_{\ell,\ell}(z_1,\ldots,z_\ell),\\\textstyle\min_{1,\ell}(z_1,\ldots,z_\ell)\,=\,\displaystyle\min_{1\le l\le \ell}z_l\,=\,\textstyle\max_{\ell,\ell}(z_1,\ldots,z_\ell).
\end{gather*}
In these terms, for all $1\le j<\ell$, we note that the definition of $s_{j,\ell}$ can be reformulated as
\begin{equation*}
s_{j,\ell}=\sum_{l=1}^j\!\bigg(\textstyle\tfrac12\big(\!\max_{j+1-l,\ell}+\min_{l,\ell}\!\big)
+\textstyle\tfrac12\big(\!\max_{j+1-l,\ell}-\min_{l,\ell}\!\big)\!\sgn\!\big(\!\max_{j+1-l,\ell}+\min_{l,\ell}\!\big)\bigg),
\end{equation*}
where now only the sign functions need to be regularized.
Given $\delta>0$, we choose a smooth odd function $\chi_\delta:\R\to[-1,1]$ as in~\eqref{eq:choice-prop-chi-delta-0}, and we define for all $1\le j<\ell$,
\begin{equation*}
\widetilde s_{j,\ell;\delta}\,=\,\sum_{l=1}^j\bigg(\textstyle\tfrac12\big(\!\max_{j+1-l,\ell}+\min_{l,\ell}\!\big)
+\tfrac12\big(\!\max_{j+1-l,\ell}-\min_{l,\ell}\!\big)\chi_\delta\Big(\tfrac{\max_{j+1-l,\ell}+\min_{l,\ell}}{1+\max_{j+1-l,\ell}-\min_{l,\ell}}\Big)\bigg),
\end{equation*}
which is obviously globally well-defined and continuous.
For $j\ge\ell$, no regularization is needed and we simply set
\[\widetilde s_{j,\ell;\delta}(z_1,\ldots,z_\ell)\,:=\,s_{j,\ell}(z_1,\ldots,z_\ell)\,=\,z_1+\ldots+z_\ell.\]
In view of properties of $\chi_\delta$, a direct computation yields
\begin{equation}\label{eq:prop-tilde-sell-1}
1\,\le\,\sum_{l=1}^\ell\partial_l\widetilde s_{j,\ell;\delta}\,\le\,(2+\delta)(j\wedge\ell),
\end{equation}
and in addition, for all $r\ge1$,
\begin{equation}\label{eq:prop-tilde-sell-2}
\bigg|\Big(\sum_{l=1}^\ell\partial_l\Big)^r\widetilde s_{j,\ell;\delta}\bigg|\,\lesssim_{\chi_\delta,r}\,j\wedge\ell,\qquad
\bigg|\Big(\sum_{l=1}^\ell z_l\partial_l\Big)^r\widetilde s_{j,\ell;\delta}-\widetilde s_{j,\ell;\delta}\bigg|\,\lesssim_{\chi_\delta,r}\,j\wedge\ell.
\end{equation}
In particular, note that $\widetilde s_{j,\ell;\delta}$ is smooth in the direction $(1,\ldots,1)$.
Next, we state the following generalization of Lemma~\ref{lem:m-function} for $\widetilde m_{\ell;\delta}=\widetilde s_{1,\ell;\delta}$, which will be key to estimate commutators with field operators. The proof is a direct adaptation of that of Lemma~\ref{lem:m-function} and we skip the detail.

\begin{lem}\label{lem:S-function}
For all $j,\ell\ge1$, there holds for all $z,z_1,\dots,z_\ell\in\R$,
\begin{equation*}
\big|\widetilde s_{j,{\ell+1};\delta}(z,z_1,\dots,z_\ell)-\widetilde s_{j,\ell;\delta}(z_1,\dots,z_\ell)\big|\,\le\, 2|z|+1.\qedhere
\end{equation*}
\end{lem}

\subsubsection{Back to conjugate operator}
With the above construction at hand, we turn to the suitable replacement for the second term in the choice~\eqref{eq:decomp-APncirc} of the conjugate operator. For all $n,\ell$, we define the operator $S_{P,n,\ell;\delta}$ on the $\ell$-boson state space as the multiplication with the function
\begin{equation}\label{eq:defin-SPnelldelta}
(y_1,\ldots,y_\ell)~\mapsto~|k_\star(n,P)|\,\widetilde s_{n,\ell;\delta}\big(\tfrac{P}{|P|}\cdot y_1,\ldots,\tfrac{P}{|P|}\cdot y_\ell\big),
\end{equation}
using position representation $y=i\nabla_k$ on $\Hf$, and we define
\begin{equation}\label{eq:defin-SPndelta}
S_{P,n;\delta}\,:=\,\bigoplus_{\ell=1}^\infty S_{P,n,\ell;\delta}\qquad\text{on~~$\Hc^\fu:=\bigoplus_{\ell=0}^\infty \Gamma_s^{(\ell)}(\Hf)$}.
\end{equation}
Coming back to~\eqref{eq:decomp-APncirc}, we then define the following conjugate operator,
\begin{equation}\label{eq:def-APn-delta-Nelson}
A_{P,n;\delta}\,:=\,D_\circ-S_{P,n;\delta}\qquad\text{on $\Hc^\fu$},
\end{equation}
where we recall that $D_\circ$ stands for the generator of dilations,
\[D_\circ\,=\,\db\Gamma(d_\circ),\qquad d_\circ\,=\,\tfrac i2(k\cdot\nabla_k+\nabla_k\cdot k).\]
By definition, the operator $A_{P,n;\delta}$ commutes with the number operator.
Given its action on $\ell$-boson state space, it is clearly essentially self-adjoint on $\Cc^\fu$, and we state that it generates an explicit unitary group that preserves the domain of fiber Hamiltonians.
The proof is analogous to that of Lemma~\ref{lem:C_k^st} and is skipped for brevity.

\begin{lem}\label{lem:C_k^st-Nelson}
The operator $A_{P,n;\delta}$ is essentially self-adjoint on $\Cc^\fu$ and its closure generates a unitary group $\{e^{itA_{P,n;\delta}}\}_{t\in\R}$ on $\Hc^\fu$, which commutes with the number operator and has the following explicit action: for all~$\ell\ge1$ and $u_\ell\in\Gamma_s^{(\ell)}(\Hf)$,
\begin{multline*}
\big(e^{itA_{P,n;\delta}}u_\ell\big)(y_1,\ldots,y_\ell)
\,=\,\exp\bigg(\!-i|k_\star(n,P)|\int_0^t\! \widetilde s_{n,\ell;\delta}\big(\tfrac{P}{|P|}\cdot e^sy_1,\ldots,\tfrac{P}{|P|}\cdot e^sy_\ell\big)\,ds\bigg)\\
\times e^{t\ell\frac d2}u_\ell(e^ty_1,\ldots,e^ty_\ell),
\end{multline*}
where we use position representation on $\Hf$.
In particular, the domain~$\Dc$ of fiber Hamiltonians~\eqref{eq:dom-HP-Nelson} is invariant under this group action.
\end{lem}

\subsubsection{Improved regularity}
As by definition $A_{P,n;\delta}$ coincides with $A_{P,n}^\circ$ on $\ell$-boson state spaces for all $\ell\le n$, it follows from~\eqref{eq:Mourre-restr-mass>n} that the Mourre estimate of Lemma~\ref{lem:Mourre-pre-Nelson} holds in the exact same form with respect to $A_{P,n;\delta}$, thus proving Theorem~\ref{th:main/Nelson}(ii).
Next, we show that the fiber Hamiltonian $H_0(P)$ is now of class $C^\infty(A_{P,n;\delta})$, which improves on the limited $C^2$-regularity available with respect to $A_{P,n}^\circ$, cf.~Lemma~\ref{lem:limit-reg-C2}(ii).
Combined with Lemma~\ref{lem:C_k^st-Nelson}, this proves Theorem~\ref{th:main/friction}(i), further noting that the $C^\infty(A_{P,n;\delta})$-regularity property indeed follows by applying the sufficient criterion in Lemma~\ref{lem:suff-crit-smooth}.

\begin{lem}\label{lem:Cinfty-H0P-A-Nelson}
For all $s\ge1$, the $s$-th iterated commutator $\ad_{iA_{P,n;\delta}}^s(H_0(P))$ extends as an $H_0(P)$-bounded self-adjoint operator.
\end{lem}

\begin{proof}
By definition~\eqref{eq:def-APn-delta-Nelson} of $A_{P,n;\delta}$, the first commutator can be split as
\begin{eqnarray*}
[H_0(P),iA_{P,n;\delta}]&=&[H_0(P),iD_\circ]-[H_0(P),iS_{P,n;\delta}]\\
&=&-\db\Gamma(k)\cdot(P-\db\Gamma(k))+\db\Gamma(k\cdot\nabla\omega(k))-[H_0(P),iS_{P,n;\delta}].
\end{eqnarray*}
Using that $|k\cdot\nabla\omega(k)|=\omega(k)-m^2\omega(k)^{-1}\le\omega(k)$,
the first two right-hand side terms are obviously $H_0(P)$-bounded operators:
we get for all $u,v\in\Cc^\fu$,
\begin{equation}\label{eq:H0P-commut-Nelson-1}
|\langle u,[H_0(P),iA_{P,n;\delta}]v\rangle|
\,\lesssim\,|P|^2\|u\|\|v\|+\|u\|\|H_0(P)v\|+|\langle u,[H_0(P),iS_{P,n;\delta}]v\rangle|,
\end{equation}
and it remains to estimate the term.
For that purpose, using position representation $y=i\nabla_\xi$ on $\Hf$, we write
\begin{eqnarray*}
[H_0(P),iS_{P,n;\delta}]&=&\tfrac12[(P+\db\Gamma(i\nabla_y))^2,iS_{P,n;\delta}]+[\db\Gamma(\omega(i\nabla_y)),iS_{P,n;\delta}]\\
&=&-\tfrac12\db\Gamma(\nabla_y)\cdot[\db\Gamma(\nabla_y),iS_{P,n;\delta}]-\tfrac12[\db\Gamma(\nabla_y),iS_{P,n;\delta}]\cdot\db\Gamma(\nabla_y)\\
&&-P\cdot[\db\Gamma(\nabla_y),S_{P,n;\delta}]
+[\db\Gamma(\omega(i\nabla_y)),iS_{P,n;\delta}],
\end{eqnarray*}
or alternatively,
\begin{multline*}
[H_0(P),iS_{P,n;\delta}]
\,=\,-[\db\Gamma(\nabla_y),S_{P,n;\delta}]\cdot(P+\db\Gamma(i\nabla_y))-\tfrac12\big[\db\Gamma(\nabla_y)\cdot,[\db\Gamma(\nabla_y),iS_{P,n;\delta}]\big]\\
+[\db\Gamma(\omega(i\nabla_y)),iS_{P,n;\delta}].
\end{multline*}
By definition~\eqref{eq:defin-SPnelldelta}--\eqref{eq:defin-SPndelta} of $S_{P,n;\delta}$,
recalling~\eqref{eq:min-symbol} and using~\eqref{eq:prop-tilde-sell-2},
the commutators $[\db\Gamma(\nabla_y),iS_{P,n;\delta}]$ and $\big[\db\Gamma(\nabla_y)\cdot,[\db\Gamma(\nabla_y),iS_{P,n;\delta}]\big]$ are bounded by $O(|P|)$. We deduce for all~\mbox{$u,v\in\Cc^\fu$},
\begin{multline}\label{eq:H0P-commut-Nelson-2}
|\langle u,[H_0(P),iS_{P,n;\delta}]v\rangle|
\,\lesssim\,|P|\|u\|\|v\|
+|P|\|u\|\|H_0(P)^\frac12v\|\\
+|\langle u,[\db\Gamma(\omega(i\nabla_y)),iS_{P,n;\delta}]v\rangle|,
\end{multline}
and it remains to estimate the last commutator $[\db\Gamma(\omega(i\nabla_y)),iS_{P,n;\delta}]$, that is, on the $\ell$-boson state space,
\begin{equation}\label{eq:split-gamma-commutS}
[\db\Gamma(\omega(i\nabla_y)),iS_{P,n;\delta}]|_{\Gamma_s^{(\ell)}(\Hf)}\,=\,\sum_{j=1}^\ell\,[\omega(i\nabla_{y_j}),iS_{P,n,\ell;\delta}].
\end{equation}
We split this task into three steps: we first show that the commutator $[|\nabla_{y_j}|,iS_{P,n,\ell;\delta}]$ is nicely bounded and then we appeal to the calculus of almost-analytic extensions to reduce the analysis of the difference $\omega(i\nabla_{y_j})-|\nabla_{y_j}|$ to that of the resolvent $(z-|\nabla_{y_j}|)^{-1}$.

\medskip
\step1 Preliminary commutator estimates: for all $1\le j\le\ell$, we have
\begin{equation}\label{eq:commut-m0-S}
\|[|\nabla_{y_j}|,iS_{P,n,\ell;\delta}]\|\,\lesssim\,\tfrac1n|P|,
\end{equation}
and in addition, for all $u_\ell,v_\ell\in\Cc^\fu\cap\Gamma_s^{(\ell)}(\Hf)$ and~$z\in\C\setminus\R$,
\begin{eqnarray}
|\langle u_\ell,[(z-|\nabla_{y_j}|)^{-1},iS_{P,n,\ell;\delta}]v_\ell\rangle|&\lesssim&\tfrac1n|P||\Im z|^{-2}\|u_\ell\|\|v_\ell\|.\label{eq:commut-resolvent}
\end{eqnarray}
We start with the proof of~\eqref{eq:commut-m0-S}. By symmetry, we focus on $j=1$. Recalling the definition of~$S_{P,n,\ell;\delta}$, cf.~\eqref{eq:defin-SPnelldelta}, and using the integral representation for $|\nabla_{y_1}|=(-\triangle_{y_1})^{1/2}$, we find for all $v_\ell\in\Cc^\fu\cap\Gamma_s^{(\ell)}(\Hf)$,
\begin{multline*}
[|\nabla_{y_1}|,iS_{P,n,\ell;\delta}]v_\ell(y_1,\ldots,y_\ell)
\,=\,C|k_\star(n,P)|\\
\times\int_{\R^d}\tfrac{1}{|y_1-y_1'|^{d+1}}\Big(i\widetilde s_{n,\ell;\delta}\big(\tfrac{P}{|P|}\cdot y_1,\tfrac{P}{|P|}\cdot y_2,\ldots,\tfrac{P}{|P|}\cdot y_\ell\big)-i\widetilde s_{n,\ell;\delta}\big(\tfrac{P}{|P|}\cdot y_1',\tfrac{P}{|P|}\cdot y_2,\ldots,\tfrac{P}{|P|}\cdot y_\ell\big)\Big)\\
\times v_\ell(y_1',y_2,\ldots,y_\ell)\,dy_1'.
\end{multline*}
Using that $\widetilde s_{n,\ell;\delta}$ is Lipschitz continuous and
appealing to the $T(1)$ theorem~\cite{David-Journe-84}, we find that this commutator defines a bounded operator on~$\Ld^2((\R^{d})^{\ell})$ with
\begin{eqnarray*}
\|[|\nabla_{y_1}|,iS_{P,n,\ell;\delta}]\|
\,\lesssim\,|k_\star(n,P)|\|\nabla_1\widetilde s_{n,\ell;\delta}\|_{\Ld^\infty(\R^\ell)}.
\end{eqnarray*}
Recalling~\eqref{eq:min-symbol} and noting that $\|\nabla_1\widetilde s_{n,\ell;\delta}\|_{\Ld^\infty(\R^\ell)}\lesssim1$, the claim~\eqref{eq:commut-m0-S} follows.
We turn to the proof of~\eqref{eq:commut-resolvent}. For all $z\in\C\setminus\R$, we can write
\begin{equation*}
[(z-|\nabla_{y_j}|)^{-1},iS_{P,n,\ell;\delta}]
\,=\,-(z-|\nabla_{y_j}|)^{-1}[|\nabla_{y_j}|,iS_{P,n,\ell;\delta}](z-|\nabla_{y_j}|)^{-1}.
\end{equation*}
As $\|(z-|\nabla_{y_j}|)^{-1}\|\le|\Im z|^{-1}$, the claim~\eqref{eq:commut-resolvent} is a direct consequence of~\eqref{eq:commut-m0-S}.

\medskip
\step2 Conclusion.\\
We start by decomposing
\begin{multline}\label{eq:decomp-omega-abs}
\langle u_\ell,[\omega(i\nabla_{y_j}),iS_{P,n,\ell;\delta}]v_\ell\rangle\\
\,=\,\langle u_\ell,[|\nabla_{y_j}|,iS_{P,n,\ell;\delta}]v_\ell\rangle
+\langle u_\ell,[\omega(i\nabla_{y_j})-|\nabla_{y_j}|,iS_{P,n,\ell;\delta}]v_\ell\rangle.
\end{multline}
In order to estimate the second right-hand side term, we appeal to the calculus of almost-analytic extensions, e.g.~\cite[Proposition~C.2.2]{Derezinski-Gerard-97}: there exist $f\in C^\infty(\C)$ and constants $C,C_N<\infty$ such that
\begin{gather*}
f(t)=\chi(t)\big((m^2+t^2)^\frac12-t\big),\quad\forall t\in\R,\\
|\tfrac{\partial f}{\partial\bar z}(z)|\le C_N\langle \Re z\rangle^{-N-2}|\Im z|^N,\quad\forall N\in\N,\\
\supp f\subset\{z\in\C:|\Im z|\le C\langle\Re z\rangle\},
\end{gather*}
where $\chi\in C^\infty(\R)$ is a cut-off function such that $\chi(t)=1$ for $t\ge0$ and $\chi(t)=0$ for $t\le-1$.
We can then represent
\[f(t)\,=\,\tfrac{i}{2\pi}\int_\C(z-t)^{-1}\tfrac{\partial f}{\partial\bar z}(z)\,dz\wedge d\bar z,\]
hence
\[\omega(i\nabla_{y_j})-|\nabla_{y_j}|\,=\,\tfrac i{2\pi}\int_\C(z-|\nabla_{y_j}|)^{-1}\tfrac{\partial f}{\partial\bar z}(z)\,dz\wedge d\bar z.\]
Using this representation to rewrite the second right-hand side term in~\eqref{eq:decomp-omega-abs}, and using properties of $f$ as well as commutator estimates~\eqref{eq:commut-m0-S} and~\eqref{eq:commut-resolvent}, we get
\begin{eqnarray*}
|\langle u_\ell,[\omega(i\nabla_{y_j}),iS_{P,n,\ell;\delta}]v_\ell\rangle|
\,\lesssim\,\tfrac1n|P|\|u_\ell\|\|v_\ell\|.
\end{eqnarray*}
In view of~\eqref{eq:split-gamma-commutS}, as $mN\le H_0(P)$, this entails for all $u,v\in\Cc^\fu$,
\begin{equation*}
|\langle u,[\db\Gamma(\omega(i\nabla_y)),iS_{P,n;\delta}]v\rangle|\,\lesssim\,\tfrac1n|P|\|u\|\|H_0(P)v\|.
\end{equation*}
Combined with~\eqref{eq:H0P-commut-Nelson-1} and~\eqref{eq:H0P-commut-Nelson-2}, this implies that the first commutator $[H_0(P),iA_{P,n;\delta}]$ satisfies for all $u,v\in\Cc^\fu$,
\[|\langle u,[H_0(P),iA_{P,n;\delta}]v\rangle|\,\lesssim\,|P|\|u\|\|v\|+|P|\|u\|\|H_0(P)v\|,\]
hence it extends uniquely to the form of an $H_0(P)$-bounded self-adjoint operator.
Similarly computing iterated commutators and using~\eqref{eq:prop-tilde-sell-2}, the full conclusion easily follows; we skip the detail.
\end{proof}

Finally, we state that the fiber interaction Hamiltonian $\Phi(\rho)$ still has the same $C^\infty$-regularity with respect to $A_{P,n;\delta}$ as in~Lemma~\ref{lem:limit-reg-C2}(iii), thus establishing Theorem~\ref{th:main/Nelson}(iii).
The proof is a straightforward adaptation of that of Lemma~\ref{lem:com-dgamma-bounded}, now appealing to Lemma~\ref{lem:S-function} instead of Lemma~\ref{lem:m-function}; we skip the detail.

\begin{lem}
Let the interaction kernel $\rho$ belong to $H^\nu(\R^d)$ with $\langle k\rangle^\nu\nabla^\nu\rho\in\Ld^2(\R^d)$ for some $\nu\ge1$. Then, for all $0\le s\le\nu$, the $s$-th iterated commutator $\ad^s_{iA_{P,n;\delta}}(\Phi(\rho))$ extends as an $N^{1/2}$-bounded self-adjoint operator.
\end{lem}

\subsection{Consequences of Mourre estimate}
Given a total momentum $|P|>|P_\star|$,
letting~\mbox{$n_P\ge1$} be defined via~\eqref{eq:def-nP-statement-1}, we turn to the proof of Corollary~\ref{cor:main2}.
By items~(i) and~(iii) in Theorem~\ref{th:main/Nelson}, the sufficient criterion in Lemma~\ref{lem:suff-crit-smooth} ensures that the coupled fiber Hamiltonian $H_g(P)$ is of class $C^\infty(A_{P,n;\delta})$ for all~$n,g$.
Next, by Theorem~\ref{th:main/Nelson}(ii), for~$n=1$ and for any $n\ge n_P$, for all $\e>0$, Lemma~\ref{lem:pert-Mourre} allows to infer that $H_g(P)$ satisfies a Mourre estimate with respect to~$A_{P,n;\delta}$ on the energy interval
\[\Big(E_0^{(n)}(P)+\e+\tfrac{gC_{P,n}}{\e}\,,\,E_0^{(n+1)}(P)-\tfrac{gC_{P,n}}{\e}\Big),\]
for some constant $C_{P,n}$. Optimizing in $\e$, we deduce that $H_g(P)$ satisfies a Mourre estimate on
\[J_{P,n;g}\,:=\,\Big(E_0^{(n)}(P)+\sqrt gC_{P,n}\,,\,E_0^{(n+1)}(P)-gC_{P,n}\Big).\]
Moreover, the Mourre estimate is strict outside $K_{P,n;g}:=\big[\frac12|P|^2-gC_{P,n}\,,\,\frac12|P|^2+gC_{P,n}\big]$.
We may then appeal to Theorem~\ref{th:Mourre-spectrum}, which states that $H_g(P)$ has no singular spectrum and at most a finite number of eigenvalues in $J_{P,n;g}$, and has no eigenvalue in $J_{P,n;g}\setminus K_{P,n;g}$.
In order to exclude the existence of eigenvalues in $K_{P,n;g}$, we appeal to Theorem~\ref{th:instab}, which states the instability of the uncoupled eigenvalue~$\frac12P^2$ provided that Fermi's condition~\eqref{eq:Fermi-cond0} holds.
Altogether, this proves item~(i) of Corollary~\ref{cor:main2}, and item~(ii) follows by further applying Theorem~\ref{th:CGH}.
It remains to make Fermi's condition~\eqref{eq:Fermi-cond0} more explicit for the model at hand, which we is the purpose of the following lemma; the proof is analogous to that of Lemma~\ref{Fermi-fric} and is skipped for brevity.

\begin{lem}\label{Fermi}
For all $|P|>|P_\star|$, we have
\begin{multline*}
\lim_{\e\downarrow0}\Big\langle\Omega\,,\,\Phi( \rho )\bar\Pi_\Omega \big( H_0(P) - \tfrac12P^2 - i\e\big)^{-1}\bar\Pi_\Omega \Phi( \rho )\Omega\Big\rangle\\
\,=\,(2\pi)^{-d}\,\pv\int_{E_0^{(1)}(P)}^\infty( t- \tfrac12P^2)^{-1}\bigg(\int_{\{k\,:\,\frac12(P-k)^2+\omega(k)=t\}}\tfrac{|\rho(k)|^2}{|k-P+\nabla\omega(k)|}\db\mathcal H_{d-1}(k)\bigg)\db t\\
+\tfrac i2(2\pi)^{1-d}\int_{\{k\,:\,\frac12(P-k)^2+\omega(k)=\frac12P^2\}}\tfrac{|\rho(k)|^2}{|k-P+\nabla\omega(k)|}\db\mathcal H_{d-1}(k),
\end{multline*}
where $\Hc_{d-1}$ stands for the $(d-1)$th-dimensional Hausdorff measure.
In particular, the imaginary part is positive if $\rho$ is nowhere vanishing.
\end{lem}

\appendix
\section{Mourre's commutator method}\label{app:Mourre}
In this appendix, we briefly recall for convenience standard definitions and statements from Mourre's theory that we use in this work; we refer e.g.\@ to~\cite{ABMG-96,GGM04a} for more detail.
We start with the notion of regularity with respect to a self-adjoint operator, which is crucial to define commutators and deal with domain issues.

\begin{defin}[Regularity]\label{def:Ck}
Let $A$ be a self-adjoint operator on a Hilbert space $\Hc$.
\begin{enumerate}[---]
\item A bounded operator $B$ on $\Hc$ is said to be of class $C^k(A)$ if for all $\phi\in\Hc$ the function $t\mapsto e^{-itA} B e^{itA}\phi$ is $k$-times continuously differentiable.
\smallskip\item A self-adjoint operator $H$ on $\Hc$ is said to be of class $C^k(A)$ if its resolvent $(H-z)^{-1}$ is of class $C^k(A)$ for some $z\in\C\setminus\R$.\qedhere
\end{enumerate}
\end{defin}

We recall the following characterization: a bounded operator $B$ is of class $C^1(A)$ if and only if it maps $ \Dc(A)$ into itself and if the commutator $\ad_{iA}(B):=[B,iA]$ extends uniquely from $\Dc(A)$ to a bounded operator on $\Hc$.
Therefore, if $H$ is a self-adjoint operator of class~$C^1(A)$, we may use the resolvent identity $[(H-z)^{-1},iA]=-(H-z)^{-1}[H,iA](H-z)^{-1}$ in the sense of forms on $\Dc(A)$, and we infer that the commutator $\ad_{iA}(H):=[H,iA]$ extends uniquely from $\Dc(H)\cap\Dc(A)$ to a bounded form on $\Dc(H)$. Equivalently, this means for all $\phi,\psi\in\Dc(H)\cap\Dc(A)$,
\begin{equation}\label{equiv:C1A}
|\langle\phi,[H,iA]\psi\rangle_\Hc|\,\lesssim\,\|(|H|+1)\phi\|_\Hc\|(|H|+1)\psi\|_\Hc.
\end{equation}
In fact, we state that the converse is also true under a technical assumption; see e.g.~\cite[Theorem~6.3.4]{ABMG-96}.
\begin{lem}[Characterization of regularity; \cite{ABMG-96}]
Let $A$ and $H$ be self-adjoint operators on a Hilbert space $\Hc$, and assume that the unitary group generated by $A$ leaves the domain of $H$ invariant,
\begin{equation}\label{eq:invar-HA}
e^{itA}\Dc(H)\subset\Dc(H)\quad\text{for all $t\in\R$}.
\end{equation}
Then, the domain $\Dc(H)\cap\Dc(A)$ is a core for $H$. In addition, $H$ is of class $C^1(A)$ if and only if~\eqref{equiv:C1A} holds.
\end{lem}

We could write down similar characterizations for higher regularity, but we shall only need the following sufficient criterion in case of $H$-bounded commutators. Note that this $H$-boundedness condition is much stronger than~\eqref{equiv:C1A} and is not always satisfied; see in particular our setting in Section~\ref{sec:conjug/Nelson}.

\begin{lem}[Sufficient criterion for higher regularity; \cite{ABMG-96}]\label{lem:suff-crit-smooth}
Let $A$ and $H$ be self-adjoint operators on a Hilbert space $\Hc$, and assume that the unitary group generated by $A$ leaves the domain of $H$ invariant, cf.~\eqref{eq:invar-HA}.
Given $\nu\ge1$, assume iteratively for all $0\le s\le\nu$, starting with $\ad_{iA}^0(H):=H$, that the iterated commutator $\ad_{iA}^s(H)$ is defined as a form on $\Dc(H)\cap\Dc(A)$ and satisfies
\begin{equation}\label{eq:ad-iterate-bnd}
\||\langle\phi,\ad_{iA}^s(H)\psi\rangle|\,\lesssim\,\|\phi\|_\Hc\|(|H|+1)\psi\|_\Hc,
\end{equation}
which entails that $\ad_{iA}^s(H)$ extends uniquely to the form of an $H$-bounded operator and that the next commutator $\ad_{iA}^{s+1}(H):=[\ad_{iA}^{s}(H),iA]$ is also well-defined as a form on $\Dc(H)\cap\Dc(A)$. Then, $H$ is of class $C^\nu(A)$.
\end{lem}

With these regularity assumptions at hand, we may now turn to Mourre commutator estimates, which constitute a key tool for spectral analysis.

\begin{defin}[Mourre estimates]
Let $A$ be a self-adjoint operator on a Hilbert space~$\Hc$, let $H$ be a self-adjoint operator of class $C^1(A)$, and let $J\subset\R$ be a bounded open interval. The operator $H$ is said to satisfy a {\it Mourre estimate on $J$} with respect to the {\it conjugate operator} $A$ if there exists a constant $c_0>0$ and a compact operator $K$ such there holds in the sense of forms,
\[\mathds1_J(H)[H,iA]\mathds1_J(H)\,\ge\, c_0\mathds1_J(H)+K.\]
The Mourre estimate is said to be {\it strict} if it holds with $K=0$, and the constant $c_0$ is referred to as the Mourre constant.
\end{defin}

The main motivation for these commutator estimates is that they lead to precise information on the nature of the spectrum of $H$; see~\cite{Mourre-80,ABMG-96}.

\begin{theor}[Mourre's theory; \cite{Mourre-80,ABMG-96}]\label{th:Mourre-spectrum}
Let $A$ be a self-adjoint operator on a Hilbert space $\Hc$, let $H$ be a self-adjoint operator of class $C^1(A)$, and assume that $H$ satisfies a Mourre estimate with respect to $A$ on a bounded open interval $J\subset\R$. Then the following properties hold:
\begin{enumerate}[---]
\item $H$ has at most a finite number of eigenvalues in $J$ (counting multiplicities);
\item if $H$ is of class $C^2(A)$, then $H$ has no singular continuous spectrum in $J$;
\item if the Mourre estimate is strict, then $H$ has no eigenvalue in $J$.
\qedhere
\end{enumerate}
\end{theor}

Next, we adapt these developments to the setting of perturbation theory. First, the following standard lemma states that, if $H$ satisfies a Mourre estimate and if a perturbation~$V$ is sufficiently regular, then the perturbed operators $H_g:=H+gV$ also satisfy a corresponding Mourre estimate for $g$ small enough. In view of Section~\ref{sec:conjug/Nelson}, care is taken not to assume that $[H,iA]$ be $H$-bounded; the outline of the proof is included for convenience.

\begin{lem}[Mourre estimates under perturbations]\label{lem:pert-Mourre}
Let $A$ be a self-adjoint operator on a Hilbert space $\Hc$,
let $H$ be a self-adjoint operator of class $C^1(A)$,
let $V$ be a symmetric~$|H|^{1/2}$-bounded operator,
and assume that:
\begin{enumerate}[---]
\item the commutator $[H,iA]$ satisfies the following strengthened version of~\eqref{equiv:C1A},
\begin{equation}\label{eq:C1A-bnd}
|\langle\phi,[H,iA]\psi\rangle_\Hc|\,\lesssim\,\|(|H|+1)^\frac12\phi\|_\Hc\|(|H|+1)\psi\|_\Hc;
\end{equation}
\item the commutator $[V,iA]$ extends as an $H$-bounded operator, in the sense that
\begin{equation}\label{eq:C1A-bnd-2}
|\langle\phi,[V,iA]\psi\rangle_\Hc|\,\lesssim\,\|\phi\|_\Hc\|(|H|+1)\psi\|_\Hc.
\end{equation}
\end{enumerate}
Then the following properties hold.
\begin{enumerate}[(i)]
\item The perturbed operator $H_g=H+gV$ is self-adjoint on~$\Dc(H)$ and is of class $C^1(A)$ for all $g\in\R$.
\smallskip\item Further assume that $H$ satisfies a Mourre estimate with respect to $A$ on a bounded interval $(a,b)$, with constant $c_0$.
Then $H_g$ satisfies a Mourre estimate with respect to~$A$ on the restricted interval
\[(a+\eta,b-\eta),\qquad \eta:=\tfrac{gC}{c_0} (1+|a|+|b|)^{\frac32},\]
for some constant $C$ only depending on the multiplicative constants in~\eqref{eq:C1A-bnd}--\eqref{eq:C1A-bnd-2}.
If in addition~$[H,iA]$ is $H$-bounded, in the sense that $(|H|+1)^{1/2}\phi$ can be replaced by~$\phi$ in the right-hand side of~\eqref{eq:C1A-bnd}, then the same holds with $\eta=\frac{gC}{c_0}(1+|a|+|b|)$.
Finally, if the Mourre estimate for~$H$ is strict, then the one for $H_g$ is strict too.
\smallskip\item Further assume that $H$ is of class $C^2(A)$ and that
$[[V,iA],iA]$ extends as an $H$-bounded operator.
Then, $H_g$ is of class $C^2(A)$ for all $g\in\R$.\qedhere
\end{enumerate}
\end{lem}

\begin{proof}
As the perturbation $V$ is $|H|^\frac12$-bounded, the perturbed operator $H_g=H+gV$ is self-adjoint and has the same domain as~$H$ for all $g\in\R$. The proof of items~(i) and~(iii) is standard, following the same lines as e.g.~\cite[proof of Proposition~2.5]{MR13}, starting from identities
\begin{eqnarray*}
(H_g-z)^{-1}&=&(H-z)^{-1}(1+gV(H-z)^{-1})^{-1},\\
{}[(H_g-z)^{-1},iA]&=&[(H-z)^{-1},iA](1+gV(H-z)^{-1})^{-1}\\
&&-g(H_g-z)^{-1}V[(H-z)^{-1},iA](1+gV(H-z)^{-1})^{-1}\\
&&-g(H_g-z)^{-1}[V,iA](H_g-z)^{-1},
\end{eqnarray*}
where $\Im z$ is chosen large enough so that $\|gV(H-z)^{-1}\|<1$.
We skip the detail and turn to item~(ii).
Assume that $H$ satisfies a Mourre estimate with respect to $H$ on a bounded interval $J=(a,b)$. Let $\eta\in(0,1)$, let $J_\eta:=(a+\eta,b-\eta)$, and choose $h_\eta\in C_c^\infty(\R)$ such that $\mathds1_{J_\eta}\le h_\eta\le \mathds1_J$ and $|\nabla h_\eta|\lesssim \frac1\eta$. Multiplying both sides of the Mourre estimate for~$H$ with $h_\eta(H)$, we get for some compact operator $K$,
\[h_\eta(H)[H,iA]h_\eta(H)\,\ge\,c_0h_\eta(H)+h_\eta(H)Kh_\eta(H),\]
hence,
as $[V,iA]$ is $H$-bounded,
\begin{equation}\label{eq:pre-commut-Hg}
h_\eta(H)[H_g,iA]h_\eta(H)\,\ge\,\Big(c_0-gC(1+|a|+|b|)\Big)h_\eta(H)+h_\eta(H)Kh_\eta(H).
\end{equation}
Next, we decompose
\begin{multline}\label{eq:decomp-Hg-commut}
h_\eta(H_g)[H_g,iA]h_\eta(H_g)\,=\,h_\eta(H)[H_g,iA]h_\eta(H)\\
+(h_\eta(H_g)-h_\eta(H))[H_g,iA]h_\eta(H_g)+h_\eta(H)[H_g,iA](h_\eta(H_g)-h_\eta(H)).
\end{multline}
Recalling~\eqref{eq:C1A-bnd}, the $|H|^\frac12$-boundedness of $V$, and the $H$-boundedness of $[V,iA]$, and noting that $\|h_\eta(H_g)-h_\eta(H)\|\lesssim \frac1\eta g$,
we easily find that the last two right-hand side terms in~\eqref{eq:decomp-Hg-commut} have operator norm bounded by~$\frac{gC}\eta (1+|a|+|b|)^{3/2}$.
Combined with~\eqref{eq:pre-commut-Hg}, this yields
\begin{equation*}
h_\eta(H_g)[H_g,iA]h_\eta(H_g)\,\ge\,\Big(c_0-\tfrac{gC}\eta (1+|a|+|b|)^{\frac32}\Big)h_\eta(H_g)+h_\eta(H)Kh_\eta(H).
\end{equation*}
Now multiplying both sides with $\mathds1_{J_\eta}(H_g)$, the conclusion~(ii) follows.
\end{proof}

An important question concerns the perturbation of an eigenvalue embedded in continuous spectrum~\cite{Simon-73}. In view of formal second-order perturbation theory, Fermi's golden rule is expected to provide an instability criterion, cf.~\eqref{eq:Fermi-cond0} below, and various works have shown how Mourre's theory can be used to establish it rigorously, e.g.~\cite{AHS-89,Hunziker-Sigal-00,Faupin-Moller-Skibsted-11}. Revisiting~\cite[Theorem~8.8]{Hunziker-Sigal-00}, we can derive for instance the following statement, where care is taken again not to assume that $[H,iA]$ is $H$-bounded; the outline of the proof is included for convenience.

\begin{theor}[Instability of embedded bound states]\label{th:instab}
Let $A$ be a self-adjoint operator on a Hilbert space $\Hc$, let $H$ be a self-adjoint operator of class $C^2(A)$,
let $V$ be a symmetric~$|H|^{1/2}$-bounded operator,
and assume that:
\begin{enumerate}[---]
\item the commutator $[H,iA]$ satisfies~\eqref{eq:C1A-bnd};
\item the commutators $[V,iA]$ and $[[V,iA],iA]$ extend as $H$-bounded operators;
\item $H$ satisfies a Mourre estimate with respect to $A$ on a bounded open interval $J\subset\R$.
\end{enumerate}
In addition, assume that $H$ has an eigenvalue $E_0\in J$, denote by $\Pi_0$ the associated eigenprojector, let $\bar\Pi_0:=1-\Pi_0$, assume that the eigenspace satisfies $\Ran(\Pi_0)\subset\Dc(A^2)$ and $\Ran(A\Pi_0)\subset\Dc(V)$, and assume that Fermi's condition holds, that is, there exists $\gamma_0>0$ such that
\begin{equation}\label{eq:Fermi-cond0}
\lim_{\e\downarrow0}\,\Im\Big\{\Pi_0V\bar\Pi_0(H-E_0-i\e)^{-1}\bar\Pi_0V\Pi_0\Big\}\,\ge\,\gamma_0\Pi_0.
\end{equation}
Then, there exists $g_0>0$ and a neighborhood $J_0\subset J$ of $E_0$ such that the perturbed operator $H_g=H+gV$ satisfies
\[\sigma_\pp(H_g)\cap J_0\,=\,\varnothing\qquad\text{for all $0<|g|\le g_0$}.\qedhere\]
\end{theor}

\begin{proof}
Note that all assumptions of Lemma~\ref{lem:pert-Mourre} are satisfied, hence the perturbed operator~$H_g$ is of class $C^2(A)$ and satisfies a Mourre estimate on $J'$ with respect to $A$ for all $J'\Subset J$ and $g$ small enough.
Consider the reduced perturbed operator $\bar H_g:=\bar\Pi_0 H_g\bar\Pi_0$ on the range $\Ran(\bar\Pi_0)$,
and set also $\bar H:=\bar\Pi_0H\bar\Pi_0$, $\bar V:=\bar\Pi_0V\bar\Pi_0$, $\bar A:=\bar\Pi_0A\bar\Pi_0$.
We follow the approach in~\cite[Theorem~8.8]{Hunziker-Sigal-00} and split the proof into three steps.

\medskip
\step1 Proof that $\bar H_g$ is of class $C^2(\bar A)$ for all $g$ and that there exists $g_0>0$ and an open interval $J_0\subset J$ with $E_0\in J_0$ such that for all $|g|\le g_0$ the operator $\bar H_g$ satisfies a strict Mourre estimate on $J_0$ with respect to $\bar A$. In particular, in view of Theorem~\ref{th:Mourre-spectrum}(iii), this entails that $\bar H_g$ has no eigenvalue in $J_0$ for any $|g|\le g_0$.

\medskip\noindent
It is easily checked that reduced operators $\bar A,\bar H,\bar V$ satisfy all the assumptions of Lemma~\ref{lem:pert-Mourre} on $\Ran(\bar\Pi_0)$.
In particular, in order to ensure that $[\bar V,i\bar A]$ and $[[\bar V,i\bar A],i\bar A]$ are $\bar H$-bounded, it suffices to decompose
\begin{eqnarray*}
{}[\bar V,i\bar A]&=&\bar\Pi_0[V,iA]\bar\Pi_0-\bar\Pi_0V\Pi_0iA\bar\Pi_0+\bar\Pi_0iA\Pi_0V\bar\Pi_0,\\
{}[[\bar V,i\bar A],i\bar A]&=&\bar\Pi_0[[V,iA], iA]\bar\Pi_0
+\bar\Pi_0ViA\Pi_0 iA\bar\Pi_0+\bar\Pi_0 iA\Pi_0iAV\bar\Pi_0\\
&&\hspace{-1.5cm}-\bar\Pi_0V\Pi_0(iA)^2\bar\Pi_0-\bar\Pi_0(i A)^2\Pi_0V\bar\Pi_0+\bar\Pi_0V\Pi_0iA\Pi_0 iA\bar\Pi_0+\bar\Pi_0i A\Pi_0iA\Pi_0V\bar\Pi_0\\
&&+2\bar\Pi_0iA\Pi_0[V,iA]\bar\Pi_0-2\bar\Pi_0[V,i A]\Pi_0iA\bar\Pi_0-2\bar\Pi_0i A\Pi_0V\Pi_0iA\bar\Pi_0,
\end{eqnarray*}
and to note that our assumptions precisely ensure that the different right-hand side terms are all $\bar H$-bounded.
Applying Lemma~\ref{lem:pert-Mourre},
we then deduce that $\bar H_g$ is of class $C^2(\bar A)$ for all~$g$ and satisfies a Mourre estimate on $J'$ with respect to $\bar A$ for all $J'\Subset J$ and $g$ small enough.
Next, multiplying both sides of this estimate with $\mathds1_L(\bar H)$ and using the fact that $\mathds1_L(\bar H)$ converges strongly to $0$ as $L\to\{E_0\}$, we deduce that there is a neighborhood $L_0$ of $E_0$ on which $\bar H$ satisfies a strict Mourre estimate.
The claimed strict Mourre estimate for $\bar H_g$ then follows from Lemma~\ref{lem:pert-Mourre}(ii) for any $J_0\Subset L_0$ and $g$ small enough.

\medskip
\step2 Proof that, if for some $|g|\le g_0$ the perturbed operator $H_g$ has an eigenvalue $E\in J_0$ with eigenvector $\psi$, then it satisfies
\begin{equation}\label{eq:eigenvalue-Fermi}
\lim_{\e\downarrow0}~\Im\Big\langle \psi\,,\,\Pi_0 W\bar\Pi_0(\bar H_g-E-i\e)^{-1}\bar\Pi_0 W\Pi_0\psi\Big\rangle_\Hc\,=\,0.
\end{equation}
This observation is found e.g.~in~\cite[Lemma~8.10]{Hunziker-Sigal-00}, but we repeat the proof for convenience.
Decomposing $1=\Pi_0+\bar\Pi_0$ and using $\Pi_0H\Pi_0=E_0\Pi_0$ and $\Pi_0H\bar\Pi_0=0$, the eigenvalue equation $H_g\psi=E\psi$ is equivalent to the system
\begin{equation}\label{eq:decomp-eigenvalue-eqn}
\left\{\begin{array}{l}
g\Pi_0 W\Pi_0\psi+g\Pi_0 W\bar\Pi_0\psi=(E-E_0)\Pi_0\psi,\\
\bar H_g\bar\Pi_0\psi+g\bar\Pi_0 W\Pi_0\psi=E\bar\Pi_0\psi.
\end{array}\right.
\end{equation}
For all $\e>0$, the second equation entails
\[\bar\Pi_0\psi=-g(\bar H_g-E-i\e)^{-1}\bar\Pi_0 W\Pi_0\psi-i\e(\bar H_g-E-i\e)^{-1}\bar\Pi_0\psi.\]
By Step~1, we know that $E\in J_0$ cannot be an eigenvalue of $\bar H_g$, hence the last right-hand side term converges strongly to $0$ as $\e\downarrow0$ and we get
\[\bar\Pi_0\psi=-g\lim_{\e\downarrow0}(\bar H_g-E-i\e)^{-1}\bar\Pi_0 W\Pi_0\psi.\]
Inserting this into the first equation of~\eqref{eq:decomp-eigenvalue-eqn}, taking the scalar product with $\psi$, and taking the imaginary part of both sides, the claim~\eqref{eq:eigenvalue-Fermi} follows.

\medskip
\step3 Conclusion.\\
In view of Step~1, as $\bar H_g$ is of class $C^2(\bar A)$ and satisfies a strict Mourre estimate on $J_0$ for all $|g|\le g_0$, Mourre's theory entails the validity of the following strong limiting absorption principle, cf.~\cite{ABMG-96,Sahbani-97}: for all $s>\frac12$ and $J_0'\Subset J_0$, the limit $\lim_{\e\downarrow0}\langle\bar A\rangle^{-s}(\bar H_g-E-i\e)^{-1}\langle\bar A\rangle^{-s}$ exists in the weak operator topology, uniformly for $E\in J_0'$ and $|g|\le g_0$. (Note that we could not find a reference for the uniformity with respect to $g$, but it is easily checked to follow from~\cite{ABMG-96,Sahbani-97} by further making use of the $H$-boundedness of $V$ and $[V,iA]$.)
Decomposing $iA\bar\Pi_0V\Pi_0=ViA\Pi_0-[V,iA]\Pi_0-iA\Pi_0V\Pi_0$ and noting that our assumptions ensure that the different right-hand side terms are all bounded, we find that $\langle A\rangle\bar\Pi_0W\Pi_0$ is bounded (and finite-rank), hence the limiting absorption principle entails that the limit
\begin{eqnarray*}
F_g(E)&:=&\lim_{\e\downarrow0}~\Pi_0V\bar\Pi_0(\bar H_g-E-i\e)^{-1}\bar\Pi_0V\Pi_0\\
&=&\lim_{\e\downarrow0}~\big(\Pi_0V\bar\Pi_0\langle\bar A\rangle\big)\Big(\langle\bar A\rangle^{-1}(\bar H_g-E-i\e)^{-1}\langle\bar A\rangle^{-1}\Big)\big(\langle\bar A\rangle\bar\Pi_0V\Pi_0\big)
\end{eqnarray*}
exists, uniformly for $E\in J_0'$ and $|g|\le g_0$. This ensures in particular that the limit in~\eqref{eq:Fermi-cond0} exists. Assumption~\eqref{eq:Fermi-cond0} takes the form $\Im F_0(E_0)\ge\gamma_0\Pi_0$, and therefore by uniformity there exists~$g_0'>0$ and a neighborhood $J''_0$ of $E_0$ such that
\[\Im F_g(E)\ge\tfrac12\gamma_0\Pi_0\qquad\text{for all $E\in J_0''$ and $|g|\le g_0'$.}\]
In view of Step~2, this implies that for $|g|\le g_0'$ any eigenvalue of $H_g$ in~$J_0''$ must have eigenvector in $\Ran(\bar\Pi_0)$. However, this would entail that it is actually an eigenvalue of the reduced operator $\bar H_g$, which is excluded by Step~1.
\end{proof}

Moreover, if there is enough analyticity for the analytic continuation of the resolvent,
the perturbed embedded eigenvalue is actually expected to become a complex resonance when dissolving in the absolutely continuous spectrum~\cite{Simon-73}. This resonance then describes the metastability of the bound state and the quasi-exponential decay of the system away from this state.
While this is not guaranteed in the general framework of Mourre's theory, the following result by Cattaneo, Graf, and Hunziker~\cite{CGH-06} shows how additional regularity allows to develop an approximate dynamical resonance theory.
We emphasize that $C^2$-regularity is no longer enough here.

\begin{theor}[Approximate dynamical resonances; \cite{CGH-06}]\label{th:CGH}
Let $A$ and $H$ be self-adjoint operators on a Hilbert space $\Hc$,
let $V$ be symmetric and $|H|^{1/2}$-bounded, and assume that for some $\nu\ge0$,
\begin{enumerate}[---]
\item the unitary group generated by $A$ leaves the domain of $H$ invariant, cf.~\eqref{eq:invar-HA};
\item for all $0\le j\le 5+\nu$, the iterated commutators $\ad^j_{iA}(H)$ and $\ad^j_{iA}(V)$ extend as $H$-bounded operators;
\item $H$ satisfies a Mourre estimate with respect to $A$ on a bounded open interval $J\subset\R$.
\end{enumerate}
In addition, assume that $H$ has a simple eigenvalue $E_0\in J$ with normalized eigenvector~$\psi_0$, denote by~$\bar\Pi_0$ the orthogonal projection on $\{\psi_0\}^\bot$, and assume that Fermi's condition is satisfied, that is,
\begin{equation}\label{eq:Fermi-cond0-re}
\gamma_0\,:=\,\lim_{\e\downarrow0}\Im\Big\langle \bar\Pi_0(V\psi_0)\,,\,(H-E_0-i\e)^{-1}\bar\Pi_0(V\psi_0)\Big\rangle\,>\,0.
\end{equation}
Then, the perturbed operator $H_g=H+gV$ satisfies the following quasi-exponential decay law: for all smooth cut-off functions $h$ supported in $J$ and equal to $1$ in a neighborhood of~$E_0$, and for all $g$ small enough, there holds for all~$t\ge0$,
\[\Big|\big\langle\psi_0,e^{-iH_gt}h(H_g)\psi_0\big\rangle-e^{-iz_gt}\Big|\,\lesssim_{h,\gamma_0}\,\left\{\begin{array}{ll}
g^2|\!\log g|\langle t\rangle^{-\nu},&\text{if $\nu\ge0$;}\\
g^2\langle t\rangle^{-(\nu-1)},&\text{if $\nu\ge1$;}
\end{array}\right.\]
where the dynamical resonance $z_g$ is given by Fermi's golden rule,
\[z_g\,=\,E_0+g\langle\psi_0,V\psi_0\rangle-g^2\lim_{\e\downarrow0}\Big\langle\bar\Pi_0(V\psi_0)\,,\,\big(H-E_0-i\e\big)^{-1}\bar\Pi_0(V\psi_0)\Big\rangle.\]
In particular, in view of~\eqref{eq:Fermi-cond0-re}, this satisfies $\Im z_g<0$.
\end{theor}

\section*{Acknowledgements}
The authors thank Jérémy Faupin and Sylvain Golenia for motivating discussions at different stages of this work. MD acknowledges financial support from F.R.S.-FNRS.

\bibliographystyle{plain}
\bibliography{biblio}

\end{document}